\documentclass[ejsv2, preprint, noshowframe]{imsart}

\RequirePackage[authoryear]{natbib}
\RequirePackage{savesym}   
\usepackage{amsmath}

\savesymbol{Bbbk}          
\RequirePackage{amssymb}   
\restoresymbol{DSF}{Bbbk}  

\usepackage[ruled,vlined,linesnumbered]{algorithm2e} 
\usepackage{verbatim,float,url,dsfont}
\usepackage{graphicx,subcaption,psfrag}
\usepackage{mathtools,enumitem}
\usepackage{multirow}
\usepackage{ragged2e}
\usepackage{array}
\usepackage{enumitem}
\usepackage[colorlinks=true,citecolor=blue,urlcolor=blue,linkcolor=blue]{hyperref}
\usepackage[utf8]{inputenc} 
\usepackage[T1]{fontenc}    
\usepackage{booktabs}       
\usepackage{nicefrac}         
\usepackage{microtype}      

\RequirePackage[colorlinks,citecolor=blue,urlcolor=blue]{hyperref}
\RequirePackage{graphicx}
\usepackage{enumitem} 
\allowdisplaybreaks
\startlocaldefs
\theoremstyle{plain}

\newtheorem{theorem}{Theorem}
\newtheorem{lemma}{Lemma}
\newtheorem{proposition}{Proposition}
\newtheorem{remark}{Remark}
\newtheorem{example}{Example}
\newtheorem{corollary}{Corollary}
\theoremstyle{definition}

\theoremstyle{remark}

\endlocaldefs
\DeclareMathOperator*{\argmin}{argmin}

\def\E{\mathds{E}}
\def\P{\mathds{P}}
\def\V{\mathds{V}}

\def\bx{\mathbf{x}}

\def\bz{\mathbf{z}}

\def\E{\mathbb{E}}
\def\P{\mathbb{P}}
\def\V{\mathbb{V}}
\def\R{\mathbb{R}}

\def\N{\mathbb{N}}

\def\indep{\perp\!\!\!\perp}

\def\cA{\mathcal{A}}
\def\cB{\mathcal{B}}
\def\cC{\mathcal{C}}
\def\cD{\mathcal{D}}

\def\cL{\mathcal{L}}

\def\cN{\mathcal{N}}
\def\cP{\mathcal{P}}

\def\cW{\mathcal{W}}

\def\bbeta{\boldsymbol{\beta}}
\def\bX{\boldsymbol{X}}

\begin{document}
\begin{frontmatter}
\title{Variable Selection for Linear Regression Imputation in Surveys}
\runtitle{Variable Selection for Linear Regression Imputation in Surveys}

\begin{aug}
\author[A]{\fnms{Ziming}~\snm{An}\ead[label=e1]{zan032@uottawa.ca}},
\author[B]{\fnms{Mehdi}~\snm{Dagdoug}\ead[label=e2]{mehdi.dagdoug@mcgill.ca}}
\and
\author[A]{\fnms{David}~\snm{Haziza}\ead[label=e3]{dhaziza@uottawa.ca}}
\address[A]{Department of Mathematics and   Statistics,
University of Ottawa, Ottawa, Canada \\ \printead[presep={\ }]{e1,e3}}

\address[B]{Department of Mathematics and Statistics, McGill University, Montréal, Canada\\ \printead[presep={\ }]{e2}}
\runauthor{Z. An et al.}
\end{aug}

\begin{abstract}
Survey sampling is concerned with the estimation of finite population parameters. In practice, survey data suffer from item nonresponse, which is commonly handled through imputation, i.e., replacing missing values with predicted values. As a result, the properties of the resulting imputed estimator depend critically on the properties of the prediction method used. In turn, prediction methods themselves  depend on the choice of variables and tuning parameters used  to fit the imputation model. In this article, we study the problem of variable selection for linear regression imputation. Although variable selection has been widely studied across many fields, primarily for identification or prediction, its role in imputation for survey data has received comparatively little attention. We introduce the notion of an optimal imputation model defined through an oracle loss function and show that, with probability tending to one, the optimal model coincides with the true model. We also examine the consequences of using misspecified models--either omitting relevant covariates or including irrelevant ones-- on consistency  and asymptotic variance.
We then develop a complete methodological framework for constructing confidence intervals after model selection. The proposed confidence intervals are shown to be asymptotically valid and optimal among all candidate models. Simulation studies indicate that the proposed methodology performs well in finite samples. 
\end{abstract}



\end{frontmatter}

\section{Introduction}

Nonresponse is a major challenge in official statistics because it affects most surveys. If missing data are ignored, the resulting estimates can be biased and inconsistent. In survey sampling, it is customary to distinguish unit from item nonresponse. The former corresponds to cases where no information is collected for a sampled unit, whereas the latter occurs when some variables are missing but others are observed. Item nonresponse is most often addressed through imputation, a procedure that replaces missing values with predicted values. Imputation can restore consistency and support valid inference, provided that the imputation model is correctly specified and that variance estimation properly reflects both sampling variability and nonresponse. In practice, this means that the imputation model matters a lot: if it is misspecified, the resulting estimates may be biased and/or inefficient.

~\\
The properties of imputed estimators have been investigated under a wide range of imputation models. A large body of work focuses on parametric imputation, such as linear regression imputation (e.g., \cite{chauvet2011balanced}), including many contributions on variance estimation \citep{fay1991design, shao1999variance, berger2006adjusted, kim2009unified}. Other contributions have examined nonparametric regression and related methods, including nearest neighbor \citep{chen2000nearest, chen2001jackknife, yang2019nearest}, the score method \citep{haziza2007construction}, predictive mean matching \citep{yang2017predictive}, and random forests \citep{dagdoug2025statistical}.\\

\noindent In the context of model-assisted estimation, we highlight the work of \cite{opsomer2005selecting}, which provides an important reference on hyperparameter selection in survey sampling, although not in the imputation setting.
In the customary i.i.d. setting, variable selection has been extensively studied with two primary objectives: identifying the best predictive model for a given task (efficiency) or recovering the true set of coefficients (consistency) \citep{yang2005can}. Seminal contributions include \cite{nishii1984asymptotic}, \cite{rao1989strongly},  \cite{shao1993linear}, and \cite{shao1997asymptotic}, who established the asymptotic properties of numerous model selection criteria such as AIC \citep{akaike1970statistical}, BIC \citep{schwarz1978estimating}, and cross-validation. 

~\\
In finite population sampling, the goal is to estimate finite population parameters, and the main concern is the performance of the resulting estimators, in particular their bias and efficiency. Therefore, when choosing an imputation model, the survey statistician is not primarily aiming to identify the true data-generating model, but rather to select the model that produces the most efficient imputed estimator, that is, the one with the smallest mean squared error. Although we will borrow several ideas from the i.i.d. literature, our objective is fundamentally different: we aim to optimize imputation efficiency for finite population inference, not to recover the underlying regression function. As shown later, however, these two goals are closely related in the case of linear models.

~\\
The main contributions of this paper are as follows. We introduce an oracle loss function to assess the efficiency of a candidate imputation model under survey sampling. We show that, under mild conditions, the model that minimizes this loss is asymptotically the true model, which connects model selection for imputation with model identification. We then study the effect of using a misspecified model. In particular, we give general conditions under which the resulting imputed estimators remain consistent, and we identify settings where overfitting (i.e., adding unnecessary covariates) does or does not increase the asymptotic variance. Under standard regularity conditions, we show that model selection criteria that are consistent in the i.i.d. framework remain consistent under survey sampling with missing data. Finally, when a consistent selection criterion is used, the resulting imputed estimator and its variance are asymptotically equivalent to those obtained under the true model, yielding oracle efficiency. We also show that a standard variance estimator for linear regression imputation is consistent and establish the asymptotic normality of the resulting point estimator. We propose a complete methodology for constructing confidence intervals after model selection, and show that it achieves the nominal coverage asymptotically, with asymptotically minimal width within the class of candidate models. Simulation studies under both equal and unequal probability sampling designs confirm the good performance of the proposed approach.

~\\
The remainder of the paper is organized as follows. Section~\ref{sec:setup} introduces the notation and general framework. Section~\ref{sec:loss} defines the oracle loss function and its properties. Section~\ref{Section4} presents the main asymptotic results for model selection for imputation. Section~\ref{sec:simu} reports simulation results, and Section~\ref{sec:finalrem} concludes with final remarks and presents future works. All proofs are relegated to the Appendix.


\section{Preliminaries} \label{sec:setup}

\subsection{The setup}
Consider a finite population $U := \{ 1, 2, ..., N\}$ of size $N$ and a survey variable $Y$. We are interested in estimating the finite population mean $$\mu:= \dfrac{1}{N}\sum_{k \in U} y_k,$$ where $y_k$ denotes the measurement of the survey variable  $Y$ of element $k \in U$. A sample $S \subset U$ of size $n$ is selected according to a sampling design $\cP$. The first and second order inclusion probabilities defined by $\pi_k:= \mathbb{P} ( k \in S )$ and $\pi_{kl}:= \mathbb{P} ( k,l \in S )$ are assumed to be strictly positive for all $k,l \in U$. The sample $S$ is also characterized by the vector of sample selection indicators $\textbf{I}_U:= [I_1, I_2, \ldots, I_N ]^\top$, where $I_k := 1$ if $k \in S$ and $I_k:=0$, otherwise. 

~\\
The Horvitz-Thompson estimator of $\mu$ is defined by 
\begin{equation}\label{ht}
\widehat{\mu}_\pi := \dfrac{1}{N}\sum_{k\in S} \dfrac{y_k}{\pi_k}.
\end{equation}
In practice, the estimator $\widehat{\mu}_\pi$ in \eqref{ht} is unfeasible when the survey variable $Y$ suffers from nonresponse. We denote by $\textbf{r}_U := \left[ r_1, r_2, \ldots, r_N\right]^\top $ the vector of response indicators, where $r_k := 1$ if $y_k$ is observed and $r_k := 0$, otherwise. The set of respondents and nonrespondents are denoted $S_r:= \{ k \in S  :  r_k = 1\}$ and $S_m:= \{ k \in S  : r_k = 0\}$ with respective cardinalities $n_r$ and $n_m$. We assume that $p$ covariates $X_1, X_2, \ldots, X_p,$ with measurements $\bx_k $ are observed for every $k \in S$. The observed data are given by $$\cD_{imp} := \{ ( \bx_k , y_k )  :  k \in S_r \} \cup \{ \bx_k : k \in S_m \}. $$
The vectors $[ \bx_k^\top, y_k,  r_k ]^\top_{k \in U}$ are assumed to be independent and identically distributed. The sampling design $\cP$ is assumed to be non-informative, that is, $\pi_k=\mathbb{P} ( I_k = 1 \rvert \boldsymbol{X}_U, \boldsymbol{y}_U) = \mathbb{P} ( I_k = 1 \rvert \boldsymbol{X}_U ) $, for every $k \in U$, where $\bX_U$ and $\boldsymbol{y}_U$ denote the population design matrix and the population vector of the survey variable, respectively. The reader is referred to \cite{pfeffermann2009inference} for further details. Similarly, the missing mechanism is assumed to satisfy the Missing At Random (MAR, \cite{rubin1976inference}) assumption, i.e., $\mathbb{P} ( r_k =1 \rvert \bx_k , y_k) =  \mathbb{P} ( r_k =1 \rvert \bx_k):=p(\bx_k) $. Finally,  we assume that there exists a positive constant $\rho>0$ such that $p(\bx)\geq \rho$, almost surely. This is the usual positivity assumption.
In this article, we restrict our attention to the customary homoscedastic linear regression model defined as
\begin{equation}\label{model}
y_k = \bx_k^\top \pmb{\beta} + \epsilon_k, \qquad k \in U,
\end{equation}
where $\epsilon_k$ satisfies $\mathbb{E} [\epsilon_k \rvert \bx_k ] = 0$ and $\mathbb{E} [ \epsilon_k^2 \rvert \bx_k ] := \sigma^2$.

\subsection{Additional notation} 
Probabilities, expectations and variances with respect to: (i) the distribution of the covariates' $\mathbb{P}_\bx$ will be denoted by the subscript $\bx$; (ii) $\mathbb{P}_{y\rvert \bx}$ will be denoted by the subscript $m$; (iii) the sampling design will be denoted by the subscript $p$; (iv) the nonresponse mechanism will be denoted by the subscript $q$. For a set of integers $A$, we write $\cP (A)$ to denote its power set, where, by convention, we remove the empty set. For two functions $f$ and $g$ defined on the same domain $D$, we write $f(x) \overset{\arg\min}{\equiv}  g(x)$ if $\argmin_{x \in D} f(x) = \argmin_{x\in D} g(x)$. Vectors like $\mathbf{x} \in \R^p$ will be bolded, and their $j$-th component will be denoted by $x^{(j)}$. We write $\bx, y, p(\bx), r, \epsilon$ to denote an independent copy of $(\bx_k, y_k, p(\bx_k), r_k,\epsilon_k)_{k\in U}$.\\ 
\section{Asymptotically optimal imputation}\label{sec:loss}
\subsection{Linear regression imputation and variable selection}

Linear regression imputation replaces missing values with predictions from a linear regression model fitted using the respondent data. Let $\bX_r := (\bx_k^\top )_{k\in S_r} \in \R^{n_r \times p}$ and $\boldsymbol{Y}_r\in \R^{n_r}$ denote the design matrix and the vector of measurements of $Y$ for the respondents, respectively. Provided that the matrix $\bX_r^\top \bX_r $ is positive definite, the ordinary least-squares estimator of $\pmb{\beta}$ in \eqref{model} is given by
\begin{equation}\label{eq:ols}
\widehat{\pmb{\beta}} = \left(\bX_r^\top \bX_r\right)^{-1} \bX_r^{\top} \boldsymbol{Y}_r.
\end{equation}
It is not uncommon in survey sampling to fit the imputation model by weighted least squares using the design weights $w_k=\pi_k^{-1}$, for $k\in S$. However, since the sampling design is assumed to be non-informative, the unweighted least-squares fit is justified. For algebraic simplicity, we therefore focus on the unweighted estimator in \eqref{eq:ols}. The same conclusions would  hold for the weighted version.
The linear regression imputed estimator of $\mu$ is defined as 
\begin{equation*}
\widehat{\mu}_{lr} := \dfrac{1}{N} \left( \sum_{k \in S_r} \dfrac{y_k}{\pi_k} + \sum_{k\in S_m} \dfrac{\bx_k^\top \widehat{\pmb{\beta}}}{\pi_k}\right). 
\end{equation*}
In practice, some coefficients of $\pmb{\beta}$ may be equal to zero. Let $\cA \subseteq \cP (\{ 1, 2, \ldots, p\})$ denote a family of candidate models. That is, each element $\alpha \in \cA$ is a subset of $\{ 1, 2, ..., p\}$ which we interpret as the set of included covariates. The complement of a model $\alpha$ is denoted $\alpha^c := \{ 1, 2, \ldots, p\}\backslash \alpha$. 
For example, if $p=3$, the model $\alpha=\{1,3\}$ includes the covariates $X_1$ and $X_3$. Its complement is $\alpha^c=\{2\}$, which corresponds to the model containing only the covariate $X_2$.

~\\
For any $\alpha \in \cA$, let $\bx_{k,\alpha}$ denote the subvector of covariates for unit $k$ corresponding to the indices in $\alpha$, and let $\bX_{r,\alpha} := (\bx_{k,\alpha}^\top)_{k\in S_r}$ denote the corresponding design matrix restricted to the respondents.

~\\
A model $\alpha \in \cA$ is correct if $\rVert \pmb{\beta}_{\alpha^c} \rVert_2= 0$, i.e., $\alpha$ contains all covariates with non-zero coefficients. The set of correct models is denoted $\cC$. The true model $\alpha^{\star}$ is defined as the support of $\pmb{\beta}$, that is, 
$$\alpha^{\star} := \left\{ j \in \left\{1, 2, ..., p \right\} \ ; \ \beta_j \neq 0\right\}.$$ 
The true model $\alpha^{\star}$ is the smallest correct model. Under the assumed linear model, it exists and is unique. A model $\alpha$ is said to be wrong if it is not correct; that is, if it omits at least one covariate whose coefficient in $\pmb{\beta}$ is nonzero. The set of wrong models is denoted $\cW$.

~\\
To each $\alpha \in \cA$, we may associate an ordinary least-squares estimator 
\begin{equation}\label{ols}
\widehat{\pmb{\beta}}_\alpha = \left( \bX_{r,\alpha}^\top \bX_{r,\alpha} \right)^{-1} \bX_{r,\alpha}^{\top} \boldsymbol{Y}_r,
\end{equation} leading to the corresponding imputed estimator of $\mu$: $$\widehat{\mu}_{\alpha} := \dfrac{1}{N} \left( \sum_{k \in S_r} \dfrac{y_k}{\pi_k} + \sum_{k\in S_m} \dfrac{\bx_{k, \alpha}^\top \widehat{\pmb{\beta}}_\alpha}{\pi_k}\right), \qquad \alpha \in \cA.$$  
The aim of this article is to determine which estimator in
\[
\Gamma_\cA := \left\{ \widehat{\mu}_\alpha \,;\, \alpha \in \cA \right\}
\]
should be used, and to develop a theoretically sound methodology for conducting inference on $\mu$ as efficiently as possible.

\subsection{A loss function for imputation}
For an arbitrary model $\alpha$, we may decompose the total error of $\widehat{\mu}_\alpha$ as
$$
\widehat{\mu}_\alpha - \mu
=
\underbrace{\left(\widehat{\mu}_\alpha - \widehat{\mu}_\pi\right)}_{\text{imputation/nonresponse error}}
+
\underbrace{\left(\widehat{\mu}_\pi - \mu\right)}_{\text{sampling error}} .
$$
The first term reflects the error induced by nonresponse through the imputation procedure (and therefore depends on the chosen model $\alpha$), whereas the second term is the usual sampling error which does not depend on nonresponse. At the imputation stage, the aim is to reduce the imputation/nonresponse error as much as possible.
 To this end, we introduce the following loss function for imputation: $$ \cL_n \left(\alpha\right) := \mathbb{E}_m \left[ \left( \widehat{\mu}_\alpha - \widehat{\mu}_\pi \right)^2\right], \qquad \alpha \in \cA.$$ The loss $\cL$ is a measure of the squared distance between the imputed estimator $\widehat{\mu}_\alpha$ based on model $\alpha$ and the complete data Horvitz-Thompson estimator. It is positive and is minimized for the model $\alpha$ that we would prefer to use for imputation. More formally, we define the optimal imputation model $\alpha_{\mathrm{opt}}$ as as any minimizer of $\cL_n$, that is, 
 $$\alpha_{\mathrm{opt}} \in \argmin_{\alpha\in \cA}\cL_n \left(\alpha\right).$$  
 In principle, one would like to use the imputed estimator $\widehat{\mu}_{\alpha_{\mathrm{opt}}}$ to estimate $\mu$. However, the loss $\cL_n$ depends on unobserved quantities and cannot be evaluated from the available data, so $\alpha_{\mathrm{opt}}$ is unknown in practice. In the remainder of the paper, we study the properties of $\widehat{\mu}_{\alpha_{\mathrm{opt}}}$ and develop practical model selection procedures to approximate $\alpha_{\mathrm{opt}}$ for inference.
 \begin{proposition} \label{prop1}
Fix $\alpha \in \cA$. Then, $\cL_n(\alpha)$ admits the following closed-form expression:
\begin{align*}
\cL_n \left(\alpha\right) &\overset{\arg\min}{\equiv}
\bigg\{ \bigg(\sum_{k \in S_m} \dfrac{\bx_k^\top}{\pi_k} - \sum_{k \in S_m} \dfrac{\bx_{k, \alpha}^{\top} }{\pi_k} \boldsymbol{A}_{r, \alpha}^{-1} \bX_{r,\alpha}^\top \bX_r \bigg)\pmb{\beta} \bigg\}^2  
\\
&\quad\quad + \sigma^2 \left( \sum_{k\in S_m} \dfrac{\bx_{k,\alpha}^\top }{\pi_k} \right) \boldsymbol{A}_{r, \alpha}^{-1} \left( \sum_{k\in S_m}\dfrac{\bx_{k,\alpha} }{\pi_k} \right)\\
&\ \ := \cL_{1,n}\left(\alpha\right) + \cL_{2,n}\left(\alpha\right). 
\end{align*}
\end{proposition}
\begin{proof}
See Appendix \ref{proof_prop1}.
\end{proof}
~\\
Proposition \ref{prop1} yields a useful decomposition of the loss into a bias term and a variance term. The component $\cL_{1,n}(\alpha)$ corresponds to the squared model bias induced by imputing under model $\alpha$, whereas $\cL_{2,n}(\alpha)$ captures a variance contribution. For any correct model $\alpha \in \cC$, we have $\cL_{1,n}(\alpha)=0$, so that correct models contribute only through $\cL_{2,n}(\alpha)$.
Lemma \ref{MatrixOverfitting} shows that $\cL_{2,n}$ is a strictly increasing set function, implying that adding covariates to the model can only increase $\cL_{2,n}$, although it may reduce $\cL_{1,n}$.  As a result, minimizing $\cL_n(\alpha)$ over $\cA$ amounts to a bias--variance trade-off.  We do not assume uniqueness of the minimizer; any $\alpha$ achieving $\min_{\alpha\in\cA}\cL_n(\alpha)$ is considered optimal.

\begin{proposition} \label{prop2}
If $\alpha_{\mathrm{opt}} \in \cC$, then, almost surely $$ \mathbb{P}_{pq} \left( \alpha_{\mathrm{opt}} = \alpha^{\star}\right)=1.$$
\end{proposition}
\begin{proof}
See Appendix \ref{proof_prop2}.
\end{proof}
~\\
Proposition \ref{prop2} states that, for any finite sample size, if the loss is minimized by a correct model (i.e., $\alpha_{\mathrm{opt}}\in\cC$), then the optimal model must be the true model. In particular, $\alpha_{\mathrm{opt}}=\alpha^\star$.

\subsection{Asymptotic behavior of the optimal imputation model}

In general, investigating the properties of $\alpha_{\mathrm{opt}} $ for finite sample sizes is challenging.  Nonetheless, we shall show in Section \ref{Section4} that, asymptotically, $\cL_n$ has a unique minimizer, which, under mild conditions, is also the true model.

~\\
To that aim, we consider the asymptotic framework of \cite{isaki1982survey}. Let $(U_v)_{v \in \N} $ be an increasing sequence of finite populations, i.e., $U_1 \subset U_2 \subset ...$ of respective increasing sizes $(N_v )_{v\in \N}$. In each population $U_v$, a sample $S_v$ is drawn randomly using a sampling design $\cP_v$ to estimate the finite population mean $\mu_v$. All subsequent steps (i.e., imputation, variance estimation, etc.) are carried out at each $v\in \N$ similarly. Note that the number of covariates $p$ remains constant as $v$ increases. In the sequel, for simplicity of notations, we will omit the index $v$ when no confusion arises.

~\\
The following assumptions are made about the superpopulation model.
\begin{enumerate}[label=({S\arabic*}), leftmargin=2.2em, start=1]
\item\label{S1} The distribution of the covariates $\mathbb{P}_\bx$ is absolutely continuous with respect to the Lebesgues measure. 
\item\label{S2} The covariates $\bx$ have bounded support, that is, there exists $C_0>0$ such that, almost surely,  $\rVert\bx \rVert_2 \leq C_0$.
\item\label{S3} There exists a constant $M_0$ such that $\mathbb{E}_{ m}[\epsilon_1^4]\leq M_0$, almost surely.
\end{enumerate}
~\\
Assumption \ref{S1} ensures that the matrix $\sum_{k\in U} \bx_k\bx_k^\top $ is almost surely positive definite. Condition \ref{S2} assumes that the covariates are almost surely bounded. This assumption is not strictly needed, but it greatly simplifies our proofs. An alternative assumption on the moments of $\rVert\bx\rVert_2$ could be used, at the price of considerably lengthening our arguments. As such, we do not pursue this avenue further. Assumption \ref{S3} is a moment condition that we use in various places. Note that in most of our results, we only require finite second moments. Generally, our conditions on the superpopulation may not be minimal. For brevity and simplicity of notation, we shall assume finite fourth moments directly, which will be needed for variance estimation later. 

\begin{remark}\label{req1}
Throughout the article, we work on the event $\mathcal{E}_{N_v} := \{\lambda_{min} (N_v^{-1}\sum_{k\in U_v}r_k\bx_k \bx_k^\top)\geqslant\gamma \}$ for some $\gamma>0$. We show in appendix (see Lemma \ref{HighProb}) that, for some universal constants $C_1, C_2$, for all large $N_v$, we have $\P_q (\mathcal{E}_{N_v}^c) \leq C_1 \exp(-C_2N_v)$ almost surely and $\mathcal{E}_{N_v} $ is thus a high-probability event. This conditioning could be removed by considering a truncated inverse $(N_v^{-1}\sum_{k\in U_v}r_k\bx_k \bx_k^\top)_{trunc}^{-1} =  (N_v^{-1}\sum_{k\in U_v}r_k\bx_k \bx_k^\top)^{-1}\mathds{1}_{\mathcal{E}_{N_v}}.$ Since the truncated and un-truncated estimators differ only with exponentially low probability, the difference between the two does not affect any of our asymptotic results. Therefore, we proceed by conditioning on $\mathcal{E}_{N_v}$ without loss of generality.

Similarly, in addition to the conditions on the sampling design described below, we assume that $\mathcal{P}$ induces a similar high-probability event $\{ \bX_r^\top \bX_r \succ 0\}$ and condition on this event. This is a natural assumption, commonly used in the literature, see e.g., \cite{chauvet2022asymptotic}.
\end{remark}

~\\
We will refer to the following assumptions on the sequence of sampling designs $(\cP_v)_{v\in \N}$.
\begin{enumerate}[label=({D\arabic*}), leftmargin=2.2em, start=1]
\item\label{D1} The sampling fraction satisfies $$ \lim_{v \to \infty} \dfrac{n_v}{N_v} := f^{\star} >0$$
almost surely.
\item\label{D2} There exists positive constants $\lambda>0$ and $\lambda^{\star}>0$ such that, for all $v \in \N$, $$ \min_{k \in U_v} \pi_{k,v}\geq\lambda, \qquad \text{and} \qquad \min_{k,l\in U_v} \pi_{kl,v}\geq\lambda^{\star} $$  almost surely.
\item\label{D3} The sampling covariances $\Delta_{kl}:=\pi_{kl}-\pi_k\pi_l$, for $k, l \in U_v$, satisfy that there exists a deterministic constant $\bar{\Delta}$ such that, almost surely, for all $v \in \mathbb{N}$, $$ n_v \max_{k,l \in U_v } \rvert \Delta_{kl} \rvert \leq \bar{\Delta}.$$
\end{enumerate}
~\\
These regularity conditions are common in the literature. Assumption \ref{D1} states that the sample size $n_v$ grows at the same rate as the population size $N_v$ as $ v$ goes to infinity. Assumption \ref{D2} requires that the first and second order inclusion probabilities remain bounded away from zero. Finally, Assumption \ref{D3} ensures that the sampling covariances decrease to zero at a rate of at least $\mathcal{O} (n_v^{-1}) $. For a more thorough discussion on these assumptions, we refer the reader to \cite{breidt2000local}.\\

\begin{theorem} \label{Theo1}
Let $( \alpha_{opt,v})_{v \in \N} $ be a sequence of minimizers of $(\cL_{v})_{v \in \N}$. For $\alpha \in \mathcal{A}$, define 
\begin{align}\label{eq:S_v} 
\pmb{S}_{U_v}&:=\sum_{k \in U_v}\frac{\bx_{k,\alpha^c}\bx_{k,\alpha^c}^{\top}p(\bx_k)\pi_k}{N_v}\nonumber\\
&\quad  -\sum_{k \in U_v}\frac{\bx_{k,\alpha^c}\bx_{k,\alpha}^{\top}p(\bx_k)\pi_k}{N_v}\left(\sum_{k\in U_v}\frac{\bx_{k,\alpha}\bx_{k,\alpha}^{\top}p(\bx_k)\pi_k}{N_v}\right)^{-1}\sum_{k \in U_v}\frac{\bx_{k,{\alpha}}\bx_{k,\alpha^c}^{\top}p(\bx_k)\pi_k}{N_v}. 
\end{align}
~\\
Assume \ref{S1}-\ref{S2} and \ref{D1}-\ref{D3} and that for $\alpha \in \mathcal{W}$,
\begin{equation}\label{wrong}
\liminf_{v \rightarrow\infty}\lambda_{min}\left(\textbf{S}_{U_v}\right)> 0 \qquad \text{a.s.} \tag{$C_1$}
\end{equation}
Then, we have
$$\lim_{v \to \infty} \mathbb{P}_{mpq} \left(\alpha_{opt,v} = \alpha^{\star}\right) = 1 \qquad \text{a.s.} $$
\end{theorem}
\begin{proof}
See Appendix \ref{proof_Theo1}.
\end{proof}
~\\
Theorem \ref{Theo1} proves that, among all candidates in $\cA$, the model $\alpha_{\mathrm{opt}}$ minimizing the (unobservable) imputation loss $\cL_n$ eventually is the true model $\alpha^\star$. In other words, asymptotically, the optimal set of covariates for imputation coincides with the true model $\alpha^\star$. Of course, in practice, the true model is unknown, so Theorem \ref{Theo1} is mainly theoretical. 
Its value lies in guiding the development of practical selection criteria that asymptotically identify the true model. We revisit this point in Section \ref{Section4}.

~\\
Assumption~\eqref{wrong} ensures that the omitted covariates $\bx_{\alpha^c}$ retain some variation once the included covariates $\bx_\alpha$ have been accounted for. To see this, 
consider the simple case where $p(\bx_k) = p$, $\pi_k=\pi$ for $k\in U$ and the intercept is included in $\alpha$. Then, some algebra shows that the matrix $ \boldsymbol{S}_{U_v}$ defined in \eqref{eq:S_v} is proportional to the following empirical covariance matrix
$$ \boldsymbol{S}_{U_v}\propto \dfrac{1}{N_v} \sum_{k\in U_v}\mathbf{r}_k \mathbf{r}_k^\top$$ with $\mathbf{r}_k := \bx_{k,\alpha^c} - \boldsymbol{B}^\top_v\bx_{k,\alpha} $ with $\boldsymbol{B}_v $ solving the matrix least squares problem $$\boldsymbol{B}_v = \argmin_{\boldsymbol{B} \in \R^{p_\alpha\times p_\alpha^c}} \dfrac{1}{N_v}\sum_{k\in U_v} \rVert \bx_{k,\alpha^c} -\boldsymbol{B}^\top\bx_{k,\alpha} \rVert_2^2 =\left( \sum_{k\in U_v} \bx_{k,\alpha} \bx_{k,\alpha}^\top \right)^{-1}  \sum_{k\in U_v} \bx_{k,\alpha} \bx_{k,\alpha^c}^\top .$$
Therefore, $ \boldsymbol{S}_{U_v}$
represents the finite population covariance matrix of the residuals obtained by regressing $\bx_{\alpha^c}$ on $\bx_\alpha$. The assumption \eqref{wrong} requires that every omitted direction must preserve a strictly positive amount of variability after removing the linear effect of the included variables. Intuitively, it rules out cases in which an omitted predictor could be perfectly reconstructed from the included predictors, thereby making some incorrect models behave as if they were correct. That is, missing covariates would have essentially no missing contribution. In our setting, assumptions \ref{S1} and \ref{S2} imply that $\E_\bx [\bx \bx^\top ]$ is positive definite, so that any covariate missing will retain some variability, even after accounting for all the other covariates. The next remark formalizes this. 
\begin{remark}
For an   equal probability sampling design, an application of the law of large numbers gives
\begin{equation*}
\pmb{S}_{U_v}\xrightarrow[v \to \infty]{a.s.} \pmb{S}_* := \pi \left(\mathbb{E}_\bx[\bx_{\alpha^c}\bx_{\alpha^c}^{\top}p(\bx)]-\mathbb{E}_\bx[\bx_{\alpha^c}\bx_{\alpha}^{\top}p(\bx)](\mathbb{E}_\bx[\bx_{\alpha}\bx_{\alpha}^{\top}p(\bx)])^{-1}\mathbb{E}_\bx[\bx_{\alpha}\bx_{\alpha^c}^{\top}p(\bx)]\right).
\end{equation*}
Now, consider the block matrix $$\boldsymbol{M}:
=
\pi\,\mathbb{E}\!\begin{bmatrix}
p\,\bx_{\alpha}\bx_{\alpha}^{\top}
&
p\,\bx_{\alpha}\bx_{\alpha^c}^{\top}
\\[4pt]
p\,\bx_{\alpha^c}\bx_{\alpha}^{\top}
&
p\,\bx_{\alpha^c}\bx_{\alpha^c}^{\top}
\end{bmatrix}
=
\begin{bmatrix}
\boldsymbol{A} & \boldsymbol{B} \\
\boldsymbol{B}^{\top} & \boldsymbol{C}
\end{bmatrix},$$
and note that $\pmb{S}_* = \boldsymbol{C}-\boldsymbol{B}^\top \boldsymbol{A}^{-1}\boldsymbol{B}$ is the Schur complement of $\boldsymbol{A}$ in $\boldsymbol{M}$. Under \ref{S1} and positivity, it can be shown that $\boldsymbol{M}$ is positive definite. This implies in turn that $\pmb{S}_*$ is also positive definite, see e.g., \cite[p.~495]{horn2012matrix}. Therefore, Assumption \eqref{wrong} holds in this case. A similar reasoning would hold for Poisson sampling. More generally, Assumption \eqref{wrong} is fairly mild and is likely to hold for most sampling designs. In spirit, the assumption is in line with identifiability assumptions of the literature, see e.g., \cite{shao1993linear}.
\end{remark}
~\\
We recall that a sequence of estimators $(\widehat{\mu}_v)_{v\in \N}$ is consistent for $\mu_v$ if $\widehat{\mu}_v - \mu_v \xrightarrow[v\to \infty]{\P} 0$. Moreover, if $\V (\widehat{\mu}_v) $ exists for all $v \in \N$ and $\V (\widehat{\mu}_v) \xrightarrow[v\to\infty]{} \tau^2$, we call $\tau^2$ the asymptotic variance of $\widehat{\mu}_v$ and write $\mathbb{\mathbb{AV}}(\widehat{\mu}_v) := \lim_{v\to \infty}\V (\widehat{\mu}_v)$. Our next result further examines the consequences, in terms of consistency and asymptotic variance, when using an imputed estimator $\widehat{\mu}_\alpha$ of $\widehat{\mu}_v$ with $\alpha \neq \alpha^{\star}$. 

\begin{proposition}\label{prop3}
Let $\alpha \in \cA$ be an arbitrary model and $(\widehat{\mu}_{\alpha, v})_{v\in \N}$ be sequence of imputed estimators. Assume \ref{S1}-\ref{S2}, \ref{D1}-\ref{D3}, and assume also that the sampling design induces equal first-order inclusion probabilities ($\pi_k=\pi$ for $k \in U_v$) and equal second-order inclusion probabilities ($\pi_{kl}=\pi^{\star}$ for $k \neq l \in U_v$). Then, we have the following results.
\begin{itemize}
\item[(i)] The sequence $(\widehat{\mu}_{\alpha, v})_{v\in \N}$ is consistent if and only if 
\begin{equation}\label{condBias}
\mathbb{E}_\bx\left[( 1 - p(\bx))\bx_{\alpha^c}^{\top}\pmb{\beta}_{\alpha^c}\right]=\mathbb{E}_\bx \left[\bx_{\alpha}^{\top}(1 - p(\bx))\right] \mathbb{E}_\bx \left[p(\bx) \bx_{ \alpha}\bx_{ \alpha} ^\top \right]^{-1} \mathbb{E}_\bx\left[\bx_{\alpha}\bx_{\alpha^c}^{\top}\pmb{\beta}_{\alpha^c}p(\bx) \right]. \tag{$C_2$}
\end{equation}
\item[(ii)] 
\begin{enumerate}
    \item[(a)] Let $\alpha \in \cC$ be a correct model. Assume that there exists a constant $K$ such that the following limit exists $\lim_{v\rightarrow \infty}N_v(\pi^{\star}-\pi^2)=K$. Then, there exists a constant $C(\alpha^{\star})$, not depending of $\alpha$, and a function $M : \cC \to \R_+^*$ such that
\begin{equation*}
\mathbb{AV}\left(\sqrt{N_v}(\widehat{\mu}_{\alpha,v}-\mu_v)\right)=C(\alpha^{\star})+\frac{\sigma^2}{\pi}M(\alpha),
\end{equation*}
where
\begin{equation*}
M: \alpha\mapsto\mathbb{E}_\bx\left[(1-p(\bx))\bx_{\alpha}^{\top}\right]\left(\mathbb{E}_\bx\left[p(\bx)\bx_{\alpha}\bx_{\alpha}^\top\right]\right)^{-1}\mathbb{E}_\bx\left[(1-p(\bx))\bx_{\alpha}\right].
\end{equation*}
\item[(b)]
The function $M$ is a non-decreasing function, that is, for $\alpha_1, \alpha_2 \in \cC$,
\begin{equation*}\label{eq:M}
 \alpha_1 \subset \alpha_2 \ \implies \    M(\alpha_1)\leq M(\alpha_2),
\end{equation*}
with equality if and only if 
\begin{align}\label{eq:cond2}
    &\mathbb{E}_\bx\left[(1-p(\bx))\bx_{\alpha_2-\alpha_1}\right] \nonumber\\
   &\quad= \mathbb{E}_\bx\left[(1-p(\bx))\bx_{\alpha_2-\alpha_1} \bx_{\alpha_1}^{\top}\right]\mathbb{E}_\bx\left[p(\bx)\bx_{\alpha_1}\bx_{\alpha_1}^{\top}\right]^{-1}\mathbb{E}_\bx\left[p(\bx)\bx_{\alpha_1}\right]. \tag{C3}
\end{align}
\end{enumerate}
\end{itemize}
\end{proposition}
\begin{proof}
See Appendix \ref{proof_Prop3}.
\end{proof}
~\\
The restriction on equal inclusion probabilities is essentially technical and could be removed if we considered weighted least-squares in \eqref{ols} with weights $(1/\pi_k)_{k\in S_r}$. The first statement (i) provides a condition \eqref{condBias} under which the imputed estimator $\widehat{\mu}_\alpha$ is consistent. The second statement (ii) derives an expression for the asymptotic variance of the rescaled error $\sqrt{N_v}(\widehat{\mu}_{\alpha,v}-\mu_v)$ for $\alpha\in \cC$. 
By extension, the term "asymptotic variance of a model $\alpha$" will refer to that of $\sqrt{N_v}(\widehat{\mu}_{\alpha,v}-\mu_v)$. The quantity 
$\mathbb{AV}(\sqrt{N_v}(\widehat{\mu}_{\alpha,v}-\mu_v))$ decomposes into two terms, $C(\alpha^{\star})$ and $M(\alpha)$. The first does not depend on $\alpha$.  Since $\alpha^{\star}$ is the smallest correct model, we have $M(\alpha^{\star})\leq M(\alpha)$ for every $\alpha \in \mathcal{C}$. 

~\\
The aim of Proposition~\ref{prop3} is therefore to highlight the consequences of using (i) too few covariates or (ii) too many. While these formulas may appear somewhat abstract, we analyze them further below to obtain more interpretable conditions.
\begin{corollary}\label{Cor1}
Consider the set-up of Proposition \ref{prop3}. Assume that either of the following two cases holds:
\begin{itemize}
\item[(i)] The model is correct, i.e., $\alpha\in \cC$. 
\item[(ii)] Let $\alpha_{\mathrm{mis}} := \alpha^{\star}\cap \alpha^c$ denote the set of correct covariates, missing in $\alpha.$ Assume the following conditions:
\begin{itemize}
\item[(a)] The correct missing covariates do not influence $p$,  i.e., for all $j \in \alpha_{\mathrm{mis}}$, we have $\left[ p(\bx) \indep \bx_{\alpha^c,j} \right]\, \rvert \,  \bx_{\alpha}.$
\item[(b)] The correct missing covariates are linearly linked to those included, that is,  for all $j \in \alpha_{\mathrm{mis}}$, $\mathbb{E}[ x_{\alpha^c,j} \rvert \bx_{\alpha} ] = \bx_\alpha^\top \boldsymbol{\gamma}_j,$ for some  $\boldsymbol{\gamma}_j \in \mathbb{R}^{p_\alpha}.$
\end{itemize}
\end{itemize}
Then, condition \eqref{condBias} holds and thus $ \widehat{\mu}_{\alpha,v} -\mu_v\xrightarrow[v \to \infty]{\mathbb{P}} 0.$
\end{corollary}
\begin{proof}
See Appendix \ref{proof_Cor1}.
\end{proof}
~\\
To obtain a consistent estimator, it is therefore sufficient to either: (i) include all covariates related to $Y$; (ii) include all covariates that are related to both $Y$ and the probability of response, modulo condition (b). The second condition is more difficult to interpret, yet relatively weak, but essential. It states that the omitted covariates relevant to $Y$ must be linearly related to those included. Although this may at first appear to be a technical assumption introduced to simplify the proof, it is in fact required for \eqref{condBias} to hold, as illustrated in the next example.

\begin{example}
\textbf{(i).}  
Recall that a missing data mechanism is said to be Missing Completely At Random (MCAR) if $r_k \indep y_k$. Then, condition (a) is trivially satisfied since $r_k$ must be independent of $\bx$. If the covariates are independent, then condition (b) reduces to requiring the intercept to be included in $\alpha$. This requirement is necessary, as illustrated in the following example.

~\\
Consider the case where $p=2$ with $\bx^\top = [x_1, x_2] =[1, x]$, with true model $\alpha^{\star} = \{1, 2\}$ and candidate $\alpha = \{2\}$, and suppose that $p(\bx) = \bar{p}$, $0<\bar{p}\le 1$. Then,  $\alpha_{\mathrm{mis}} = \{1\}$. Condition (a) is satisfied but condition (b) is not since $\mathbb{E}\left[ 1 \rvert \bx_{\alpha} \right] = 1$ and, because the intercept is not included, there does not exist a constant $\gamma$ such that $\gamma$ such that $1 = x \gamma$ almost surely, unless $x$ is almost surely constant. Condition (b) is indeed needed since it can be shown that $$ \widehat{\mu}_{\alpha,v} -\mu_v\xrightarrow[v \to \infty]{\mathbb{P}} -\beta_0 \left( 1-\bar{p}\right) \dfrac{\mathbb{V}_\bx\left(x\right)}{\mathbb{E}_\bx \left[x^2\right]},$$ where $\beta_0$ denotes the true intercept coefficient. This asymptotic bias is non-zero whenever $\beta_0 \neq 0$, $\bar{p} \neq 1$ and $\mathbb{V}_\bx(x)\neq 0$; note that if $x$ has zero variance then there exists $\gamma$ such that $1 = x \gamma$, and consistency holds. 

~\\
\textbf{(ii).} If the intercept is included and $\bx$ has an elliptically symmetric distribution (e.g., multivariate Gaussian), or if its components are independent, then condition (b) is always satisfied and (a) is sufficient to ensure consistency. 
\end{example}
~\\
In the following corollary, we examine the behavior of the asymptotic variance obtained in statement (ii) of Proposition \ref{prop3}.

\begin{corollary}\label{Cor2}
Consider the setup of Proposition \ref{prop3} with $\alpha_1, \alpha_2 \in \cC$ such that $\alpha_1\subset \alpha_2$. 
Consider the following assumptions.
\begin{itemize}
\item[(i)] Assume that:
\begin{enumerate}
    \item[(a)] For all $j \in \alpha_{2}-\alpha_1$, $\left[ p(\bx) \indep x_{\alpha_2-\alpha_1,j} \right]\, \rvert \,  \bx_{\alpha_1}.$
\item[(b)] For all $j \in \alpha_{2}-\alpha_1$, $\mathbb{E}\left[ x_{\alpha_2-\alpha_1,j} \rvert \bx_{\alpha_1} \right] = \bx_{\alpha_1}^\top \boldsymbol{\gamma}_j,$ for some  $\boldsymbol{\gamma}_j \in \mathbb{R}^{p_{\alpha_1}}.$
\end{enumerate}
Then, \eqref{eq:cond2} holds and the asymptotic variances of the models based on $\alpha_1$ and $\alpha_2$ are the same.

\item[(ii)] Assume that there exists a direction $\mathbf{c}\in \R^{\alpha_2-\alpha_1}$ such that 
\begin{enumerate}
    \item[(a)] $\E_\bx[p(\bx) \mathbf{c}^\top \bx_{\alpha_2-\alpha_1} \bx_{\alpha_1}^\top] = \mathbf{0}.$
    \item[(b)] $\E_\bx[(1-p(\bx)) \mathbf{c}^\top \bx_{\alpha_2-\alpha_1}] \neq  0.$
\end{enumerate}
Then, \eqref{eq:cond2} does not hold and $\alpha_2$ has a strictly greater variance than $\alpha_1.$
\end{itemize}
\end{corollary}

\begin{proof}
    See Appendix \ref{proofCoro2}.
\end{proof}
~\\
In the above corollary, we make explicit some sufficient conditions under which one "pays a price" when adding superfluous covariates, and when one does not. Specifically, in part (i), our conditions mean that: (a) the covariates added in $\alpha_2$ (compared to $\alpha_1$) do not explain the nonresponse mechanism, given the covariates in $\alpha_1$. This occurs for instance when $p(\bx) = p(\bx_{\alpha_1})$. Condition (b) means that the addtional covariates are linearly related to those already included in $\alpha_1$. Together, these conditions describe a situation in which the covariates in $\alpha_2-\alpha_1$ are both "uninformative, in mean"  and "irrelevant" for the nonresponse mechanism, once the covariates in $\alpha_1$ are included. For example, consider the true model $\alpha^\star=\{1,2\}$ with $X_1=1$, and suppose that $X_2$ and $X_3$ are independent. Then condition (b) is satisfied, and if $p(\bx)=p(X_1,X_2)$, using $X_1$, $X_2$, and $X_3$ leads to the same asymptotic variance as using only $X_1$ and $X_2$.  In part (ii), we move in the opposite direction  and provide an example where adding superfluous covariates increases the variance. To better interpret these conditions, consider a new random variable $z = \mathbf{c}^\top \bx_{\alpha_2-\alpha_1}$, representing a direction in the space spanned by the added covariates $\alpha_2 - \alpha_1$. The first condition means that $z$ is a new direction, not captured by the covariates in $\alpha_1$, among respondents, while the second requires that there be a nonzero nonresponse signal along this direction.

\section{A methodology for asymptotically optimal variable selection}\label{Section4}
Theorem~\ref{Theo1} establishes that the true model is optimal for imputation under the loss function~$\mathcal{L}$. In practice, however, the true model is unknown, and the corresponding imputed estimator $\widehat{\mu}_{\alpha^{\star}}$ serves only as an oracle benchmark.  In this section, we  develop a practical procedure that attains the same asymptotic efficiency as this oracle imputed estimator and yields asymptotically valid confidence intervals.
\subsection{Asymptotic equivalence with the oracle}
Model selection for identification (i.e., recovering the true non-zero coefficients of a linear regression model) has been extensively studied for i.i.d. data; we refer the reader to \cite{shao1993linear, shao1997asymptotic} or \cite{rao2001model} for a textbook discussion. We call a model selection criterion any (possibly data dependent) measurable map $\cC_n: \cA \to \R$ used to select a model $\alpha\in \cA$, that is, to choose $\widehat{\alpha}_{\cC_n}\in \cA$ satisfying $$\widehat{\alpha}_{\cC_n} \in \argmin_{\alpha\in \cA} \cC_n\left(\alpha\right).$$ Common model selection criteria for linear models include the Akaike Information Criterion (AIC, \cite{akaike1970statistical}), the Bayesian Information Criterion (BIC, \cite{schwarz1978estimating}), cross-validation, among many others. On i.i.d. data, a model selection criterion $\cC_n$ is said to be consistent if $$\lim_{n \to \infty} \mathbb{P}_m \left(\widehat{\alpha}_{\cC_n} = \alpha^{\star} \right)=1$$ almost surely. That is, as the sample size $n$ diverges to infinity, the model selection tends to only select the true model. Not all model selection criteria are consistent; for instance, under appropriate conditions, the BIC criterion is consistent, while leave-one-out cross-validation is not. The next lemma formalizes that, under the MAR and non-informativeness assumptions, if $\cC_n$ is consistent on i.i.d. data, then it is also consistent with survey data. 
\begin{lemma} \label{Lemma2}
Let $\cC_{U_v}$ be a model selection criterion fitted on $\left\{\left(\bx_k, y_k \right) : k \in U_v \right\}$ and  $\cC_{S_{r}}$ be the corresponding algorithm fitted on $\{(\bx_k, y_k ) : k \in S_{r,v} \}$. Denote by $\widetilde{\alpha}_{U_v}$ and $\widehat{\alpha}_{S_{r,v}}$ their respective minimizers. Assume \ref{S1}-\ref{S2}, \ref{D1}-\ref{D3}.        \\
If $$\lim_{v \to \infty} \mathbb{P}_{m} \left(\widetilde{\alpha}_{U_v} = \alpha^{\star}\right)  =1, \qquad a.s.$$ then $$\lim_{v \to \infty} \mathbb{P}_{mpq} \left(\widehat{\alpha}_{S_{r,v}} = \alpha^{\star}\right)  =1, \qquad a.s.$$
\end{lemma}
\begin{proof}
See Appendix \ref{proof_Lemma2}.
\end{proof}
~\\
Lemma \ref{Lemma2} formalizes the intuition that model-selection consistency transfers from the population to the sample of respondents, provided that the sampling design and the nonresponse mechanism do not shift the conditional distribution $\mathbb{P}_m$. Although this result is quite natural, it plays an important role in the developments that follow. A natural extension of Lemma \ref{Lemma2} would be to consider informative sampling designs with appropriately weighted model selection criteria. Such criteria have been suggested, for instance, in \cite{lumley2015aic, wieczorek2022k, iparragirre2023variable}; however, to the best of our knowledge, their consistency has not yet been formally established. This question is beyond the scope of this article and will be relegated to future work. 

~\\
The first step in our methodology is to use a consistent model selection criterion to select a model $\widehat{\alpha}$. Then, to use the imputed estimator $\widehat{\mu}_{\widehat{\alpha}}$ based on the selected set of covariates $\widehat{\alpha}$. This will asymptotically lead to an optimal estimator among all possible models. 

\begin{theorem} \label{Theo2}
Let $(\widehat{\alpha}_v)_{v\in \N}$ be a sequence of models selected by a consistent model selection procedure and $(\widehat{\mu}_{\widehat{\alpha}, v})_{v\in \N}$ be the corresponding sequence of imputed estimators. Assume \ref{S1}-\ref{S2}, \ref{D1}-\ref{D3}. Then,
\begin{equation*}
\sqrt{n_v} \left(\widehat{\mu}_{\widehat{\alpha},v} -\mu_v \right)=      \sqrt{n_v} \left(\widehat{\mu}_{\alpha^{\star},v} -\mu_v \right) + o_{\mathbb{P}} (1).
\end{equation*}
\end{theorem}
\begin{proof}
See Appendix \ref{Proof_Theo2}.
\end{proof}
~\\
The above states that, once a model selection procedure asymptotically identifies the true model $\alpha^\star$ with probability one, the asymptotic distribution of the feasible estimator $\widehat{\mu}_{\widehat{\alpha}}$ is the same as that of the oracle estimator. In other words, we can proceed with inference as if the true model were known a priori. This legitimizes the use of standard model-selection tools, such as BIC, for survey imputation and shows that the bias or extra variability introduced by using a data-driven model $\widehat{\alpha}$ disappears asymptotically.  The key idea of the proof is that, on the event $\widehat{\alpha}_v = \alpha^{\star}$, we have $\widehat{\mu}_{\widehat{\alpha}} = \widehat{\mu}_{\alpha^\star}$. Because this event has asymptotic probability $1$, what happens on the complement does not matter.

\subsection{Consistent variance estimation}

The second step of the methodology is to perform "classical" variance estimation, based on the model $\widehat{\alpha}$ selected by a consistent model selection procedure. By classical, we mean using the same variance estimators as traditionally used, but with model $\widehat{\alpha}$ instead of all available covariates. In this article, we focus on the reverse approach \citep{fay1991design, shao1999variance}, although a similar approach with the method of Särndal \citep{sarndal1992methods} could also be used. For additional details on variance estimation with the reverse approach, we refer the reader to \cite{kim2009unified} and \cite{haziza2020variance}. 
The next theorem establishes that, for any $\alpha \in \mathcal{C}$, there exists a first-order asymptotically equivalent linear (in the sampling indicators) estimator of $\widehat{\mu}_{\alpha}$. More specifically, let $$    \widetilde{\mu}_{\alpha,v}:=\frac{1}{N_v}\sum_{k \in S_v}\frac{\eta_{k,\alpha}}{\pi_k},
$$ denote the linearized imputed estimator with $\eta_{k,\alpha}=\bx_{k,\alpha}^{\top}\pmb{\beta}_{\alpha}+r_k(1+\pi_k\textbf{c}_{\alpha,v}^{\top}\bx_{k,\alpha})\epsilon_k$ and 
\begin{equation}\label{cbold}
\textbf{c}_{\alpha,v}=\left(\sum_{k \in U_v}\frac{r_k\pi_k\bx_{k,\alpha}\bx_{k,\alpha}^{\top}}{N_v}\right)^{-1}\sum_{k \in U_v}\frac{(1-r_k)\bx_{k,\alpha}}{N_v}.
\end{equation}
The result below closely parallels Theorem~1 of \cite{kim2009unified}; however, their proof does not apply directly to our setting because we use an unweighted least-squares estimator $\widehat{\pmb{\beta}}$, whereas they rely on a weighted estimator. For completeness, we extend their theorem to unweighted regression imputation.

\begin{theorem}\label{Theo3}
Let $\alpha \in \cC$ be a correct model and consider sequences $(\widehat{\mu}_{\alpha,v})_{v\in \N}$ and $(\widetilde{\mu}_{\alpha,v})_{v\in \N}$. 
Under assumptions \ref{S1}-\ref{S3} and \ref{D1}-\ref{D3}, we have
\begin{equation}\label{kim}
\sqrt{n_v}\left(  \widehat{\mu}_{\alpha,v}-\widetilde{\mu}_{\alpha,v}\right)=o_{\mathbb{P}}(1).
\end{equation}
\end{theorem}
\begin{proof}
See Appendix \ref{proof_Theo3}.
\end{proof}

To perform variance estimation, consider the following decomposition of $V_T(\alpha)$, the variance of $\widetilde{\mu}_{\alpha}-\mu$:
\begin{align*}
V_T(\alpha)= \mathbb{V}_{mpq}(\widetilde{\mu}_{\alpha}-\mu)& \, \, =\, \mathbb{E}_q\left[\mathbb{E}_m\left[\mathbb{V}_p(\widetilde{\mu}_{\alpha}-\mu)\right]\right]+\mathbb{E}_q\left[\mathbb{V}_m\left(\mathbb{E}_p[\widetilde{\mu}_{\alpha}-\mu]\right)\right]\\
&\quad+\mathbb{V}_{q}\left(\mathbb{E}_m\left[\mathbb{E}_p[\widetilde{\mu}_{\alpha}-\mu]\right]\right)\\&:=V_1(\alpha)+V_2(\alpha)+V_3(\alpha).
\end{align*}
Noting that for $\alpha\in \cC$,
\begin{equation*}
\mathbb{E}_{mp}[\widetilde{\mu}_{\alpha}-\mu ]=0,
\end{equation*}
it follows that $$\mathbb{V}_{mpq}(\widetilde{\mu}_{\alpha}-\mu)=V_1(\alpha)+V_2(\alpha).$$
We start by estimating $\pmb{\beta}_{\alpha}$ by $\widehat{\pmb{\beta}}_\alpha$ and $\boldsymbol{c}_\alpha$ with 
\begin{equation}\label{cwidehat}
\widehat{\textbf{c}}_{\alpha}=\left(\sum_{k \in S}\frac{r_k\bx_{k,\alpha}\bx_{k,\alpha}^{\top}}{N}\right)^{-1}\sum_{k \in S}\frac{(1-r_k)\bx_{k,\alpha}}{N\pi_k}.
\end{equation}
~\\
The estimator $\widehat{V}_{1}(\alpha)$ of $V_1(\alpha)$ is defined as  
\begin{equation*}
\widehat{V}_1(\alpha)=\frac{1}{N^2}\sum_{k \in S}\sum_{l \in S}\frac{\Delta_{kl}}{\pi_{kl}}\frac{\widehat{\eta}_k}{\pi_k}\frac{\widehat{\eta}_l}{\pi_l},
\end{equation*}
with \begin{equation}\label{etawidehat}
\widehat{\eta}_{k,\alpha}=\bx_{k,\alpha}^{\top}\widehat{\pmb{\beta}}_{\alpha}+r_k\left(1+\pi_k\widehat{\textbf{c}}_{\alpha}^{\top}\bx_{k,\alpha}\right)\left(y_k-\bx_{k,\alpha}^{\top}\widehat{\pmb{\beta}}_{\alpha}\right), \qquad k\in S.
\end{equation}
~\\
To estimate $V_2(\alpha)$, we start by estimating $\sigma^2$ with \begin{equation}\label{sigmaest}
\widehat{\sigma}_{\alpha}^2=\sum_{k \in S}\frac{r_k\left(y_k-\bx_{k,\alpha}^{\top}\widehat{\pmb{\beta}}_{\alpha}\right)^2}{n_r-p_{\alpha}}.
\end{equation}
We then estimate $V_2(\alpha)$ by
\begin{equation*}
\widehat{V}_2(\alpha)=\widehat{\sigma}_{\alpha}^2\sum_{k \in S}\frac{1-r_k+r_k\left(\pi_k\widehat{\textbf{c}}_{\alpha}^{\top}\bx_{k,\alpha}\right)^2}{N^2\pi_k}.
\end{equation*}
Finally, for $\alpha\in \cC$, an estimator of the total variance, $\V_{mpq}(\widehat{\mu}_\alpha - \mu ),$ is given by $\widehat{V}_T(\alpha)=\widehat{V}_1(\alpha)+\widehat{V}_2(\alpha)$. This is, for a fixed $\alpha \in \cC$, the customary variance estimator obtained via the reverse approach; see, e.g., \cite{kim2009unified, haziza2020variance} for additional details. Although widely used, the consistency of the variance estimator $\widehat{V}_T$ has, to our knowledge, not been rigorously established. The next result provides a formal justification. To this end, we impose the following additional assumption on the sampling design:
\begin{enumerate}[label=({D\arabic*}), leftmargin=2.2em, start=4]
\item \label{D4}    Let $D_{4,N_v}$ be the set of the distinct of 4-tuples from $U_v$. The sampling design satisfies
\begin{equation*}
\lim_{v \rightarrow \infty}\max_{(i,j,k,l)\in D_{4,N_v}}\rvert(I_iI_j-\pi_{ij})(I_kI_l-\pi_{kl})\rvert=0,
\end{equation*}
almost surely.
\end{enumerate}
This assumption is commonly used to establish the consistency of the Horvitz-Thompson variance estimator, see, e.g., \cite{breidt2000local}.

\begin{theorem}\label{Theo4}
Let $\alpha \in \cC$ and consider a sequence of variance estimators $\{\widehat{V}_{T,v}(\alpha)\}_{v \in \N}$. Assume \ref{S1}-\ref{S3} and \ref{D1}-\ref{D4}, and that there exists $c_{1,\alpha}$ and $c_{2,\alpha}$ such that
\begin{equation*}
N_v\mathbb{V}_p\left[(\tilde{\mu}_{\alpha,v}-\mu_v)\right]\xrightarrow[v \rightarrow \infty]{\mathbb{P}}c_{1,\alpha},\qquad N_v\mathbb{V}_{m}\left(\mathbb{E}_p\left[\tilde{\mu}_{\alpha,v}-\mu_v \right] \right)\xrightarrow[v\rightarrow\infty]{\P}c_{2,\alpha}.
\end{equation*}

Then, we have
\begin{equation*}
n_v\left\rvert\widehat{V}_{T,v}(\alpha)-V_{T,v}(\alpha)\right\rvert=o_{\mathbb{P}}(1).
\end{equation*}
\end{theorem}
\begin{proof}
See Appendix \ref{proof_Theo4}.
\end{proof}
~\\
We propose estimating the variance based on the model $\widehat{\alpha}$ selected by a consistent model selection criterion.  The next result establishes the validity of this approach.

\begin{theorem}\label{Theo5}
Let $\{\widehat{\alpha}_v\}_{v\in \mathbb{N}}$ be a sequence of models selected by a consistent model selection procedure, and $\{\widehat{V}_T(\widehat{\alpha}_v)\}_{v \in \mathbb{N}}$ be the corresponding sequence of variance estimators. Assume \ref{S1}-\ref{S3}, \ref{D1}-\ref{D4}, and that there exists a constant $c>0$ such that, almost surely,
\begin{equation*}
\lim_{v \to \infty} n_v V_{T,v}(\alpha^{\star}) = c.
\end{equation*}
Then, the estimator $\widehat{V}_{T,v}(\widehat{\alpha}_v)$ is consistent, that is,
\begin{equation*}
\dfrac{\widehat{V}_{T,v}(\widehat{\alpha}_v)}{V_{T,v}(\alpha^{\star})}\xrightarrow[v \to \infty]{\mathbb{P}}1.
\end{equation*}
\end{theorem}
\begin{proof}
See Appendix \ref{proof_Theo5}.
\end{proof}
~\\
Theorem~\ref{Theo5} shows that the variance estimator computed under the selected model $\widehat{\alpha}_v$ is asymptotically equivalent to the variance that would be obtained under the true model $\alpha^\star$. In other words, using a consistent model selection procedure does not affect first-order variance estimation: we may estimate the variance as if the true model were known.
\subsection{Asymptotically valid and optimal confidence intervals}

The final step of the proposed methodology is to derive asymptotically valid confidence intervals. To this end, we first establish the asymptotic distribution of the imputed estimator based on a consistent model selection procedure.

\begin{theorem}\label{Theo6}
Consider a sequence $(\widehat{\alpha}_v)_{v\in \mathbb{N}}$ of consistent models and let $(\widehat{\mu}_{\widehat{\alpha}})_{v\in \N}$ and $(\widehat{V}_{T,v}(\alpha))_{v\in \N}$ be the corresponding point and variance estimators, respectively. Let $\cB_v := \sigma\left(\left(r_k, \bx_k, \epsilon_k\right)_{k\in U_v}\right)$.  Assume \ref{S1}-\ref{S3}, \ref{D1}-\ref{D4} and the following conditions. 
\begin{enumerate}
\item[(i)] The oracle estimator $\widetilde{\mu}_{\alpha^{\star},v}$ satisfies a design central limit theorem, that is, 
$$\mathbb{V}_p \left( \widetilde{\mu}_{\alpha^{\star},v}  \right)^{-1/2} \left( \dfrac{1}{N_v}\sum_{k\in U_v} \left(\dfrac{I_k}{\pi_k}-1 \right) \eta_{k, \alpha^{\star}}\right) \bigg\rvert \ \cB_v \xrightarrow[v \to \infty]{\cL} \cN (0,1), $$
\item[(ii)] There exists a constant $c_3>0$ such that $$N_v\mathbb{V}_p \left( \widetilde{\mu}_{\alpha^{\star},v}  \right) \xrightarrow[v \to \infty]{\mathbb{P}}c_3.$$
\item[(iii)] The first-order inclusion probabilities $(\pi_k)_{k\in U_v}$ are such that there exist constant positive definite matrices $\boldsymbol{C}_1$ and $\boldsymbol{C}_2$ such that $$\dfrac{1}{N_v}\sum_{k \in U_v}r_k\pi_k\bx_{k,\alpha^{\star}}\bx_{k,\alpha^{\star}}^{\top}\xrightarrow[v \to \infty]{\P} \boldsymbol{C}_1, \qquad \dfrac{1}{N_v}\sum_{k\in U_v} r_k \pi_k^2 \bx_{k,\alpha^{\star}} \bx_{k,\alpha^{\star}}^\top \xrightarrow[v \to \infty]{\P} \boldsymbol{C}_2.$$ 
\end{enumerate}
Then, 
\begin{equation*}
\dfrac{\widehat{\mu}_{\widehat{\alpha},v}-\mu_v}{\sqrt{\widehat{V}_{T,v}(\widehat{\alpha}_v)}}\xrightarrow[v\to \infty]{\cL} \mathcal{N}(0,1).
\end{equation*}
\end{theorem}
\begin{proof}
See Appendix \ref{proof_Theo6}.
\end{proof}
~\\
The regularity conditions (i) - (iii) in Theorem \ref{Theo6} are standard and fairly weak. Condition (i) means that the Horvitz-Thompson estimator with values $(\eta_{k, \alpha^\star})_{k\in U_v}$ is asymptotically normal. The asymptotic normality of Horvitz-Thompson estimators has been established for some commonly used designs: \cite{hajek1960} for simple random sampling without replacement, \cite{hajek1964} for conditional Poisson sampling, \cite{bickel1984} for stratified sampling, and \cite{krewski1981} for probability-proportional-to-size cluster sampling with replacement. For instance, in case of simple random sampling without replacement, is is enough to show that, (see, e.g. \cite{thompson1997theory}, page 59) a.s., $\E_{mpq}[\eta_{k, \alpha}^{2+\delta}]<\infty$ for some $\delta>0$. In our setting, this condition can be verified to hold. The two other conditions can also be shown to hold in simple random sampling and other common sampling designs. 

~\\
The proposed methodology is summarized in Algorithm~\ref{algo1}.
The central idea is that when the imputation model is selected using a consistent model selection criterion, the effect of model uncertainty becomes asymptotically negligible. More precisely, the selected model coincides with the true model with probability tending to one, so that standard inference procedures applied conditionally on the selected model remain asymptotically valid.
Consequently, point and variance estimation may be carried out as if the selected model were known in advance. This greatly simplifies inference after model selection: asymptotically valid confidence intervals can be obtained without the need for post-selection corrections or more elaborate resampling schemes, while retaining full efficiency.

~\\
Combining our previous results yields the following property for the confidence intervals obtained by Algorithm \ref{algo1}.
\begin{flushleft}
\begin{algorithm}[htbp]
\caption{Variable selection and inference for linear regression imputation} \label{algo1}
\DontPrintSemicolon
\vspace{0.4em}
\flushleft{\textbf{Input:} The observed data $\{(\bx_k, y_k, \pi_k): k \in S_r\} 
\cup \{(\bx_k, \pi_k): k \in S_m\}$, a 
family of candidates $\cA$, and a consistent model selection criterion $\cC_n$.\;}

\vspace{0.8em}
\textbf{Step 1. Model selection.}\;
Select the covariate set
\[
\widehat{\alpha} = \arg\min_{\alpha \in \cA} \cC_n(\alpha),
\]
where $\cC_n$ is a consistent selection criterion.\;

\vspace{0.8em}
\textbf{Step 2. Point estimation.}\;
Compute the imputed estimator $\widehat{\mu}_{\widehat{\alpha}}$ using the model $\widehat{\alpha}$ selected in Step~1.\;

\vspace{0.8em}
\textbf{Step 3. Variance estimation.}\;
Estimate the total variance $\widehat{V}_T(\widehat{\alpha})$ using the variance estimators corresponding to the model $\widehat{\alpha}$.\;

\vspace{0.8em}
\textbf{Step 4. Confidence interval.}\;
Construct the $(1-2\alpha)$ interval
\[
\mathrm{CI}_{1-\alpha}(\widehat{\mu}_{\widehat{\alpha}})
= \bigg[\widehat{\mu}_{\widehat{\alpha}} - z_{1-\alpha/2}\sqrt{\widehat{V}_T(\widehat{\alpha})} \ , \ \;
\widehat{\mu}_{\widehat{\alpha}} + z_{1-\alpha/2}\sqrt{\widehat{V}_T(\widehat{\alpha})}\bigg],
\]
where $z_{1-\alpha/2}$ denotes the $1-\alpha/2$-quantile of $\cN(0,1)$.\;

\vspace{0.8em}
\textbf{Output:} The $100(1-\alpha)\%$ confidence interval $\mathrm{CI}_{1-\alpha}(\widehat{\mu}_{\widehat{\alpha}})$ for the finite population mean $\mu$.
\vspace{0.4em}

\end{algorithm}
\end{flushleft}

\begin{corollary} \label{coro3}
Consider a sequence of confidence intervals $\left(\mathrm{CI}_{1-\alpha}(\widehat{\mu}_{\widehat{\alpha}}) \right)_{v\in \N}$ obtained via Algorithm \ref{algo1}. Assume the conditions of Theorem \ref{Theo6}. Then, 
\begin{equation*}
\lim_{v \to \infty} \P_{mpq} \left( \mu_v \in \mathrm{CI}_{1-\alpha}(\widehat{\mu}_{\widehat{\alpha}_v})  \right)  = 1-\alpha
\end{equation*}
almost surely. 
Moreover, the confidence intervals are asymptotically optimal, attaining the minimum width achievable under any competing model.\end{corollary}
~\\
The asymptotic coverage follows immediately from Theorem \ref{Theo6}. The optimality follows from the fact that $\widehat{\mu}_{\widehat{\alpha}}$ has the same asymptotic variance as $\widehat{\mu}_{\alpha^\star}$, which by Proposition \ref{prop3} is the lowest since $M$ is non-decreasing. 
\section{Simulation studies}\label{sec:simu}
In this section, we present the results of simulation studies evaluating the efficiency of the proposed methodology. We first examine the behavior of the loss function $\cL_n(\alpha)$ across different model specifications, the efficiency of point estimators based on various model selection criteria, and then assess the performance of their variance estimators, including their ability to achieve the desired asymptotic coverage for the confidence intervals obtained from Algorithm~\ref{algo1}.

~\\
In our simulation experiments, we adopted a finite-population framework. 
Specifically, for each scenario, we generated multiple finite populations of large size $N$ from the assumed superpopulation model. From each generated population, a sample of size $n$ was selected according to the specified sampling design.  Item nonresponse was then generated within each selected sample according to the prescribed response mechanism, after which the proposed imputation procedure was applied. This entire process---population generation, sampling, nonresponse generation, and imputation---was repeated $20{,}000$ times to evaluate the finite-sample performance of the proposed methodology.

\subsection{Simulation set-up}
We generated finite populations of sizes $N=(1000,2000,5000)$ and $p = 20$ independent covariates from a Gamma distribution with shape parameter $5$ and scale parameter $2$. The survey variable was generated according to $$y_k = \bx_k^\top \pmb{\beta}+\epsilon_k, \qquad k\in U,$$ with 
$\pmb{\beta}=[10,9,9,8,8,7,\boldsymbol{0}_{14}^\top]^\top$ and $\epsilon_k\thicksim\mathcal{N}(0,3600)$ for $k \in U$. The first six covariates, therefore, correspond to signal variables, while the remaining fourteen covariates are noise variables with no effect on the outcome. 

~\\
We considered three sample sizes $n=(100,200,500)$ under the following designs:
\begin{enumerate}
\item[(i)] Simple random sampling without replacement.
\item[(ii)] Stratified sampling: the population was first sorted by $-(3x_1+2x_2+4x_3+5x_4)$. It was then partitioned into $H=4$ strata defined sequentially from the ordered list, containing respectively $50\%$, $25\%$, $20\%$, and $5\%$ of the population units. In each stratum, the sample $S_h$ was selected by simple random sampling without replacement of size $n_h$ based on $x_2$-optimal allocation.
\end{enumerate}
These sample sizes were chosen to correspond to the population sizes $N=(1000,2000,5000)$, yielding a constant sampling fraction of $10\%$. 
In other words, both the population size and the sample size increase at the same rate across scenarios while maintaining a fixed sampling fraction. 

~\\
We choose an embedded collection of models $\mathcal{A}$ consisting of the following 20 models $$\underbrace{\{1\}}_{:=\alpha_1}\subseteq \underbrace{\{1,2\}}_{:=\alpha_2}\subseteq \dots \subseteq \underbrace{\{1,2,\dots,20\}}_{:=\alpha_{20}}.$$ 
Although we do not write it explicitly, note that the intercept is included in every model. 

~\\
Response indicators were generated with $$\mathrm{logit}(p(\bx_k))=0.1\times\left(-70+\bx_k^{\top}\pmb{\zeta}\right), \qquad k\in U,$$ with
$\pmb{\zeta}=\left[1,1,1,1,0,0,1,1,1, \boldsymbol{0}_{11}^\top\right]^{\top}$. This led to a response rate of approximately $50\%$.

\subsection{Behavior of $\mathcal{L}_n(\alpha)$ with different models}
We first investigated the behavior of $\mathcal{L}_n(\alpha)$ for $\alpha \in \mathcal{A}$.  Here, we present only the results for the case $N=5000$ and $n=500$. 
The results obtained for the other combinations of $N$ and $n$ were very similar and are therefore not reported here for brevity. Given that the loss function $\mathcal{L}_n(\alpha)$ is unknown in practice, we estimated it by a Monte-Carlo approximation with $B = 20, \ 000$ iterations. 

~\\
As a measure of bias of a point estimator, we used the Monte-Carlo Relative Bias (RB) defined by
\begin{equation*}
\mathrm{RB}(\widehat{\mu})=100\%\times \frac{1}{B}\sum_{b=1}^B\frac{\widehat{\mu}^{(b)}-\mu^{(b)}}{\mu^{(b)}},
\end{equation*}
where $\widehat{\mu}^{(b)}$ denotes an arbitrary estimator and $\mu^{(b)}$ denotes the finite population mean at iteration $b$, respectively. As a measure of efficiency, we computed the Monte-Carlo Relative Efficiency (RE) defined by
\begin{equation*}
\mathrm{RE}(\widehat{\mu})=100\%\times \frac{\sum_{b=1}^B(\widehat{\mu}^{(b)}-\mu^{(b)})^2}{\sum_{b=1}^B(\widehat{\mu}_{\pi}^{(b)}-\mu^{(b)})^2},
\end{equation*}
where $\widehat{\mu}_{\pi}^{(b)}$ denotes the HT estimator in the $b$-th replication. The results of the simulation are presented in Table \ref{tab1}.
\begin{table}[h]\label{Table3}
\centering
\caption{Loss $\mathcal{L}_n(\widehat{\mu}_{\alpha})$, relative biases and relative efficiencies across models for sample size $n=500$.} \label{tab1}
\begin{tabular}{l|ccc|ccc}
\toprule
Model  & $\mathcal{L}_n(\alpha)$ & RB & RE & $\mathcal{L}(\alpha)$ & RB & RE\\
& \multicolumn{3}{c}{SRSWOR} &\multicolumn{3}{c}{Stratified}\\
\midrule
$\alpha_1$ &  491.3 & 4.3 &2267.1  & 587.5 & 4.7 & 3940.9\\
$\alpha_2$ & 235.5  & 2.9 & 1139.4 & 294.8  & 3.2 & 2023.6 \\
$\alpha_3$& 63.6  & 1.3 & 381.3 & 89.2 & 1.7  & 677.3 \\
$\alpha_4$& 15.5   &0.1& 166.6  & 16.3 &0.0& 202.5\\
$\alpha_5$& 12.1 & 0.0  & 152.2 &  12.7&  0.0 & 180.1 \\
$\alpha_6$& 9.7 & 0.0  & 142.2  &  10.0 & 0.0 & 162.6\\
$\alpha_7$& 10.3 & 0.0   & 144.5  & 10.5 & 0.0& 165.9 \\
$\alpha_8$& 11.0 & 0.0   & 147.7   & 11.3 & 0.0 &170.7\\
$\alpha_9$& 11.7  & 0.0  & 150.7  & 11.9 & 0.0 & 174.9\\
$\alpha_{10}$& 11.8 & 0.0 &151.0 & 12.0 & 0.0  & 175.3\\
$\alpha_{11}$& 11.9 & 0.0 &151.1 &  12.0 & 0.0  & 175.5 \\
$\alpha_{12}$& 11.9 & 0.0 & 151.3 & 12.1 & 0.0 & 175.9 \\
$\alpha_{13}$& 11.9 & 0.0  &151.5 & 12.1 & 0.0 & 176.2 \\
$\alpha_{14}$& 12.0 & 0.0  &151.7 & 12.2 & 0.0 & 176.6 \\
$\alpha_{15}$& 12.0 & 0.0   &151.9 &12.2 & 0.0 & 177.0 \\
$\alpha_{16}$& 12.1 & 0.0  & 152.3 &12.3 &  0.0& 177.4 \\
$\alpha_{17}$& 12.1 & 0.0   & 152.5& 12.4 & 0.0 & 177.7 \\
$\alpha_{18}$& 12.2 & 0.0   & 152.6 & 12.4  & 0.0 & 178.1\\
$\alpha_{19}$& 12.3  &0.0  & 153.0 & 12.4& 0.0 & 178.5\\
$\alpha_{20}$& 12.3 &0.0 &153.2 & 12.5 & 0.0 &178.8\\
\bottomrule

\end{tabular}

\end{table}
We start by observing that the rankings based on our proposed loss $\cL_n$ and on the relative efficiency (RE) were perfectly aligned: under both sampling designs, the $20$ models were ranked in exactly the same order. This confirms that the loss $\cL_n$ behaves as expected and that the model minimizing $\cL_n$, which we refer to as the optimal imputation model, is indeed the best choice in practice.

~\\
Next, recall that the models labeled $\alpha_1$–$\alpha_5$ are misspecified,  as they fail to include at least one of the important predictors among $X_2,\ldots,X_6$. The variables $X_1$ to $X_4$ are also correlated with the missingness mechanism. 
Consequently, since these variables are associated with the response mechanism, omitting them from the imputation model may induce a bias that does not vanish asymptotically, leading to an inconsistent estimator, as indicated by Corollary~\ref{Cor1}. In contrast, once $X_1$, $X_2$, $X_3$, and $X_4$ are included in the imputation model, no asymptotic bias is expected. To illustrate this, consider for example model $\alpha_4$, for which $\alpha_{\mathrm{mis}} = \{5,6\}$. Since all covariates are independent, Condition (b) of Corollary \ref{Cor1} is satisfied: $X_5$ and $X_6$ do not explain $p(\bx)$ conditional on $\bx_{\alpha_4}$. Moreover, 
$$\mathbb{E}[X_5 \mid X_1, X_2, X_3, X_4] = \mathbb{E}[X_5] 
\quad \text{and} \quad 
\mathbb{E}[X_6 \mid  X_1, X_2, X_3, X_4] = \mathbb{E}[X_6]$$
are constants. Because the intercept is included in the model, both conditions of Corollary \ref{Cor1} hold for $\alpha_4$. The simulations clearly reflect this behavior: models $\alpha_1$ to $\alpha_3$ were biased, and exhibited poor efficiency with values of RE ranging from 381.3 to 2267.1 for simple random sampling without replacement.  As soon as $X_4$ was added to $\alpha_3$, the bias vanished. However, the model $\alpha_4$ was not the most efficient. Indeed, the true model, $\alpha_6$, was the most efficient, thereby illustrating Theorem \ref{Theo1}. 

~\\
The model $\alpha_6$ is the true model. As expected, it exhibited negligible bias, the smallest $\mathcal{L}_n(\alpha)$, and the highest efficiency. This is consistent with Theorem \ref{Theo1}, which suggests that the true model should also be the optimal model for imputation, as observed here. Recall that the models $\alpha_7-\alpha_{20}$ belong to the set of correct models, and thus their biases are also negligible. Moreover, $\mathcal{L}_n(\alpha_j)\geq \mathcal{L}_n(\alpha_6)$ for $j=7,\dots,20$, as explained by Proposition \ref{prop2}. While $X_7-X_9$ explains $p(\bx)$, including these variables in the imputation model leads to an appreciable increase in $\mathcal{L}_n(\alpha)$ and in the relative efficiency of the resulting estimator. Furthermore, Proposition \ref{prop3} implies that the asymptotic variance of the imputed estimator based on $\alpha_7-\alpha_9$ is larger than that based on $\alpha_6$, resulting in lower efficiency for these models. In contrast, $X_{10}-X_{20}$ do not explain $p(\bx)$, conditional on $\bx_{\alpha_9}$. Moreover, $\mathbb{E}_\bx[X_j|\bx_{\alpha_9}]=\mathbb{E}_\bx[X_j]$ for $j=10,\dots,20,$ which are constants. Because the intercept is included in the model, both conditions of Corollary \ref{Cor2} hold for $\alpha_9$. Hence, the efficiency of the imputed estimators based on $\alpha_{10}-\alpha_{20}$ is not substantially different from that obtained under  $\alpha_9$.
\subsection{Point estimation with model selection criteria}

We now study the behavior of imputed estimators based on models selected by commonly used model selection criteria such as AIC, BIC, and $K$-folds cross-validation with $k=5$. 

~\\
To further investigate the model selection capabilities of each of the criteria, we computed the Monte-Carlo identification probability defined by
\begin{equation*}
\mathbb{P}_{MC}\left(\widehat{\alpha}=\{\alpha^{\star}\}\right)=\frac{1}{B}\sum_{b=1}^B \mathds{1}\left(\widehat{\alpha}^{(b)}=\alpha^{\star}\right),
\end{equation*}
where $\widehat{\alpha}$ denotes a model selection criterion. We present the results in 
Table \ref{tab:table2}.\\

\begin{table}[h]
\centering
\caption{Performance of imputation under different model selection criteria across
different sample sizes. \\
Note: we used the notation $\mathcal{C}^-=\mathcal{C}-\alpha^{\star}$ to denote the set of overfitted models, i.e., models that contain all true covariates, and additional superfluous ones.}
\begin{tabular}{ll|ccccc|ccccc}
\toprule
Sample size& Criteria& RB&RE& $\mathcal{W}$& $\alpha^{\star}$&$\mathcal{C}^{-}$&RB&RE& $\mathcal{W}$& $\alpha^{\star}$&$\mathcal{C}^{-}$\\
& & \multicolumn{5}{c}{SRSWOR}&\multicolumn{5}{c}{Stratified}\\
\midrule
\multirow{3}{*}{$n=100$} & AIC &0.0&164.1&1.6&49.9&48.5&0.0&193.8&1.5&52.2&46.3\\
& BIC& 0.0&151.5&7.5&82.0&10.5&0.0&176.3&7.0&83.0&10.0\\
& Cross-validation&0.0&157.7&3.9&37.5&58.6&0.0&185.7&3.5&37.5&59.0\\&True model &0.0&147.9& -&- &- &0.0&170.2& -& -&-\\
\cline{1-12}
\multirow{3}{*}{$n=200$}&AIC&0.0&150.5&0.0&63.9&36.1&0.0&170.5&0.0&63.7&36.3\\
& BIC&0.0&145.0&0.5&94.5&5.0&0.0&163.6&0.3&94.9&4.8\\
& Cross-Validation&0.0&150.8&0.2&39.2&60.6&0.0&170.6&0.2&39.1&60.7\\&True model &0.0&143.6& -&- &- &0.0&162.3& -& -&-\\
\cline{1-12}
\multirow{3}{*}{$n=500$}&AIC&0.0&146.8&0.0&69.2&30.7&0.0&168.9&0.0&68.5&31.5\\
& BIC& 0.0&142.7&0.0&97.6&2.4&0.0&163.3&0.0&97.6&2.4\\
&Cross-Validation&0.0&146.9&0.0&39.3&60.7&0.0&170.1&0.0&38.5&60.5\\
&True model &0.0&142.2& -&- &- &0.0&162.6& -& -&-\\
\bottomrule

\end{tabular}

\label{tab:table2}
\end{table}
~\\
Across all sample sizes, the imputed estimators based on the AIC, BIC, and cross-validation criteria showed negligible bias for both sampling designs. This is explained by the fact that these three model selection criteria are known to satisfy $\P(\widehat{\alpha}_v \in \cW) \xrightarrow[]{} 0$ as $v \to \infty$ at the population level \citep{zhang1993model, shao1997asymptotic}, and hence  also at the sample level by an application of Lemma \ref{Lemma2}. This behavior was confirmed in our simulations, where the probability of selecting a wrong model converged to $0$ for all three criteria. Consequently, the imputed estimators based on these selected models are consistent. However, AIC and cross-validation are not consistent model selection procedures and exhibited overfitting probabilities of $30.7\%$ and $60.7\%$, respectively, for $n =500$ in our simulations. Since the imputed estimators based on these models are consistent, the lower efficiency indicates a larger variance. For $n=500$, the imputed estimators based on AIC and cross-validation were less efficient, with RE values of $146.8\%$ and $146.9\%$, respectively, compared with $142.2\%$ for the true model. As explained in part~ii) of Proposition~\ref{prop3}, the true model achieves the smallest asymptotic variance, whereas the variance of an imputed estimator based on an overfitted model is necessarily larger.
In contrast, BIC remained consistent with the probability of selecting the true model reaching $97.6\%$ for $n=500$. As a result, the point estimator based on BIC was the most efficient in all cases. Moreover, as established in Theorem~\ref{Theo3}, the imputed estimator based on BIC is asymptotically equivalent to the oracle imputation estimator based on the true model.

\subsection{Variance estimation and confidence intervals}

We now turn to the problem of variance estimation and confidence intervals, as per the procedure described in Algorithm \ref{algo1}. More specifically, we were interested in the relative bias of $\widehat{V}_T(\widehat{\alpha})$ as an estimator of $V_T(\widehat{\alpha})$, with $\widehat{\alpha}$ denoting a model selected via a consistent model selection procedure; here, the BIC criterion was adopted. We computed
\begin{equation*}
\mathrm{RB}(\widehat{V}_T(\widehat{\alpha}))=100\%\times \frac{1}{B}\sum_{b=1}^B\frac{\widehat{V}_T^{(b)}(\widehat{\alpha})-\mathbb{V}_{MC}(\widehat{\mu}_{\widehat{\alpha}})}{\mathbb{V}_{MC}(\widehat{\mu}_{\widehat{\alpha}})},
\end{equation*}
with $$\mathbb{V}_{MC}(\widehat{\mu}_{\widehat{\alpha}}) = \dfrac{1}{B-1}\sum_{b=1}^B \left( \widehat{\mu}_{\widehat{\alpha}}^{(b)} - \dfrac{1}{m}\sum_{m=1}^m \widehat{\mu}_{\widehat{\alpha}}^{(b)} \right)^2$$ denoting the Monte-Carlo variance of $\widehat{\mu}_{\widehat{\alpha}}$. We also were interested in verifying empirically that $\P \left( \mu_v \in \mathrm{CI}_{1-\alpha}(\widehat{\mu}_{\widehat{\alpha}})  \right) \approx 1-\alpha$, as should be asymptotically by Corollary \ref{coro3}. We computed
\begin{equation*}
\mathrm{CP}(\widehat{\mu}_{\widehat{\alpha}})=100\%\times \frac{1}{B}\sum_{b=1}^B \mathds{1}\left(\mu^{(b)} \in \mathrm{CI}_{1-\alpha}^{(b)}(\widehat{\mu}_{\widehat{\alpha}}^{(b)}\right),
\end{equation*}
where $\mathrm{CI}_{1-\alpha}^{(b)}(\widehat{\mu}_{\widehat{\alpha}}^{(b)})$ denotes the output of Algorithm \ref{algo1} at iteration $b$. The results are presented in Table \ref{tab:table3}.\\

\begin{table}[h]
\centering
\caption{Relative biases of variance estimators and coverage probabilities evolution as $n$ increases with various sampling fractions.}
\label{tab:table3}
\begin{tabular}{l|cccccc|cccccc}
\toprule
& \multicolumn{6}{c}{SRSWOR}&\multicolumn{6}{c}{Stratified}\\
& \multicolumn{2}{c}{$5\%$}& \multicolumn{2}{c}{$10\%$}&\multicolumn{2}{c}{$20\%$}& \multicolumn{2}{c}{$5\%$}&\multicolumn{2}{c}{$10\%$}&\multicolumn{2}{c}{$20\%$}\\
&RB&CP& RB&CP&RB&CP&RB&CP&RB&CP&RB&CP\\
\midrule
$n=200$ &-2.9& 94.5&-4.9&94.2&-8.0&93.8&-4.9&94.1&-5.5&94.1&-10.6&93.5 \\
$n=500$ &-3.1&94.6&-3.0&94.6&-5.8&94.2&-4.0&94.5&-4.0&94.3&-8.2&93.8\\
$n=1 000$&-3.2&94.5&-1.9&94.8&-4.6&94.3&-1.8&94.7&-3.2&94.7&-4.5&94.3\\
\bottomrule
\end{tabular}

\end{table}
~\\
For smaller sample sizes, the relative biases were slightly larger as the sampling fraction increased. Consequently, the coverage probabilities were slightly below the nominal $95\%$ level, particularly when $n=200$. However, for a fixed sampling fraction, as both the sample and population sizes increased, the negative biases diminished and eventually became negligible. This pattern was consistent across both sampling designs. 

~\\
These results confirmed the asymptotic validity of the proposed variance estimation procedure: as the sample size increased, the bias of the variance estimator vanished, and the empirical coverage probabilities converged to their nominal levels. Hence, Algorithm~\ref{algo1} produced reliable confidence intervals in large samples under both simple random sampling without replacement and stratified sampling.

\section{Final remarks}\label{sec:finalrem}

This paper develops a theoretical framework for model selection in the context of imputation under survey sampling. We introduced an oracle loss function that quantifies the efficiency of an imputation model and showed that its minimizer asymptotically coincides with the true model. Under standard regularity conditions and a non-informative sampling design, model selection procedures that are consistent in the i.i.d. setting remain consistent when applied to survey data. In particular, the BIC criterion asymptotically identifies the true imputation model.

~\\
Based on this framework, we established the asymptotic properties of the imputed estimator and its variance when the imputation model is selected using a consistent criterion. The resulting estimators are asymptotically equivalent to those obtained under the true model, therefore achieving oracle efficiency. These results provide a theoretical justification for the use of standard model-selection tools in imputation problems involving survey data.

~\\
Simulation results support the theoretical findings. The loss function $\cL_n$ discriminates effectively between models, with bias vanishing once all relevant predictors are included. The BIC criterion consistently identifies the true model, whereas AIC and cross-validation tend to favor overly complex specifications. The proposed variance estimator performs well in finite samples, exhibiting negligible bias and empirical coverage close to the nominal level.

~\\
The proposed framework establishes a rigorous connection between model selection and imputation in survey data. It shows that, under appropriate conditions, selecting an imputation model with a consistent criterion yields asymptotically valid inference. It would be of interest to extend these results to broader classes of imputation models or to settings with informative sampling designs. A first step in this direction could be to leverage the favorable model-selection properties of cross-validation in nonparametric regression \citep{yang2007consistency}. Another related promising avenue would be to investigate the use of aggregation \citep{nemirovski2000topics, bunea2007aggregation} to combine several models. This was tested empirically for the treatment of unit nonresponse with good results in \cite{Larbi2025}, but a theoretical investigation of the topic is currently lacking.
\appendix
\section{Additional notation}
The Euclidean vector norm is denoted $\rVert\cdot \rVert_2$. The operator and Frobenius norms of a matrix are denoted by $\rVert \cdot\rVert_{op}$ and $\rVert \cdot\rVert_{F}$, respectively. The largest and smallest eigenvalues of a symmetric matrix $\boldsymbol{S}$ are denoted $\lambda_{min} \left(\boldsymbol{S}\right)$ and $\lambda_{max} \left(\boldsymbol{S}\right)$, respectively. We write $\text{Tr}\left(\boldsymbol{S}\right)$ to denote the trace of a square matrix $\boldsymbol{S}$. We use $\boldsymbol{S} \succeq 0$ (resp, $\boldsymbol{S} \succ 0$) to denote that the symmetric matrix  $\boldsymbol{S}$ is a positive semi-definite (resp, positive definite) and write $\boldsymbol{S}_1 \succeq \boldsymbol{S}_2$ to mean that $\boldsymbol{S}_1 -\boldsymbol{S}_2\succeq 0$. The identity matrix of $\mathbb{R}^{n\times n}$ is denoted $\boldsymbol{I}_n$.   \\
Here, $z_n=o_{\P}(1)$ indicates that for any $\epsilon>0$, we have
\begin{equation*}
    \lim_{v \rightarrow \infty}\P_{ mpq}\left( \rvert z_n \rvert>\epsilon\right)=0
\end{equation*}
almost surely. We may omit the phrase "almost surely" or "a.s." in some intermediate steps of the proof. To express that $A\cup B$ is a disjoint union, we use $A \uplus B$. We use $x_n \sim_p z_n$ to say that $x_n - z_n = o_{\mathbb{P}}(1)$. To say that two random sequences $(x_n)_{n\in\N}$ and $(z_n)_{n\in\N}$ are of the same order in probability, we write $x_n \asymp_\P z_n$.\\
The matrix $\bX_r^\top \bX_r$ is denoted by $\boldsymbol{A}_r:=\bX_r^\top \bX_r$. Similarly, for $\alpha\in \mathcal{A}$, we write $\boldsymbol{A}_{r,\alpha}:=\bX_{r,\alpha}^{\top}\bX_{r,\alpha}$.

\section{Proof of main results}

\subsection{Invertibility with high-probability}

\begin{lemma}\label{HighProb} 
  Assume \ref{S1}-\ref{S2}, \ref{D1}-\ref{D2}. Let $$\boldsymbol{A}_v^w := \dfrac{1}{N_v}\sum_{k\in U_v} w_k r_k \bx_k \bx_k^\top,$$ for a set of $\sigma(\bX)$-measurable weights $(w_k)_{k\in U_v}$ satisfying that, there exists $\nu>0$ such that, for all $v \in \N$,  $w_k \geqslant \nu$, for all $k \in U_v$.  
  Then, there exists constants $C_1, C_2, C_3\in \R^*_+$ and an index $v_0 := v_0\left((\bx_k)_{k\in U}\right) \in \N$ such that, for all $v \geqslant v_0$, almost surely,  \begin{equation}\label{eq:high-dim}
        \P_q \left( \lambda_{\min}( \boldsymbol{A}_v )<C_3\right) \leq C_1 \exp \left( -C_2N_v\right).
    \end{equation}
\end{lemma}
\begin{proof}
Since $w_k \geqslant \nu$ uniformly, we have $$\boldsymbol{A}_v^w  \succeq \nu \boldsymbol{A}_v, \quad \text{with}\quad  \boldsymbol{A}_v:= \dfrac{1}{N}\sum_{k\in U_v} r_k \bx_k \bx_k^\top.$$ In particular, this ensures $$\lambda_{\min}( \boldsymbol{A}_v^w ) \geq \nu \lambda_{\min}( \boldsymbol{A}_v ),$$ so we shall prove the inequality for $\boldsymbol{A}_v$ instead. Note that $$ \E_q \left[\boldsymbol{A}_v\right] = \dfrac{1}{N}\sum_{k\in U_v} p(\bx_k)\bx_k \bx_k^\top := \boldsymbol{A}_v^p \xrightarrow[v \to \infty]{a.s.}\E_\bx \left[ p(\bx)\bx\bx^\top\right]:=\boldsymbol{A}^{(\infty)}.$$ 
In our setting, $\boldsymbol{A}^{(\infty)}\succ \boldsymbol{0}$ so that $\sigma^{(\infty)} := \lambda_{\min}( \boldsymbol{A}^{(\infty)}) >0.$ Since the convergence to $\boldsymbol{A}^{(\infty)}$ is almost sure, there exists an event $E$ such that $\P_\bx(E)=1$, on which for all $\omega \in E$, there exists $v_0(\omega) $ such that for all $v \geq v_0,$ $$ \lambda_{\min}( \boldsymbol{A}_v) \geq \dfrac{\sigma^{(\infty)} }{2}. $$ We now fix $\omega\in E$, and $v \geq v_0$. Let $$\boldsymbol{Z}_k := N_v^{-1} r_k \bx_k \bx_k^\top, \qquad k \in U_v,$$ and note that $(\boldsymbol{Z}_k)_{k\in U_v}$ are (i) positive semi-definite, (ii) of the operator norm uniformly bounded $\rVert \boldsymbol{Z}_k\rVert_{op} \leq C_0^2/N$ with $C_0$ satisfying $\rVert\bx_k\rVert_2 \leq C_0$. In particular, $$\mu_{\min} := \lambda_{\min} \left( \E_q \left[\sum_{k\in U_v} \boldsymbol{Z}_k \right]\right) = \lambda_{\min} \left(\boldsymbol{A}^p_v \right) \geq \dfrac{\sigma^{(\infty)}}{2}.$$
An application of Matrix Chernoff's inequality (see, e.g., Theorem 5.1.1 in \cite{tropp2015introduction}) gives $$\P_q \left( \lambda_{\min}\left( \sum_{k\in U_v} \boldsymbol{Z}_k \right) \leq \dfrac{\mu_{\min}}{2} \right) \leq p \left( \sqrt{\dfrac{2}{e}}\right)^{\dfrac{\mu_{\min}N_v}{C_0^2}}\leq p \left( \sqrt{\dfrac{2}{e}}\right)^{\dfrac{\sigma^{(\infty)}N_v}{2C_0^2}}.$$ Noting that $2/e<1$, we write 
$$ p \left( \sqrt{\dfrac{2}{e}}\right)^{\dfrac{\sigma^{(\infty)}N_v}{2C_0^2}} = C_1\exp \left(-C_2N_v \right), \quad \text{with} \qquad C_1 := p, \quad C_2 := -\dfrac{\sigma^{(\infty)}}{2C_0^2}\ln \left( \sqrt{\dfrac{2}{e}}\right)>0.$$
Overall, since $$\left\{ \lambda_{\min}\left( \sum_{k\in U_v} \boldsymbol{Z}_k \right) \leq \dfrac{\sigma^{(\infty)}}{4} \right\}  \subset \left\{ \lambda_{\min}\left( \sum_{k\in U_v} \boldsymbol{Z}_k \right) \leq \dfrac{\mu_{\min}}{2} \right\} ,$$ we get that, almost surely, for all $v \geqslant v_0$,
$$\P_q \left( \lambda_{\min}\left( \sum_{k\in U_v} \boldsymbol{Z}_k \right) \leq \dfrac{\sigma^{(\infty)}}{4} \right) \leq C_1 \exp\left(-C_2N_v\right)$$ which also translates to its weighted version. 
\end{proof}

\subsection*{Proof of Proposition \ref{prop1}.} \label{proof_prop1}
In the remainder of the proof, we fix an arbitrary $\alpha \in \cA$. Observe that 
\begin{align*}
\cL_n \left(\alpha\right)  &= \mathbb{E}_m \left[ \sum_{k \in S_m} \dfrac{\left(y_k-\bx_{k,\alpha}^\top \widehat{\pmb{\beta}}_\alpha\right)^2}{\pi_k^2} \right] +  \mathbb{E}_m \left[ \sum_{k \in S_m}\sum_{\substack{l \in S_m\\ l\neq k}} \dfrac{y_k-\bx_{k,\alpha}^\top \widehat{\pmb{\beta}}_\alpha}{\pi_k}\dfrac{y_l-\bx_{l,\alpha}^\top \widehat{\pmb{\beta}}_\alpha}{\pi_l} \right]\\
&:= A_n\left(\alpha\right) + B_n\left(\alpha\right).
\end{align*}
We compute each term separately. 

\paragraph{Computation of $A_n\left(\alpha\right)$.}
By expanding $y_k = \bx_{k,\alpha}^\top \bbeta + \epsilon_k$ for $k \in S_m$, we obtain
\begin{align*}
A_n\left(\alpha\right)&=\mathbb{E}_m\left[\sum_{k \in S_m}\dfrac{\left(\bx_k^{\top}\pmb{\beta}-\bx_{k,\alpha}^{\top}\widehat{\pmb{\beta}}_{\alpha}\right)^2}{\pi_k^2}\right]+\mathbb{E}_m\left[\sum_{k \in S_m}\dfrac{\epsilon_k^2}{\pi_k^2}\right]
+2\mathbb{E}_m\left[\sum_{k \in S_m}\dfrac{\epsilon_k}{\pi_k^2}\left(\bx_k^{\top}\pmb{\beta}-\bx_{k,\alpha}^{\top}\widehat{\pmb{\beta}}_{\alpha} \right) \right]  \\
&=\mathbb{E}_m\left[\sum_{k \in S_m}\dfrac{\bigg(\bx_k^{\top}\pmb{\beta}-\bx_{k,\alpha}^{\top}\widehat{\pmb{\beta}}_{\alpha}\bigg)^2}{\pi_k^2}\right]+\sum_{k \in S_m}\dfrac{\sigma^2}{\pi_k^2}.
\end{align*}
The first term can be rewritten as
\begin{align*}
&\mathbb{E}_m\left[\sum_{k \in S_m}\dfrac{\left(\bx_k^{\top}\pmb{\beta}-\bx_{k,\alpha}^{\top}\widehat{\pmb{\beta}}_{\alpha}\right)^2}{\pi_k^2}\right] \\&=\mathbb{E}_m\left[\sum_{k \in S_m}\dfrac{\left(\bx_k^{\top}\pmb{\beta}-\bx_{k,\alpha}^{\top}\boldsymbol{A}_{r,\alpha}^{-1}\bX_{r,\alpha}^{\top}\boldsymbol{Y}_r\right)^2}{\pi_k^2} \right]\nonumber\\
&=\mathbb{E}_m\left[\sum_{k \in S_m}\dfrac{\left( \bx_k^{\top}\pmb{\beta}-\bx_{k,\alpha}^{\top}\boldsymbol{A}_{r,\alpha}^{-1}\bX_{r,\alpha}^{\top}\bX_r\pmb{\beta}\right)^2}{\pi_k^2}  \right]\\
&\quad\quad+\sum_{k \in S_m}\dfrac{\bx_{k,\alpha}^{\top}\boldsymbol{A}_{r,\alpha}^{-1}\bX_{r,\alpha}^{\top}\mathbb{E}_m\left[\boldsymbol{\epsilon}_r \pmb{\epsilon}_r^\top \right]\bX_{r,\alpha}\boldsymbol{A}_{r, \alpha}^{-1} \bx_{k,\alpha} }{\pi_k^2} \\
&= \sum_{k \in S_m}\dfrac{\left( \bx_k^{\top}\pmb{\beta}-\bx_{k,\alpha}^{\top}\boldsymbol{A}_{r,\alpha}^{-1}\bX_{r,\alpha}^{\top}\bX_r\pmb{\beta}\right)^2}{\pi_k^2}  +\sigma^2\sum_{k \in S_m}\dfrac{\bx_{k,\alpha}^{\top}\boldsymbol{A}_{r, \alpha}^{-1} \bx_{k,\alpha} }{\pi_k^2}.
\end{align*}
Overall, 
\begin{align} \label{eq:An}
A_n\left(\alpha\right) &=\sum_{k \in S_m}\dfrac{\left( \bx_k^{\top}\pmb{\beta}-\bx_{k,\alpha}^{\top}\boldsymbol{A}_{r,\alpha}^{-1}\bX_{r,\alpha}^{\top}\bX_r\pmb{\beta}\right)^2}{\pi_k^2}  +\sigma^2 \left(\sum_{k \in S_m}\dfrac{\bx_{k,\alpha}^{\top}\boldsymbol{A}_{r, \alpha}^{-1} \bx_{k,\alpha} }{\pi_k^2}+  \sum_{k \in S_m}\dfrac{1}{\pi_k^2}\right).
\end{align}
\paragraph{Computation of $B_n\left(\alpha\right)$.} Note that
\begin{align}
&\sum_{k \in S_m}\sum_{\substack{l \in S_m \\ l\neq k}}\left(y_k-\bx_{k,\alpha}^{\top}\widehat{\pmb{\beta}}_{\alpha}\right)\left(y_l-\bx_{l,\alpha}^{\top}\widehat{\pmb{\beta}}_{\alpha}\right)\nonumber\\
&=\sum_{k \in S_m}\sum_{\substack{l \in S_m \\ l\neq k}}\left(\bx_k^{\top}\pmb{\beta}+\epsilon_k-\bx_{k,\alpha}^{\top}\widehat{\pmb{\beta}}_{\alpha}\right)\left(\bx_l^{\top}\pmb{\beta}+\epsilon_l-\bx_{l,\alpha}^{\top}\widehat{\pmb{\beta}}_{\alpha}\right)\nonumber \\
&=\sum_{k \in S_m}\sum_{\substack{l \in S_m \\ l\neq k}}\epsilon_k\epsilon_l+\sum_{k \in S_m}\sum_{\substack{l \in S_m \\ l\neq k}}\epsilon_k\left(\bx_l^{\top}\pmb{\beta}-\bx_{l,\alpha}^{\top}\widehat{\pmb{\beta}}_{\alpha}\right)+\sum_{k \in S_m}\sum_{\substack{l \in S_m \\ l\neq k}}\epsilon_l\left(\bx_k^{\top}\pmb{\beta}-\bx_{k,\alpha}^{\top}\widehat{\pmb{\beta}}_{\alpha}\right)\nonumber\\
&\quad +\sum_{k \in S_m}\sum_{\substack{l \in S_m \\ l\neq k}}\left(\bx_k^{\top}\pmb{\beta}-\bx_{k,\alpha}^{\top}\widehat{\pmb{\beta}}_{\alpha}\right)\left(\bx_l^{\top}\pmb{\beta}-\bx_{l,\alpha}^{\top}\widehat{\pmb{\beta}}_{\alpha}\right). \label{eq:fourth}
\end{align}
It follows from the above that the first three terms of \eqref{eq:fourth} have mean zero. Consider the following decomposition of the last term:
\begin{align*}
&\sum_{k \in S_m}\sum_{\substack{l \in S_m \\ l\neq k}}\dfrac{\left(\bx_k^{\top}\pmb{\beta}-\bx_{k,\alpha}^{\top}\widehat{\pmb{\beta}}_{\alpha}\right)}{\pi_k}\dfrac{\left(\bx_l^{\top}\pmb{\beta}-\bx_{l,\alpha}^{\top}\widehat{\pmb{\beta}}_{\alpha}\right)}{\pi_l}\nonumber\\
&=\sum_{k \in S_m}\sum_{\substack{l \in S_m \\ l\neq k}}\dfrac{\left(\bx_k^{\top}\pmb{\beta}-\bx_{k,\alpha}^{\top}\boldsymbol{A}_{r,\alpha}^{-1}\bX_{r,\alpha}^{\top}(\bX_r\pmb{\beta}+\pmb{\epsilon}_r)\right)}{\pi_k}\dfrac{\left(\bx_l^{\top}\pmb{\beta}-\bx_{l,\alpha}^{\top}\boldsymbol{A}_{r,\alpha}^{-1}\bX_{r,\alpha}^{\top}(\bX_r\pmb{\beta}+\pmb{\epsilon}_r)\right)}{\pi_l}\nonumber\\
&=\sum_{k \in S_m}\sum_{\substack{l \in S_m \\ l\neq k}}\dfrac{\left(\bx_k^{\top}\pmb{\beta}-\bx_{k,\alpha}^{\top}\boldsymbol{A}_{r,\alpha}^{-1}\bx_{r,\alpha}^{\top}\bX_r\pmb{\beta}\right)}{\pi_k}\dfrac{\left(\bx_l^{\top}\pmb{\beta}-\bx_{l,\alpha}^{\top}\boldsymbol{A}_{r,\alpha}^{-1}\bX_{r,\alpha}^{\top}\bX_r\pmb{\beta}\right)}{\pi_l}\nonumber\\
&\quad +2\sum_{k \in S_m}\sum_{\substack{l \in S_m \\ l\neq k}}\dfrac{\bx_{k,\alpha}^{\top}\boldsymbol{A}_{r,\alpha}^{-1}\bX_{r,\alpha}\pmb{\epsilon}_r}{\pi_k}\dfrac{\left(\bx_l^{\top}\pmb{\beta}-\bx_{l,\alpha}^{\top}\boldsymbol{A}_{r,\alpha}^{-1}\bX_{r,\alpha}^{\top}\bX_r\pmb{\beta}\right)}{\pi_l} \nonumber \\
&\quad +\sum_{k \in S_m}\sum_{\substack{l \in S_m \\ l\neq k}}\dfrac{\bx_{k,\alpha}^{\top}\boldsymbol{A}_{r,\alpha}^{-1}\bX_{r,\alpha}\pmb{\epsilon}_r}{\pi_k}\dfrac{\pmb{\epsilon}_r^\top \bX_{r, \alpha}^\top \boldsymbol{A}_{r, \alpha}^{-1} \bx_{l,\alpha}}{\pi_l}.
\end{align*}
It follows that
\begin{align}
\mathbb{E}_m& \left[ \sum_{k \in S_m}\sum_{\substack{l \in S_m\\ l\neq k}} \dfrac{y_k-\bx_{k,\alpha}^\top \widehat{\pmb{\beta}}_\alpha}{\pi_k}\dfrac{y_l-\bx_{l,\alpha}^\top \widehat{\pmb{\beta}}_\alpha}{\pi_l} \right] \nonumber\\&= \mathbb{E}_m \left[ \sum_{k \in S_m}\sum_{\substack{l \in S_m \\ l\neq k}}\dfrac{\left(\bx_k^{\top}\pmb{\beta}-\bx_{k,\alpha}^{\top}\widehat{\pmb{\beta}}_{\alpha}\right)}{\pi_k}\dfrac{\left(\bx_l^{\top}\pmb{\beta}-\bx_{l,\alpha}^{\top}\widehat{\pmb{\beta}}_{\alpha}\right)}{\pi_l} \right] \nonumber\\
&= \sum_{k \in S_m}\sum_{\substack{l \in S_m \\ l\neq k}}\dfrac{\left(\bx_k^{\top}\pmb{\beta}-\bx_{k,\alpha}^{\top}\boldsymbol{A}_{r,\alpha}^{-1}\bX_{r,\alpha}^{\top}\bX_r\pmb{\beta}\right)}{\pi_k}\dfrac{\left(\bx_l^{\top}\pmb{\beta}-\bx_{l,\alpha}^{\top}\boldsymbol{A}_{r,\alpha}^{-1}\bX_{r,\alpha}^{\top}\bX_r\pmb{\beta}\right)}{\pi_l}\nonumber \\&+ \sigma^2 \sum_{k \in S_m}\sum_{\substack{l \in S_m \\ l\neq k}}\dfrac{\bx_{k,\alpha}^{\top}\boldsymbol{A}_{r, \alpha}^{-1} \bx_{l,\alpha}}{\pi_k\pi_l}. \label{eq:Bn}
\end{align}
Combining the closed-form formulas of $A_n\left(\alpha\right)$ in \eqref{eq:An} and $B_n\left(\alpha\right)$ in \eqref{eq:Bn}, we get 
\begin{align*}
\cL_n \left(\alpha\right) &= \left\{\sum_{k \in S_m}\pi_k^{-1}\left(\bx_k^{\top}-\bx_{k,\alpha}^{\top}\boldsymbol{A}_{r,\alpha}^{-1}\bX_{r,\alpha}^{\top}\bX_r\right)\pmb{\beta}\right\}^2\nonumber \\
&\quad+ \sigma^2 \left\{\left(\sum_{k \in S_m} \dfrac{\bx_{k,\alpha}^\top}{\pi_k}\right)\boldsymbol{A}_{r, \alpha}^{-1} \left(\sum_{l\in S_m}\dfrac{\bx_{l,\alpha}}{\pi_l}\right)+\sum_{k \in S_m}\dfrac{1}{\pi_k^2} \right\}.
\end{align*}
\subsection*{Proof of Proposition \ref{prop2}.} \label{proof_prop2}

Let $\alpha_1, \alpha_2 \in \cC$ such that $\alpha_1 \subset \alpha_2$. Recall that, for any $\alpha \in \cC$, $\cL_{1,n}\left(\alpha\right) = 0.$ Thus, \begin{align*}
\cL_{n} \left(\alpha_1\right) < \cL_{n} \left(\alpha_2\right) \qquad &\Leftrightarrow \qquad   \cL_{2,n} \left(\alpha_1\right) < \cL_{2,n} \left(\alpha_2\right). 
\end{align*}
Since $\alpha^{\star}$ is the smallest correct model, the result follows using Lemma \ref{MatrixOverfitting}.

\subsection*{Proof of Theorem \ref{Theo1}.} \label{proof_Theo1}

We will show that $\lim_{v\to\infty} \mathbb{P}_{pq} \left( \alpha_{opt,v} \neq \alpha^{\star}\right)= 0$. We may write the event $\{\alpha_{opt,v} \neq \alpha^{\star}\}$ as a disjoint union $$\{\alpha_{opt,v} \neq \alpha^{\star}\} =\{\alpha_{opt,v} \in \cC, \alpha_{opt,v} \neq \alpha^{\star}\} \biguplus \\\{\alpha_{opt,v} \in \cW\},$$ so that 
\begin{equation*}
\mathbb{P}_{pq} \left(\alpha_{opt,v} \neq \alpha^{\star}\right)   =  \mathbb{P}_{pq}  \left(\alpha_{opt,v} \in \cC, \alpha_{opt,v} \neq \alpha^{\star}\right)   +  \mathbb{P}_{pq}  \left(\alpha_{opt,v} \in \cW\right) := p_{1,v} + p_{2,v}.
\end{equation*}
Using Proposition \ref{prop2}, we have directly that $p_{1,v}=0$ for all $v \in \N$.\\

Let $\alpha_F:=\left\{1,\dots,p\right\}\in\cC$ be the set of all covariates. For $p_{2,v}$, consider the following decomposition
\begin{align*}
p_{2,v}&=\mathbb{P}_{pq}\left(\mathcal{L}_n(\alpha_{opt,v})<\mathcal{L}_n(\alpha^{\star})\right)\\&=\mathbb{P}_{pq}\left(\mathcal{L}_n(\alpha_{opt,v})<\mathcal{L}_n(\alpha^{\star}), \mathcal{L}_n(\alpha^{\star}\right)\leq \mathcal{L}_n(\alpha_F))\\
&\quad+\mathbb{P}_{pq}\left(\mathcal{L}_n(\alpha_{opt,v})<\mathcal{L}_n(\alpha^{\star}),\mathcal{L}_n(\alpha^{\star})>\mathcal{L}_n(\alpha_F)\right)\\
&\leq\mathbb{P}_{pq}\left(\mathcal{L}_n(\alpha_{opt,v})<\mathcal{L}(\alpha_F)\right)+\mathbb{P}_{pq}\left(\mathcal{L}_n(\mathcal{\alpha}_F)<\mathcal{L}_n(\alpha^{\star})\right):=q_{1,v}+q_{2,v}.
\end{align*}
\paragraph{Convergence of $q_{1,v}$.} On the one hand, for $\alpha_{opt,v}\in \cW$, the loss $\mathcal{L}_{1,n}(\alpha_{opt,v})$ satisfies the following identity 
\begin{align}\label{q1s}
\mathcal{L}_{1,n}(\alpha_{opt,v})&=\left\{\sum_{k \in S_{m,v}}\pi_k^{-1}\left(\bx_k^{\top}-\bx_{k,\alpha}^{\top}\boldsymbol{A}_{r,\alpha}^{-1}\bX_{r,\alpha}^{\top}\bX_r\right)\pmb{\beta}\right\}^2\nonumber\\
&=\Bigg\{\sum_{k \in S_{m,v}}\pi_k^{-1}\Big(\bx_{k,\alpha_{opt,v}}^{\top}\pmb{\beta}_{\alpha_{opt,v}}+\bx_{k,\alpha_{opt,v}^c}^{\top}\pmb{\beta}_{\alpha_{opt,v}^c}\\&\quad-\underbrace{\bx_{k,\alpha_{opt,v}}^{\top}\boldsymbol{A}_{r,\alpha_{opt,v}}^{-1}\bX_{r,\alpha_{opt,v}}^{\top}\bX_{r,\alpha_{opt,v}}\pmb{\beta}_{\alpha_{opt,v}}}_{=\bx_{k,\alpha_{opt,v}}^{\top}\pmb{\beta}_{\alpha_{opt,v}}}\nonumber\\&\quad-\bx_{k,\alpha_{opt,v}}^{\top}\boldsymbol{A}_{r,\alpha_{opt,v}}^{-1}\bX_{r,\alpha_{opt,v}}^{\top}\bX_{r,\alpha_{opt,v}^c}\pmb{\beta}_{\alpha_{opt,v}^c}\Big)\Bigg\}^2\nonumber\\
&=\left\{\sum_{k \in S_m}\pi_k^{-1}\left(\bx_{k,\alpha_{opt,v}^c}^{\top}-\bx_{k,\alpha_{opt,v}}^{\top}\boldsymbol{A}_{r,\alpha_{opt,v}}^{-1}\bX_{r,\alpha_{opt,v}}^{\top}\bX_{r,\alpha_{opt,v}^c}\right)\pmb{\beta}_{\alpha_{opt,v}^c}\right\}^2.
\end{align}
On the other hand, since $\alpha_F\in \mathcal{C}$, we have $\mathcal{L}_{1,n}\left(\alpha_F\right)=0$. Thus, $\mathcal{L}_{n}\left(\alpha_{opt,v}\right)<\mathcal{L}_{n}\left(\alpha_F\right)$ is equivalent to
\begin{align}\label{q1v}
\mathcal{L}_n\left(\alpha_{opt,v}\right)<\mathcal{L}_n\left(\alpha_F\right) \quad &\Leftrightarrow \quad \mathcal{L}_{1,n}\left(\alpha_{opt,v}\right)+\mathcal{L}_{2,n}\left(\alpha_{opt,v}\right)<\mathcal{L}_{1,n}\left(\alpha_F\right)+\mathcal{L}_{2,n}\left(\alpha_F\right)\nonumber\\  &\Leftrightarrow \quad \mathcal{L}_{1,n}\left(\alpha_{opt,v}\right)<\mathcal{L}_{2,n}\left(\alpha_F\right)-\mathcal{L}_{2,n}\left(\alpha_{opt,v}\right).
\end{align}
Combining \eqref{q1s} and \eqref{q1v} leads to 
\begin{align*}
&\mathcal{L}_{1,n}\left(\alpha_{opt,v}\right)<\mathcal{L}_{2,n}\left(\alpha_F\right)-\mathcal{L}_{2,n}\left(\alpha_{opt,v}\right)
\\&\Leftrightarrow\left(\sum_{k \in S_{m,v}}\frac{\bx_{k,\alpha_{opt,v}^c}^{\top}}{N_v\pi_k}-\sum_{k \in S_{m,v}}\frac{\bx_{k,\alpha_{opt,v}}^{\top}}{N_v\pi_k}\boldsymbol{A}_{r,\alpha_{opt,v}}^{-1}\bx_{r,\alpha_{opt,v}}^{\top}\bx_{r,\alpha_{opt,v}^c}\right)\\&\quad\quad\times\pmb{\beta}_{\alpha_{opt,v}^c}\pmb{\beta}_{\alpha_{opt,v}^c}^{\top}\left(\sum_{k \in S_{m,v}}\frac{\bx_{k,\alpha_{opt,v}^c}^{\top}}{N_v\pi_k}-\sum_{k \in S_{m,v}}\frac{\bx_{k,\alpha_{opt,v}}^{\top}}{N_v\pi_k}\boldsymbol{A}_{r,\alpha_{opt,v}}^{-1}\bx_{r,\alpha_{opt,v}}^{\top}\bx_{r,\alpha_{opt,v}^c}\right)^{\top}\nonumber\\
&\quad< \sigma^2\left(\sum_{k \in S_{m,v}}\frac{\bx_{k,\alpha_F}^{\top}}{N_v\pi_k}\right)\boldsymbol{A}_{r,\alpha_F}^{-1}\left(\sum_{k \in S_{m,v}}\frac{\bx_{k,\alpha_F}}{N_v\pi_k}\right)\\
&\quad\quad-\sigma^2\left(\sum_{k \in S_{m,v}}\frac{\bx_{k,\alpha_{opt,v}}^{\top}}{N_v\pi_k}\right)\boldsymbol{A}_{r,\alpha_{opt,v}}^{-1}\left(\sum_{k \in S_{m,v}}\frac{\bx_{k,\alpha_{opt,v}}}{N_v\pi_k}\right)\nonumber.
\end{align*}
Using equality \eqref{prooflm2}, by setting $\alpha_1:=\alpha_{opt,v}$, $\alpha_2:=\alpha_F$,  $\alpha_{opt,v}^c:=\alpha_{opt,v}^c$, we obtain
\begin{align*}
&\sigma^2\bigg(\sum_{k \in S_{m,v}}\frac{\bx_{k,\alpha_F}^{\top}}{N_v\pi_k}\bigg)\boldsymbol{A}_{r,\alpha_F}^{-1}\bigg(\sum_{k \in S_{m,v}}\frac{\bx_{k,\alpha_F}}{N_v\pi_k}\bigg)-\sigma^2\bigg(\sum_{k \in S_{m,v}}\frac{\bx_{k,\alpha_{opt,v}}^{\top}}{N_v\pi_k}\bigg)\boldsymbol{A}_{r,\alpha_{opt,v}}^{-1}\bigg(\sum_{k \in S_{m,v}}\frac{\bx_{k,\alpha_{opt,v}}}{N_v\pi_k}\bigg)\\
&\quad=\sigma^2\bigg(\sum_{k \in S_{m,v}}\frac{\bx_{k,\alpha_{opt,v}^c}^{\top}}{N_v\pi_k}-\sum_{k \in S_{m,v}}\frac{\bx_{k,\alpha_{opt,v}}^{\top}}{N_v\pi_k}\boldsymbol{A}_{r,\alpha_{opt,v}}^{-1}\bx_{r,\alpha_{opt,v}}^{\top}\bx_{r,\alpha_{opt,v}^c}\bigg)\\&\quad\quad\times\boldsymbol{S}_v^{-1}\bigg(\sum_{k \in S_{m,v}}\frac{\bx_{k,\alpha_{opt,v}^c}^{\top}}{N_v\pi_k}-\sum_{k \in S_{m,v}}\frac{\bx_{k,\alpha_{opt,v}}^{\top}}{N_v\pi_k}\boldsymbol{A}_{r,\alpha_{opt,v}}^{-1}\bx_{r,\alpha_{opt,v}}^{\top}\bx_{r,\alpha_{opt,v}^c}\bigg)^{\top},
\end{align*}
where
\begin{align*}
\boldsymbol{S}_v&=\bx_{r,\alpha_{opt,v}^c}^{\top}\left(\boldsymbol{I}_{n_{r,v}}-\bx_{r,\alpha_{opt,v}}\boldsymbol{A}_{r,\alpha_{opt,v}}^{-1}\bx_{r,\alpha_{opt,v}}^{\top}\right)\bx_{r,\alpha_{opt,v}^c}
\end{align*}
is a positive and definite matrix.
Let $$\textbf{z}_{v}^{\top}=\sum_{k \in S_{m,v}}\frac{\bx_{k,\alpha_{opt,v}^c}^{\top}}{N_v\pi_k}-\sum_{k \in S_{m,v}}\frac{\bx_{k,\alpha_{opt,v}}^{\top}}{N_v\pi_k}\boldsymbol{A}_{r,\alpha_{opt,v}}^{-1}\bx_{r,\alpha_{opt,v}}^{\top}\bx_{r,\alpha_{opt,v}^c}.$$
Then $\mathcal{L}_{1,n}\left(\alpha_{opt,v}\right)<\mathcal{L}_{2,n}\left(\alpha_F\right)-\mathcal{L}_{2,n}\left(\alpha_{opt,v}\right)$ reduces to
\begin{equation*}
\mathcal{L}_{1,n}\left(\alpha_{opt,v}\right)<\mathcal{L}_{2,n}\left(\alpha_F\right)-\mathcal{L}_{2,n}\left(\alpha_{opt,v}\right) \quad \Leftrightarrow \quad \textbf{z}_v^{\top}\pmb{\beta}_{\alpha_{opt,v}^c}\pmb{\beta}_{\alpha_{opt,v}^c}^{\top}\textbf{z}_v<\sigma^2\textbf{z}_v^{\top}\boldsymbol{S}_v^{-1}\textbf{z}_v.
\end{equation*}
Consider the following decomposition of $q_{1,v}$,
\begin{align*}
\lim_{v \rightarrow \infty}q_{1,v}&=\lim_{v \rightarrow \infty}\mathbb{P}_{pq}\left(\textbf{z}_v^{\top}\pmb{\beta}_{\alpha_{opt,v}^c}\pmb{\beta}_{\alpha_{opt,v}^c}^{\top}\textbf{z}_v<\sigma^2\textbf{z}_v^{\top}\boldsymbol{S}_v^{-1}\textbf{z}_v,\textbf{z}_v\neq \textbf{0}\right)\\&\quad+\lim_{v \rightarrow \infty}\mathbb{P}_{pq}\left(\textbf{z}_v^{\top}\pmb{\beta}_{\alpha_{opt,v}^c}\pmb{\beta}_{\alpha_{opt,v}^c}^{\top}\textbf{z}_v<\sigma^2\textbf{z}_v^{\top}\boldsymbol{S}_v^{-1}\textbf{z}_v,\textbf{z}_v= \textbf{0}\right)
\end{align*}
If $\textbf{z}_v=\mathbf{0}$, then the event $\{\textbf{z}_v^{\top}\pmb{\beta}_{\alpha_{opt,v}^c}\pmb{\beta}_{\alpha_{opt,v}^c}^{\top}\textbf{z}_v<\sigma^2\textbf{z}_v^{\top}\boldsymbol{S}_v^{-1}\textbf{z}_v \}$ reduces to $\left\{0<0\right\} = \emptyset$. Therefore,
\begin{equation*}
\mathbb{P}_{pq}\left(\textbf{z}_v^{\top}\pmb{\beta}_{\alpha_{opt,v}^c}\pmb{\beta}_{\alpha_{opt,v}^c}^{\top}\textbf{z}_v<\sigma^2\textbf{z}_v^{\top}\boldsymbol{S}_v^{-1}\textbf{z}_v,\textbf{z}_v= \textbf{0}\right)=0, \quad \text{for all} \ \ v\in \mathbb{N}.
\end{equation*}
Recall that for two quadratic forms $\textbf{z}^{\top}\boldsymbol{A}\textbf{z}$ and $\textbf{z}^{\top}\boldsymbol{B}\textbf{z}$, if for $\textbf{z}\neq \textbf{0}$, $\textbf{z}^{\top}\boldsymbol{A}\textbf{z}<\textbf{z}^{\top}\boldsymbol{B}\textbf{z}$, then $\lambda_{max}(\boldsymbol{A}-\boldsymbol{B})<0$. Recall also that from Courant-Fisher's Theorem, for real-valued symmetric matrices $\boldsymbol{A}$ and $\boldsymbol{B}$, we have $\lambda_{max}\left(\boldsymbol{A}\right)+\lambda_{min}\left(\boldsymbol{B}\right)\leq\lambda_{max}\left(\boldsymbol{A}+\boldsymbol{B}\right)\leq \lambda_{max}\left(\boldsymbol{A}\right)+\lambda_{max}\left(\boldsymbol{B}\right)$. Therefore, combining the two statements above, we obtain
\begin{align*}
&\lim_{v \rightarrow \infty}\mathbb{P}_{pq}\left(\textbf{z}_v^{\top}\pmb{\beta}_{\alpha_{opt,v}^c}\pmb{\beta}_{\alpha_{opt,v}^c}^{\top}\textbf{z}_v<\sigma^2\textbf{z}_v^{\top}\boldsymbol{S}_v^{-1}\textbf{z}_v,\textbf{z}_v\neq \textbf{0}\right)\\&\quad=\lim_{v \rightarrow \infty}\mathbb{P}_{pq}\left(\lambda_{max}(\pmb{\beta}_{\alpha_{opt,v}^c}\pmb{\beta}_{\alpha_{opt,v}^c}^{\top}-\sigma^2\boldsymbol{S}_v^{-1})<0\right) 
\\ &\quad\leq\lim_{v \rightarrow \infty}\mathbb{P}_{pq}\left(\lambda_{max}(\pmb{\beta}_{\alpha_{opt,v}^c}\pmb{\beta}_{\alpha_{opt,v}^c}^{\top})+\sigma^2(\lambda_{min}(-\boldsymbol{S}_v^{-1}))<0\right)\\
   &\quad=\lim_{v \rightarrow \infty}\left(\lambda_{max}(\pmb{\beta}_{\alpha_{opt,v}^c}\pmb{\beta}_{\alpha_{opt,v}^c}^{\top})-\sigma^2(\lambda_{max}(\boldsymbol{S}_v^{-1}))<0\right)\\
&\quad=\lim_{v \rightarrow \infty}\mathbb{P}_{pq}\left(\lambda_{max}(\pmb{\beta}_{\alpha_{opt,v}^c}\pmb{\beta}_{\alpha_{opt,v}^c}^{\top})-\sigma^2(\lambda_{min}(\boldsymbol{S}_v))^{-1}<0\right)\\
&\quad= \lim_{v \rightarrow \infty}\mathbb{P}_{pq}\left(\Arrowvert\pmb{\beta}_{\alpha_{opt,v}^c}\Arrowvert_2^2-N_v^{-1}\sigma^2(\lambda_{min}(N_v^{-1}\boldsymbol{S}_v))^{-1}<0\right)  \\
&\quad=\lim_{v \rightarrow \infty}\mathbb{P}_{pq}\left(\lambda_{min}(N_v^{-1}\boldsymbol{S}_v)<\frac{\sigma^2}{N_v\Arrowvert\pmb{\beta}_{\alpha_{opt,v}^c}\Arrowvert_2^2}\right).
\end{align*}
Using Lemma \ref{SMatrix}, there exists a constant $K_0>0$ such that
\begin{align*}
&\lim_{v \rightarrow \infty}\mathbb{P}_{pq}\left(\lambda_{min}(N_v^{-1}\boldsymbol{S}_v)<\frac{\sigma^2}{N_v\Arrowvert\pmb{\beta}_{\alpha_{opt,v}^c}\Arrowvert_2^2}\right)\\
&\quad=\lim_{v \rightarrow \infty}\mathbb{P}_{pq}\left(\lambda_{min}(N_v^{-1}\boldsymbol{S}_v)<\frac{\sigma^2}{N_v\Arrowvert\pmb{\beta}_{\alpha_{opt,v}^c}\Arrowvert_2^2},\lambda_{min}(N_v^{-1}\boldsymbol{S}_v)\geq K_0\right)\\
&\quad\quad+\lim_{v \rightarrow \infty}\mathbb{P}_{pq}\left(\lambda_{min}(N_v^{-1}\boldsymbol{S}_v)<\frac{\sigma^2}{N_v\Arrowvert\pmb{\beta}_{\alpha_{opt,v}^c}\Arrowvert_2^2},\lambda_{min}(N_v^{-1}\boldsymbol{S}_v)< K_0\right)\\
&\quad\leq \lim_{v \rightarrow \infty}\mathbb{P}_{pq}\left(K_0\leq \frac{\sigma^2}{N_v\Arrowvert\pmb{\beta}_{\alpha_{opt,v}^c}\Arrowvert_2^2} \right)+\lim_{v \rightarrow \infty}\mathbb{P}_{pq}\left(\lambda_{min}(N_v^{-1}\boldsymbol{S}_v)< K_0\right)=0.
\end{align*} 
Therefore, it follows that $\lim_{v \rightarrow \infty}q_{1,v}=0$ almost surely.

\paragraph{Convergence of $q_{2,v}$.} It remains to show that, almost surely, $$\lim_{v\to \infty} q_{2,v}= \lim_{v\to \infty}  \mathbb{P}_{pq}\left(\mathcal{L}_n(\mathcal{\alpha}_F)<\mathcal{L}_n(\alpha^{\star})\right) =0.$$ However, since $\alpha^{\star}\subset\alpha_F \in \mathcal{C}$, the event $\mathcal{L}_n\left(\mathcal{\alpha}_F\right)<\mathcal{L}_n\left(\alpha^{\star}\right)$ implies $\mathcal{L}_{2,n}\left(\alpha_F\right)<\mathcal{L}_{2,n}\left(\alpha^{\star}\right)$, which by Proposition \ref{prop2} is a negligible event. It follows that $q_{2,v}=0$ for all $v \in \mathbb{N}$.
~\\ Therefore, we conclude that $\lim_{v \rightarrow \infty}p_{2,v}=0$ which then implies
\begin{equation*}
\lim_{v \rightarrow \infty}\mathbb{P}_{pq}\left(\alpha_{opt,v}\neq \alpha^{\star}\right)=0
\end{equation*}
almost surely. Finally, applying the Lebesgue convergence theorem gives
$$\lim_{v \to \infty} \mathbb{P}_{mpq} \left(\alpha_{opt,v} = \alpha^{\star}\right) = 1 $$
almost surely.
\subsection*{Proof of Proposition \ref{prop3}.} \label{proof_Prop3}

\subsubsection*{Proof of statement (i).} 
For $\alpha \in \mathcal{A}$, consider the decomposition of $\widehat{\mu}_{\alpha,v}-\mu_v$
\begin{equation*}
\widehat{\mu}_{\alpha,v}-\mu_v=\underbrace{\widehat{\mu}_{\alpha,v}-\widehat{\mu}_{\pi,v}}_{(a)}+\underbrace{\widehat{\mu}_{\pi,v}-\mu_v}_{(b)}.
\end{equation*}
As $\pi_k=\pi$ for all $k \in U_v$, we have
\begin{equation*}
\widehat{\pmb{\beta}}_{\alpha,v}=\boldsymbol{A}_{r,\alpha}^{-1}\bX_{r,\alpha}^{\top}\boldsymbol{Y}_r=\boldsymbol{A}_{r,\alpha}^{-1}\sum_{k\in S_{r,v}}\bx_{k,\alpha}y_k=\left(\frac{\boldsymbol{A}_{r,\alpha}}{N_v\pi}\right)^{-1}\sum_{k \in S_{r,v}}\frac{\bx_{k,\alpha}y_k}{N_v\pi}.
\end{equation*}
By expanding the term (a), we obtain
\begin{align*}
\widehat{\mu}_{\alpha,v}-\widehat{\mu}_{\pi,v}&=\sum_{k \in S_{m,v}}\frac{\bx_{k,\alpha}^{\top}\left(\widehat{\pmb{\beta}}_{\alpha,v}-\pmb{\beta}\right)+\epsilon_k}{N_v\pi}\\
&=\sum_{k \in S_{m,v}}\frac{\bx_{k,\alpha}^{\top}}{N_v\pi}\left(\frac{\boldsymbol{A}_{r,\alpha}}{ N_v\pi}\right)^{-1}\sum_{k\in S_{r,v}}\frac{\bx_{k,\alpha}y_k}{N_v\pi}-\sum_{k \in S_{m,v}}\frac{\bx_{k,\alpha}^{\top}\pmb{\beta}_{\alpha}}{N_v\pi}+\sum_{k\in S_{r,v}}\frac{\epsilon_k}{N_v\pi}\\
&=\sum_{k \in S_{m,v}}\frac{\bx_{k,\alpha}^{\top}}{N_v\pi}\left(\frac{\boldsymbol{A}_{r,\alpha}}{ N_v\pi}\right)^{-1}\sum_{k\in S_{r,v}}\frac{\bx_{k,\alpha}\bx_k^{\top}\pmb{\beta}+\bx_{k,\alpha}\epsilon_k}{N_v\pi}\\
&\quad-\sum_{k \in S_{m,v}}\frac{\bx_{k,\alpha}^{\top}\pmb{\beta}_{\alpha}}{N_v\pi}+\sum_{k\in S_{r,v}}\frac{\epsilon_k}{N_v\pi}.
\end{align*}
The law of large numbers gives
\begin{align*}
    &\qquad  \sum_{k \in S_{m,v}}\frac{\epsilon_k}{N_v \pi}\xrightarrow[v \rightarrow \infty]{\mathbb{P}}0, \qquad
    \sum_{k \in S_{m,v}}\frac{\bx_{k,\alpha}^{\top}\pmb{\beta}_{\alpha}}{N_v\pi}\xrightarrow[v \rightarrow \infty]{\mathbb{P}}\mathbb{E}_\bx\left[(1-p(\bx))\bx_{\alpha}^{\top}\pmb{\beta}_{\alpha}\right],\\
    &\qquad \left(\frac{\boldsymbol{A}_{r,\alpha}}{ N_v\pi}\right)^{-1}\xrightarrow[v \rightarrow \infty]{\mathbb{P}}\left(\mathbb{E}_\bx\left[p(\bx)\bx_{\alpha}\bx_{\alpha}^{\top}\right]\right)^{-1},\\&\qquad 
\sum_{k \in S_{r,v}}\frac{\bx_{k,\alpha}\bx_{k}^{\top}\pmb{\beta}+\bx_{k,\alpha}\epsilon_k}{ N_v\pi}\xrightarrow[v \rightarrow \infty]{\mathbb{P}}\mathbb{E}_\bx\left[p(\bx)\bx_{\alpha}\bx^{\top}\pmb{\beta}\right].
\end{align*}

Using the continuous mapping theorem, we obtain
\begin{align*}
\widehat{\mu}_{\alpha,v}-\mu_v&\xrightarrow[v \rightarrow \infty]{\mathbb{P}}\mathbb{E}_\bx\left[\left(1-p(\bx)\right)\bx_{\alpha}^{\top}\right]\left(\mathbb{E}_\bx\left[p(\bx)\bx_{\alpha}\bx_{\alpha}^{\top}\right]\right)^{-1}\mathbb{E}_\bx\left[p(\bx)\bx_{\alpha}\bx^{\top}\pmb{\beta}\right]\\
&\qquad-\mathbb{E}_\bx\left[\left(1-p(\bx)\right)\bx^{\top}\pmb{\beta}\right]
\end{align*}
Recall that $\widehat{\mu}_{\alpha,v}-\mu_v$ is consistent if $\widehat{\mu}_{\alpha,v}-\mu_v\xrightarrow[v \rightarrow \infty]{\mathbb{P}}0$. Equivalently,
\begin{equation*}
\mathbb{E}_\bx\left[\left(1-p(\bx)\right)\bx_{\alpha}^{\top}\right]\left(\mathbb{E}_\bx\left[p(\bx)\bx_{\alpha}\bx_{\alpha}^{\top}\right]\right)^{-1}\mathbb{E}_\bx\left[p(\bx)\bx_{\alpha}\bx^{\top}\pmb{\beta}\right]-\mathbb{E}_\bx\left[\left(1-p(\bx)\right)\bx^{\top}\pmb{\beta}\right]=0.
\end{equation*}
To obtain \eqref{condBias}, one can use the decomposition $\bx^\top \boldsymbol{\beta} = \bx_{ \alpha}^\top \bbeta_\alpha+\bx_{ \alpha^c}^\top \bbeta_\alpha^c$ on the above equality, from which the result follows after some algebra.
\subsubsection*{Proof of statement (ii), (a).}
By Theorem \ref{Theo3}, for, $\alpha \in \cC$, we have $$\widehat{\mu}_{\alpha,v}-\mu_v=\widehat{\mu}_{\alpha,v}-\widetilde{\mu}_{\alpha,v}+o_{\mathbb{P}}\left(n_v^{-1/2}\right).$$ Consequently, $\mathbb{AV}\left(\sqrt{N_v}(\widetilde{\mu}_{\alpha,v}-\mu_v)\right)=\mathbb{AV}\left(\sqrt{N_v}(\widehat{\mu}_{\alpha,v}-\mu_v)\right).$
To obtain the asymptotic variance, we define
\begin{align*}
\mathbb{AV}\left(\sqrt{N_v}(\widehat{\mu}_{\alpha,v}-\mu_v)\right)&=\lim_{v \rightarrow \infty}\bigg(\mathbb{E}_{mq}\left[\mathbb{V}_p\left(\sqrt{N_v}\big(\widetilde{\mu}_{\alpha,v}-\mu_v\big)\right)\right]\\
&\quad+\mathbb{E}_q\left[\mathbb{V}_m\left(\mathbb{E}_p[\sqrt{N_v}(\widetilde{\mu}_{\alpha,v}-\mu_v)]\right)\right]\\
&\quad+\mathbb{V}_q\left(\mathbb{E}_{mp}[\sqrt{N_v}(\widetilde{\mu}_{\alpha,v}-\mu_v)]\right)\bigg).
\end{align*}
Recalling that $\mathbb{V}_q(\mathbb{E}_{mp}[\sqrt{N_v}(\widetilde{\mu}_{\alpha,v}-\mu_v)])=0$, only the first two terms remain and we write
\begin{align*}
\mathbb{AV}\left(\sqrt{N_v}(\widehat{\mu}_{\alpha,v}-\mu_v)\right)&=\lim_{v \rightarrow \infty}\Bigg(\mathbb{E}_{mq}\left[\mathbb{V}_p\big(\sqrt{N_v}(\widetilde{\mu}_{\alpha,v}-\mu_v)\big)\right]\\
&\quad+\mathbb{E}_q\left[\mathbb{V}_m\left(\mathbb{E}_p[\sqrt{N_v}(\widetilde{\mu}_{\alpha,v}-\mu_v)]\right)\right]\Bigg)\\
&:= \text{AV}_1 + \text{AV}_2,
\end{align*}

\paragraph{Asymptotic behavior of $\text{AV}_1$.} \emph{Main idea.}  If there exists $c_{1,\alpha}$ such that $$\mathbb{V}_p(\sqrt{N_v}\left(\widetilde{\mu}_{\alpha,v}-\mu_v\right))\xrightarrow[v\rightarrow\infty]{\mathbb{P}}c_{1,\alpha},$$ and if $\mathbb{V}_p(\sqrt{N_v}(\widetilde{\mu}_{\alpha,v}-\mu_v))$ is uniformly integrable, then, almost surely, we would have $$\lim_{v \rightarrow \infty}\mathbb{E}_{mq}\left[\left|\mathbb{V}_p(\sqrt{N_v}\left(\widetilde{\mu}_{\alpha,v}-\mu_v)\right)-c_{1,\alpha}\right|\right]=0.$$ Furthermore, Jensen inequality would imply $\lim_{v \rightarrow \infty}|\mathbb{E}_{mq}[\mathbb{V}_p(\sqrt{N_v}(\widetilde{\mu}_{\alpha,v}-\mu_v))-c_{1,\alpha}]|=0$ almost surely. We shall follow the above architecture.

 ~\\
 \emph{Uniform integrability.}  The uniform integrability of $\mathbb{V}_p(\sqrt{N_v}(\widetilde{\mu}_{\alpha,v}-\mu_v))$ follows by Lemma \ref{BoundedV1V2} in which we prove that  $\mathbb{E}_{mq}[(\mathbb{V}_p(\sqrt{N_v}(\widetilde{\mu}_{\alpha,v}-\mu_v)))^2]$ is bounded. 

~\\
\emph{Limit determination.} We now proceed to find $c_{1,\alpha}$ such that $\mathbb{V}_p(\sqrt{N_v}(\widetilde{\mu}_{\alpha,v}-\mu_v))\xrightarrow[v\rightarrow\infty]{\mathbb{P}}c_{1,\alpha}.$ To that aim, write
\begin{align*}
\mathbb{V}_p\left(\sqrt{N_v}\left(\widetilde{\mu}_{\alpha,v}-\mu_v\right)\right)&=\frac{1}{N_v}\sum_{k \in U_v}\sum_{l \in U_v}\frac{\eta_{k,\alpha}}{\pi}\frac{\eta_{l,\alpha}}{\pi}\Delta\\
&=\frac{1}{N_v}\sum_{k \in U_v}\sum_{l \in U_v}\frac{\bx_{k,\alpha}^{\top}\pmb{\beta}_{\alpha}}{\pi}\frac{\bx_{l,\alpha}^{\top}\pmb{\beta}_{\alpha}}{\pi}\Delta\\
&\quad+\frac{2}{N_v}\sum_{k \in U_v}\sum_{l \in U_v}\frac{\bx_{k,\alpha}^{\top}\pmb{\beta}_{\alpha}}{\pi}\frac{r_l(1+\pi\textbf{c}_{\alpha,v}^{\top}\bx_{l,\alpha})\epsilon_l}{\pi}\Delta\\
&\quad+\frac{1}{N_v}\sum_{k \in U_v}\sum_{l \in U_v}\frac{r_k(1+\pi\textbf{c}_{\alpha,v}^{\top}\bx_{k,\alpha})\epsilon_k}{\pi}\frac{r_l(1+\pi\textbf{c}_{\alpha,v}^{\top}\bx_{l,\alpha})\epsilon_l}{\pi}\Delta\\
&:=A_v+B_v+C_v.
\end{align*}
Next, we investigate the asymptotic behavior $A_v$, $B_v$, and $C_v$ separately.
\paragraph{Behavior of $A_v$.} Since $\alpha \in \mathcal{C}$, we have $\bx_{k,\alpha}^{\top}\pmb{\beta}_{\alpha}=\bx_{k,\alpha^{\star}}^{\top}\pmb{\beta}_{\alpha^{\star}}$ and thus  $A_v$ does not depend on $\alpha$. As a result, there exists a constant $C_1(\alpha^{\star})$ such that $A_v-C_1(\alpha^{\star})=o_{\mathbb{P}}(1).$

\paragraph{Behavior of $B_v$.} Write
\begin{align*}
B_v&=2\frac{1-\pi}{\pi}\sum_{k \in U_v}\frac{\bx_{k,\alpha}^{\top}\pmb{\beta}_{\alpha}r_k\left(1+\pi\textbf{c}_{\alpha,v}^{\top}\bx_{k,\alpha}\right)\epsilon_k}{N_v}\\
&\quad+\frac{2}{N_v}\frac{1-\pi}{\pi}\sum_{k \in U_v}\sum_{\substack{l\in U_v\\l\neq k}}\frac{\bx_{k,\alpha}^{\top}\pmb{\beta}_{\alpha}}{\pi}\frac{r_l\left(1+\pi\textbf{c}_{\alpha,v}^{\top}\bx_{l,\alpha}\right)\epsilon_l}{\pi}\Delta\\
&:=B_{1,v}+B_{2,v}.
\end{align*}
By substituting $\pi\textbf{c}_{\alpha,v}$, $B_{1,v}$ gives
\begin{align*}
B_{1,v}&=\frac{1-\pi}{\pi}\sum_{k \in U_v}\frac{\bx_{k,\alpha}^{\top}\pmb{\beta}_{\alpha}r_k\epsilon_k}{N_v}\\
&\quad+\frac{1-\pi}{\pi}\sum_{k \in U_v}\frac{\bx_{k,\alpha}^{\top}\pmb{\beta}_{\alpha}r_k}{N_v}\bx_{k,\alpha}^{\top}\epsilon_k\left(\sum_{k \in U_v}\frac{r_k\bx_{k,\alpha}\bx_{k,\alpha}^{\top}}{N_v}\right)^{-1}\sum_{k \in U_v}\frac{\left(1-r_k\right)\bx_{k,\alpha}}{N_v}.
\end{align*}
By the law of large numbers, we obtain
\begin{align*}
\dfrac{1}{N_v}\sum_{k \in U_v}\bx_{k,\alpha}^{\top}\pmb{\beta}_{\alpha}r_k\epsilon_k\xrightarrow[v\rightarrow\infty]{\mathbb{P}}\mathbb{E}_{\bx pq}\left[\bx_{\alpha}^{\top}\pmb{\beta}_{\alpha}p(\bx)\epsilon\right]=0,
\end{align*}
\begin{align*}
\dfrac{1}{N_v}\sum_{k \in U_v}\bx_{k,\alpha}^{\top}\pmb{\beta}_{\alpha}r_k\bx_{k,\alpha}^{\top}\epsilon_k&\xrightarrow[v\rightarrow\infty]{\mathbb{P}}\mathbb{E}_{\bx pq}\left[\bx_{\alpha}^{\top}\pmb{\beta}_{\alpha}p(\bx)\bx_{\alpha}\epsilon\right]=\boldsymbol{0}_{\alpha}^{\top},
\end{align*}
and
\begin{equation*}
\left(\dfrac{1}{N_v}\sum_{k \in U_v}r_k\bx_{k,\alpha}\bx_{k,\alpha}^{\top}\right)^{-1}\dfrac{1}{N_v}\sum_{k \in U_v}(1-r_k)\bx_{k,\alpha}\xrightarrow[v\rightarrow\infty]{\mathbb{P}}\left(\mathbb{E}_\bx\left[p(\bx)\bx_{\alpha}\bx_{\alpha}^{\top}\right]\right)^{-1}\mathbb{E}_\bx\left[\left(1-p(\bx)\right)\bx_{\alpha}\right].
\end{equation*}
As a result, $B_{1,v}\xrightarrow[v\rightarrow\infty]{\mathbb{P}}0.$ For $B_{2,v}$, recall that $K:=\lim_{v \rightarrow \infty}N_v\Delta$, we have
\begin{align*}
B_{2,v}&=\frac{2}{N_v}\sum_{k \in U_v}\sum_{\substack{l\in U_v\\l\neq k}}\frac{\bx_{k,\alpha}^{\top}\pmb{\beta}_{\alpha}}{\pi}\frac{r_l\left(1+\pi\textbf{c}_{\alpha,v}^{\top}\bx_{l,\alpha}\right)\epsilon_l}{\pi}\Delta\\
&=\frac{2}{\pi^2}N_v\Delta\sum_{k \in U_v}\sum_{\substack{l\in U_v\\l\neq k}}\frac{\bx_{k,\alpha}^{\top}\pmb{\beta}_{\alpha}}{N_v}\frac{r_l\left(1+\pi\textbf{c}_{\alpha,v}^{\top}\bx_{l,\alpha}\right)\epsilon_l}{N_v}\\
&=\frac{2}{\pi^2}(K+o(1))\Bigg(\sum_{k \in U_v}\sum_{\substack{l\in U_v}}\frac{\bx_{k,\alpha}^{\top}\pmb{\beta}_{\alpha}}{N_v}\frac{r_l\left(1+\pi\textbf{c}_{\alpha,v}^{\top}\bx_{l,\alpha}\right)\epsilon_l}{N_v}\\
&\quad-\frac{1}{N_v}\sum_{k \in U_v}\frac{\bx_{k,\alpha}^{\top}\pmb{\beta}_{\alpha}r_k\left(1+\pi\textbf{c}_{\alpha,v}^{\top}\bx_{k,\alpha}\right)\epsilon_k}{N_v} \Bigg)\\
&=\frac{2}{\pi^2}\left(K+o(1)\right)\left(\left(\sum_{k \in U_v}\frac{\bx_{k,\alpha}^{\top}\pmb{\beta}_{\alpha}}{N_v}\right)\left(\sum_{k \in U_v}\frac{r_k\left(1+\pi\textbf{c}_{\alpha,v}^{\top}\bx_{k,\alpha}\right)\epsilon_k}{N_v}\right)+o_{\mathbb{P}}\left(N_v^{-1}\right)\right)\\
&=\frac{2}{\pi^2}\left(K+o(1)\right)\left((\mathbb{E}_\bx[\bx_{k,\alpha}^{\top}\pmb{\beta}_{\alpha}]+o_{\mathbb{P}}(1))\left(0+o_{\mathbb{P}}(1))+o_{\mathbb{P}}(N_v^{-1}\right)\right)=o_{\mathbb{P}}(1),
\end{align*}
since
\begin{align}\label{3ii}
\sum_{k \in U_v}\frac{r_k\left(1+\pi\textbf{c}_{\alpha,v}^{\top}\bx_{k,\alpha}\right)\epsilon_k}{N_v}&=\sum_{k \in U_v}\frac{r_k\epsilon_k}{N_v}\nonumber\\
&\quad+\sum_{k \in U_v}\frac{r_k\epsilon_k\bx_{k,\alpha}^{\top}}{N_v}\left(\sum_{k \in U_v}\frac{r_k\bx_{k,\alpha}\bx_{k,\alpha}^{\top}}{N_v}\right)^{-1}\sum_{k \in U_v}\frac{\left(1-r_k\right)\bx_{k,\alpha}}{N_v}\nonumber\\
&\xrightarrow[v \rightarrow \infty]{\mathbb{P}}0+\textbf{0}_{\alpha}^{\top}\left(\mathbb{E}_\bx\left[p(\bx)\bx_{\alpha}\bx_{\alpha}^{\top}\right]\right)^{-1}\mathbb{E}_\bx[(1-p(\bx))\bx_{\alpha}]=0.
\end{align}
Finally, we have $B_v\xrightarrow[v\rightarrow\infty]{\mathbb{P}} 0$. 
\paragraph{Behavior of $C_v$.}Write
\begin{align*}
C_v&=\frac{1-\pi}{\pi}\sum_{k \in U_v}\frac{r_k^2\left(1+\pi\textbf{c}_{\alpha,v}^{\top}\bx_{k,\alpha}\right)^2\epsilon_k^2}{N_v}\\
&\quad+\frac{1}{N_v}\sum_{ k\in U_v}\sum_{\substack{l \in U_v\\ l\neq k}}\frac{r_k\left(1+\pi\textbf{c}_{\alpha,v}^{\top}\bx_{k,\alpha}\right)\epsilon_k}{\pi}\frac{r_l\left(1+\pi\textbf{c}_{\alpha,v}^{\top}\bx_{l,\alpha}\right)\epsilon_l}{\pi}\Delta\\
&:=C_{1,v}+C_{2,v}.
\end{align*}
Expanding $C_{1,v}$ gives
\begin{align*}
C_{1,v}&=\frac{1-\pi}{\pi}\sum_{k \in U_v}\frac{r_k^2\left(1+\pi\textbf{c}_{\alpha,v}^{\top}\bx_{k,\alpha}\right)^2\epsilon_k^2}{N_v}\nonumber\\
&=\frac{1-\pi}{\pi}\sum_{k \in U_v}\frac{r_k^2\epsilon_k^2}{N_v}+2\frac{1-\pi}{\pi}\sum_{k \in U_v}\frac{r_k^2\epsilon_k^2\pi\textbf{c}_{\alpha,v}^{\top}\bx_{k,\alpha}}{N_v}+\frac{1-\pi}{\pi}\sum_{k \in U_v}\frac{r_k^2\epsilon_k^2\left(\pi\textbf{c}_{\alpha,v}^{\top}\bx_{k,\alpha}\right)^2}{N_v}\nonumber\\
&=\frac{1-\pi}{\pi}\sum_{k \in U_v}\frac{r_k^2\epsilon_k^2}{N_v}+2\frac{1-\pi}{\pi}\sum_{k \in U_v}\frac{r_k^2\epsilon_k^2\pi\bx_{k,\alpha}^{\top}\textbf{c}_{\alpha,v}}{N_v}\\
&\quad+\frac{1-\pi}{\pi}\left(\pi\textbf{c}_{\alpha,v}^{\top}\right)\sum_{k \in U_v}\frac{r_k^2\epsilon_k^2\bx_{k,\alpha}\bx_{k,\alpha}^{\top}}{N_v}\left(\pi\textbf{c}_{\alpha,v}\right)\nonumber\\
&=\frac{1-\pi}{\pi}\sum_{k \in U_v}\frac{r_k^2\epsilon_k^2}{N_v}+2\frac{1-\pi}{\pi}\sum_{k \in U_v}\frac{r_k^2\epsilon_k^2\bx_{k,\alpha}^{\top}}{N_v}\left(\sum_{k \in U_v}\frac{r_k\bx_{k,\alpha}\bx_{k,\alpha}}{N_v}\right)^{-1}\sum_{k \in U_v}\frac{(1-r_k)\bx_{k,\alpha}}{N_v}\nonumber\\
&\quad+\frac{1-\pi}{\pi}\sum_{k \in U_v}\frac{(1-r_k)\bx_{k,\alpha}^{\top}}{N_v}\left(\sum_{k \in U_v}\frac{r_k\bx_{k,\alpha}\bx_{k,\alpha}^{\top}}{N_v}\right)^{-1}\left(\sum_{k \in U_v}\frac{r_k^2\epsilon_k^2}{N_v}\bx_{k,\alpha}\bx_{k,\alpha}^{\top}\right)\nonumber\\
&\quad\quad\times\left(\sum_{k \in U_v}\frac{r_k\bx_{k,\alpha}\bx_{k,\alpha}^{\top}}{N_v}\right)^{-1}\sum_{k \in U_v}\frac{(1-r_k)\bx_{k,\alpha}}{N_v}.
\end{align*}
Hence, using the continuous mapping theorem, $C_{1,v}$ converges to
\begin{align*}
&\sigma^2\frac{1-\pi}{\pi}\mathbb{E}_\bx\left[p(\bx)\right]-2\sigma^2\frac{1-\pi}{\pi}\mathbb{E}_\bx\left[p(\bx)\bx_{\alpha}^{\top}\right]\left(\mathbb{E}_\bx\left[p(\bx)\bx_{\alpha}\bx_{\alpha}^{\top}\right]\right)^{-1}\mathbb{E}_\bx\left[(1-p(\bx))\bx_{\alpha}\right]\\
&\quad+\frac{1-\pi}{\pi}\mathbb{E}_\bx\left[(1-p(\bx))\bx_{\alpha}^{\top}\right]\left(\mathbb{E}_\bx\left[p(\bx)\bx_{\alpha}\bx_{\alpha}^{\top}\right]\right)^{-1}\\
&\quad\quad\times\mathbb{E}_\bx\left[p(\bx)\bx_{\alpha}\bx_{\alpha}^{\top}\right]\left(\mathbb{E}_\bx\left[p(\bx)\bx_{\alpha}\bx_{\alpha}^{\top}\right]\right)^{-1}\mathbb{E}_\bx\left[(1-p(\bx))\bx_{\alpha}\right]\\
&=\sigma^2\frac{1-\pi}{\pi}\mathbb{E}_\bx[p(\bx)]-2\sigma^2\frac{1-\pi}{\pi}\mathbb{E}_\bx\left[(1-p(\bx))\bx_{\alpha}^{\top}\right]\left(\mathbb{E}_\bx\left[p(\bx)\bx_{\alpha}\bx_{\alpha}^{\top}\right]\right)^{-1}\mathbb{E}_\bx\left[p(\bx)\bx_{\alpha}\right]\\
&\quad+\sigma^2\frac{1-\pi}{\pi}\mathbb{E}_\bx\left[(1-p(\bx))\bx_{\alpha}^{\top}\right]\left(\mathbb{E}_\bx\left[p(\bx)\bx_{\alpha}\bx_{\alpha}^{\top}\right]\right)^{-1}\mathbb{E}_\bx\left[(1-p(\bx))\bx_{\alpha}\right]\\
&\overset{(*)}{=}\sigma^2\frac{1-\pi}{\pi}\mathbb{E}_\bx[p(\bx)]-2\sigma^2\frac{1-\pi}{\pi}\mathbb{E}_\bx[1-p(\bx)]+\sigma^2\frac{1-\pi}{\pi}M(\alpha)\\
&=\sigma^2\frac{1-\pi}{\pi}\mathbb{E}_\bx[2-p(\bx)]+\sigma^2\frac{1-\pi}{\pi}M(\alpha).
\end{align*}
The equality $(*)$ is due to Lemma  \ref{AlgebraEx}. For $C_{2,v}$, using \eqref{3ii}, write
\begin{align*}
C_{2,v}&=\frac{1}{N_v}\sum_{ k\in U_v}\sum_{\substack{l \in U_v\\ l\neq k}}\frac{r_k\left(1+\pi\textbf{c}_{\alpha,v}^{\top}\bx_{k,\alpha}\right)\epsilon_k}{\pi}\frac{r_l\left(1+\pi\textbf{c}_{\alpha,v}^{\top}\bx_{l,\alpha}\right)\epsilon_l}{\pi}\Delta\\
&=\frac{N_v\Delta}{\pi^2}\sum_{ k\in U_v}\sum_{\substack{l \in U_v\\ l\neq k}}\frac{r_k\left(1+\pi\textbf{c}_{\alpha,v}^{\top}\bx_{k,\alpha}\right)\epsilon_k}{N_v}\frac{r_l\left(1+\pi\textbf{c}_{\alpha,v}^{\top}\bx_{l,\alpha}\right)\epsilon_l}{N_v}\\
&=\frac{K+o(1)}{\pi^2}\Bigg(\sum_{ k\in U_v}\sum_{\substack{l \in U_v}}\frac{r_k\left(1+\pi\textbf{c}_{\alpha,v}^{\top}\bx_{k,\alpha}\right)\epsilon_k}{N_v}\frac{r_l\left(1+\pi\textbf{c}_{\alpha,v}^{\top}\bx_{l,\alpha}\right)\epsilon_l}{N_v}\\
&\quad-\frac{1}{N_v} \sum_{k \in U_v}\frac{r_k^2\left(1+\pi\textbf{c}_{\alpha,v}^{\top}\bx_{k,\alpha}\right)\epsilon_k^2}{N_v}\Bigg)\\
&=\frac{K+o(1)}{\pi^2}\left(\left(\sum_{k \in U_v}\frac{r_k(1+\pi\textbf{c}_{\alpha,v}\bx_{k,\alpha})\epsilon_k}{N_v}\right)^2+o_{\mathbb{P}}\left(n_v^{-1}\right) \right)\\
&=\frac{K+o(1)}{\pi^2}\left((0+o_{\mathbb{P}}(1))^2+o_{\mathbb{P}}\left(n_v^{-1}\right) \right)=o_{\mathbb{P}}(1).
\end{align*}
As a result, $C_{2,v}\xrightarrow[v\rightarrow\infty]{\mathbb{P}}0.$ Finally, we conclude 
\begin{equation*}
C_v \xrightarrow[v\rightarrow\infty]{\mathbb{P}}\sigma^2\frac{1-\pi}{\pi}\mathbb{E}_\bx\left[2-p(\bx)\right]+\sigma^2\frac{1-\pi}{\pi}M(\alpha).
\end{equation*}
Overall,
\begin{equation*}
c_{1,\alpha}=C_1(\alpha^{\star})+\sigma^2\frac{1-\pi}{\pi}\mathbb{E}_\bx\left[2-p(\bx)\right]+\sigma^2\frac{1-\pi}{\pi}M(\alpha).
\end{equation*}

\paragraph{Asymptotic behavior of $\text{AV}_2$.} We proceed similarly as before.

~\\
\emph{Uniform integrability.}
In Lemma \ref{BoundedV1V2}, we prove that  $\mathbb{E}_{mq}[(\V_m(\mathbb{E}_p[\sqrt{N_v}(\widetilde{\mu}_{\alpha,v}-\mu_v)]))^2]$ is almost surely bounded, from which we deduce uniform integrability.

~\\
\emph{Limit determination.} Write
\begin{align*}
&\mathbb{V}_m\left(\mathbb{E}_p\left[\sqrt{N_v}\left(\widetilde{\mu}_{\alpha,v}-\mu_v\right)\right]\right)\\&\quad=\sigma^2\sum_{k \in U_v}\frac{1-r_k+r_k(\pi\textbf{c}_{\alpha,v}^{\top}\bx_{k,\alpha})^2}{N_v}\\
&\quad=\sigma^2\sum_{k \in U_v}\frac{1-r_k}{N_v}+\sigma^2\mathbb{E}_\bx\left[(1-p(\bx))\bx_{k,\alpha}^{\top}\right]\left(\mathbb{E}_\bx\left[p(\bx)\bx_{\alpha}\bx_{\alpha}^{\top}\right]\right)^{-1}\\
&\quad\quad\quad\times\sum_{k \in U_v}\frac{r_k\bx_{k,\alpha}\bx_{k,\alpha}^{\top}}{N_v}\left(\mathbb{E}_\bx\left[p(\bx)\bx_{\alpha}\bx_{\alpha}^{\top}\right]\right)^{-1}\mathbb{E}_\bx\left[(1-p(\bx))\bx_{k,\alpha}\right]\\
&\quad \xrightarrow[v\rightarrow\infty]{\mathbb{P}}\sigma^2\mathbb{E}_\bx\left[1-p(\bx)\right]+\sigma^2\mathbb{E}_\bx\left[(1-p(\bx))\bx_{\alpha}^{\top}\right]\left(\mathbb{E}_\bx\left[p(\bx)\bx_{\alpha}\bx_{\alpha}^{\top}\right]\right)^{-1}\mathbb{E}_\bx\left[(1-p(\bx))\bx_{k,\alpha}\right]\\
&\quad=\sigma^2\mathbb{E}_\bx\left[1-p(\bx)\right]+\sigma^2M(\alpha).
\end{align*}
As a result, 
\begin{equation*}
c_{2,\alpha}=\sigma^2\mathbb{E}_\bx[1-p(\bx)]+\sigma^2M(\alpha).
\end{equation*}
\paragraph{Putting all things together.} By absorbing the terms that do not depend on $\alpha$, called $C(\alpha^{\star})$, we conclude
\begin{align*}
\mathbb{AV}\left(\sqrt{N_v}(\widehat{\mu}_{\alpha,v}-\mu_v)\right)&=\lim_{v \rightarrow \infty}\Bigg(\mathbb{E}_{mq}\left[\mathbb{V}_p\left(\sqrt{N_v}(\widetilde{\mu}_{\alpha,v}-\mu_v)\right)\right]\\
&\quad+\mathbb{E}_q\left[\mathbb{V}_m\left(\mathbb{E}_p\left[\sqrt{N_v}(\widetilde{\mu}_{\alpha,v}-\mu_v)\right]\right)\right]\Bigg)\\
&=C_1(\alpha^{\star})+\sigma^2\frac{1-\pi}{\pi}\mathbb{E}_\bx[2-p(\bx)]+\sigma^2\frac{1-\pi}{\pi}M(\alpha)\\
&\quad+\sigma^2\mathbb{E}_\bx[1-p(\bx)]+\sigma^2M(\alpha)\\
&=C(\alpha^{\star})+\frac{\sigma^2}{\pi}M(\alpha).
\end{align*}
\subsection*{Proof of Corollary \ref{Cor1}.}\label{proof_Cor1}
Statement (i) is obvious as $\Arrowvert\pmb{\beta}_{\alpha^c}\Arrowvert_2=0$ which implies that $\pmb{\beta}_{\alpha^c} = \boldsymbol{0}_{\alpha^c}$. For statement (ii), we need to show that under our assumption,            $$    \mathbb{E}_\bx\left[( 1 - p(\bx))\bx_{\alpha^c}^{\top}\pmb{\beta}_{\alpha^c}\right]=\mathbb{E}_\bx \left[\bx_{\alpha}^{\top}(1 - p(\bx))\right] \left(\mathbb{E}_\bx \left[p(\bx) \bx_{ \alpha}\bx_{ \alpha} ^\top \right]\right)^{-1} \mathbb{E}_\bx\left[\bx_{\alpha}\bx_{\alpha^c}^{\top}\pmb{\beta}_{\alpha^c}p(\bx) \right] $$ where we have suppressed the index for simplicity of notation. This is equivalent to showing that 
\begin{equation} \label{eq:c1j}
\mathbb{E}_\bx\left[( 1 - p(\bx))x_{\alpha^c,j}\right]=\mathbb{E}_\bx \left[\bx_{\alpha}^{\top}(1 - p(\bx))\right] \left(\mathbb{E}_\bx \left[p(\bx) \bx_{ \alpha}\bx_{ \alpha} ^\top \right]\right)^{-1} \mathbb{E}_\bx\left[p(\bx)\bx_{\alpha}x_{\alpha^c,j} \right], 
\end{equation}
$j \in \alpha_{\mathrm{mis}}$. 
Since, for $j \in \alpha^c \backslash \alpha_{\mathrm{mis}} ,$ $\beta_{\alpha^c,j} =0$, we can express the left side of \eqref{eq:c1j} as $$L_j := \mathbb{E}_\bx\left[( 1 - p(\bx))x_{\alpha^c,j}\right] = \mathbb{E}_\bx\left[( 1 - p(\bx))\bx_{\alpha}^\top \gamma_j\right],$$ using both the conditional independence (a) and the linear link (b). Similarly, using the same technique, we may write $$\mathbb{E}_\bx\left[p(\bx)\bx_{\alpha}x_{\alpha^c,j} \right] = \mathbb{E}_\bx\left[p(\bx)\bx_{\alpha}\bx_\alpha^\top  \right]\boldsymbol{\gamma}_j.$$
Therefore, the right side of \eqref{eq:c1j} is equal to \begin{align*}
R_j &:= \mathbb{E}_\bx \left[\bx_{\alpha}^{\top}(1 - p(\bx))\right] \left(\mathbb{E}_\bx \left[p(\bx) \bx_{ \alpha}\bx_{ \alpha} ^\top \right]\right)^{-1} \mathbb{E}_\bx\left[p(\bx)\bx_{\alpha}x_{\alpha^c,j} \right] \\
&=  \mathbb{E}_\bx \left[(1 - p(\bx))\bx_{\alpha}^{\top}\right]\boldsymbol{\gamma}_j = L_j,
\end{align*}
$ j \in \alpha_{\mathrm{mis}}$.
The result follows. 
\subsection*{Proof of Corollary \ref{Cor2}.} \label{proofCoro2}

\subsubsection*{Proof of (i).} From Proposition \ref{prop3}, the equality holds when, for $j\in \alpha_2-\alpha_1$, we need to show
\begin{equation}\label{equality}
\mathbb{E}_\bx\left[(1-p(\bx))\bx_{\alpha_1}^{\top}\right]\left(\mathbb{E}_\bx\left[p(\bx)\bx_{\alpha_1}\bx_{\alpha_1}^{\top}\right]\right)^{-1}\mathbb{E}_\bx[p(\bx)\bx_{\alpha_1}x_{\alpha_2-\alpha_1,j}]=\mathbb{E}_\bx[(1-p(\bx))x_{\alpha_2-\alpha_1,j}].
\end{equation}
Recall that if $Y$ and $Z$ are conditionally independent given $X$, then $\mathbb{E}[YZ|X]=\mathbb{E}[Y|X]\mathbb{E}[Z|X]$.
Next, we need to prove that the left-hand side and the right-hand side of \eqref{equality} are equal. The left-hand side of \eqref{equality} gives
\begin{align*}
L_j&=\mathbb{E}_\bx\left[(1-p(\bx))\bx_{\alpha_1}^{\top}\right]\left(\mathbb{E}_\bx\left[p(\bx)\bx_{\alpha_1}\bx_{\alpha_1}^{\top}\right]\right)^{-1}\mathbb{E}_\bx[p(\bx)\bx_{\alpha_1}x_{\alpha_2-\alpha_1,j}]\\
&=\mathbb{E}_\bx\left[(1-p(\bx))\bx_{\alpha_1}^{\top}\right]\left(\mathbb{E}_\bx\left[p(\bx)\bx_{\alpha_1}\bx_{\alpha_1}^{\top}\right]\right)^{-1}\mathbb{E}_\bx[\mathbb{E}_\bx[p(\bx)\bx_{\alpha_1}x_{\alpha_2-\alpha_1,j}|\bx_{\alpha_1}]]\\
&=\mathbb{E}_\bx\left[(1-p(\bx))\bx_{\alpha_1}^{\top}\right]\left(\mathbb{E}_\bx\left[p(\bx)\bx_{\alpha_1}\bx_{\alpha_1}^{\top}\right]\right)^{-1}\mathbb{E}_\bx[\bx_{\alpha_1}\mathbb{E}_\bx[p(\bx)|\bx_{\alpha_1}]\mathbb{E}_\bx[x_{\alpha_2-\alpha_1,j}|\bx_{\alpha_1}]]\\
&=\mathbb{E}_\bx\left[(1-p(\bx))\bx_{\alpha_1}^{\top}\right]\left(\mathbb{E}_\bx\left[p(\bx)\bx_{\alpha_1}\bx_{\alpha_1}^{\top}\right]\right)^{-1}\mathbb{E}_\bx\left[\mathbb{E}_\bx[p(\bx)|\bx_{\alpha_1}]\bx_{\alpha_1}\bx_{\alpha_1}^{\top}\pmb{\gamma}_{j}\right]\\
&=\mathbb{E}_\bx\left[\mathbb{E}_\bx[1-p(\bx)|\bx_{\alpha_1}]\bx_{\alpha_1}^{\top}\right]\pmb{\gamma}_{j}\\&=\mathbb{E}_\bx\left[\mathbb{E}_\bx[1-p(\bx)|\bx_{\alpha_1}]\bx_{\alpha_1}^{\top}\pmb{\gamma}_{j}\right]
\end{align*}
Similarly, the right-hand side of \eqref{equality} gives
\begin{align*}
R_j&=\mathbb{E}_\bx[(1-p(\bx))x_{\alpha_2-\alpha_1,j}]=\mathbb{E}_\bx[\mathbb{E}_\bx[(1-p(\bx))x_{\alpha_2-\alpha_1,j}|\bx_{\alpha}]]\\
&=\mathbb{E}_\bx[\mathbb{E}_\bx[1-p(\bx)|\bx_{\alpha_1}]\mathbb{E}_\bx[x_{\alpha_2-\alpha_1,j}|\bx_{\alpha_1}]]\\&=\mathbb{E}_\bx\left[\mathbb{E}_\bx[1-p(\bx)|\bx_{\alpha_1}]\bx_{\alpha_1}^{\top}\pmb{\gamma}_j\right].
\end{align*}
As a result, both sides are equal under our assumptions.

\subsubsection*{Proof of (ii).} Let $\mathbf{c}$ be a direction satisfying (a) and (b); we need to show that $$ \mathbb{E}_\bx\left[(1-p(\bx))\bx_{\alpha_2-\alpha_1}\right] \neq \mathbb{E}_\bx\left[(1-p(\bx))\bx_{\alpha_2-\alpha_1} \bx_{\alpha_1}^{\top}\right]\left(\mathbb{E}_\bx\left[p(\bx)\bx_{\alpha_1}\bx_{\alpha_1}^{\top}\right]\right)^{-1}\mathbb{E}_\bx\left[p(\bx)\bx_{\alpha_1}\right].$$
We take the Euclidean inner product of the above with $\mathbf{c}$. The left side gives $$L = \mathbf{c}^\top  \mathbb{E}_\bx\left[(1-p(\bx))\bx_{\alpha_2-\alpha_1}\right] =   \mathbb{E}_\bx\left[(1-p(\bx))\mathbf{c}^\top\bx_{\alpha_2-\alpha_1}\right] \neq  0$$ by (b). On the right hand side, 
\begin{align*}
    R_j &= \mathbf{c} ^\top\mathbb{E}_\bx\left[(1-p(\bx))\bx_{\alpha_2-\alpha_1} \bx_{\alpha_1}^{\top}\right]\left(\mathbb{E}_\bx\left[p(\bx)\bx_{\alpha_1}\bx_{\alpha_1}^{\top}\right]\right)^{-1}\mathbb{E}_\bx\left[p(\bx)\bx_{\alpha_1}\right] \\&=   \mathbb{E}_\bx\left[(1-p(\bx))\mathbf{c} ^\top\bx_{\alpha_2-\alpha_1} \bx_{\alpha_1}^{\top}\right]\left(\mathbb{E}_\bx\left[p(\bx)\bx_{\alpha_1}\bx_{\alpha_1}^{\top}\right]\right)^{-1}\mathbb{E}_\bx\left[p(\bx)\bx_{\alpha_1}\right]\\&= 0
\end{align*}since, by (a), $$ \mathbb{E}_\bx\left[(1-p(\bx))\mathbf{c} ^\top\bx_{\alpha_2-\alpha_1} \bx_{\alpha_1}^{\top}\right] = \boldsymbol{0}.$$
\subsection*{Proof of Lemma \ref{Lemma2}.} \label{proof_Lemma2}
Assuming that 
\begin{equation} \label{eq:consPop}
\lim_{v \to \infty} \mathbb{P}_{m} \left(\widetilde{\alpha}_{U_v} = \alpha^{\star}\right)  =1
\end{equation}
almost surely, 
we wish to prove that $$\lim_{v \to \infty} \mathbb{P}_{mpq} \left(\widehat{\alpha}_{S_{r,v}} \neq \alpha^{\star}\right)  = \lim_{v \to \infty} \mathbb{P}_{mpq} \left( \cC_{S_{r,v}} \left( \widehat{\alpha}_{S_{r,v}} \right) <  \cC_{S_{r,v}} \left( \alpha^{\star}\right)\right) = 0$$ almost surely.
With a slight abuse of notation, we write $\mathbb{P}_{\cdot \rvert \cdot}$ to denote a generic conditional distribution, whose precise specification may vary but will be clear from the context. Using the missing at random assumption and non-informativeness of the sampling design, we get 
\begin{align*}
\mathbb{P}_{m} \left( \mathcal{C}_{S_{r,v}} \left( \widehat{\alpha}_{S_{r,v}} \right) <  \cC_{S_{r,v}} \left( \alpha^{\star}\right)\right) &= \int_{\R^{n_{r,v}}}  \mathds{1}{(\cC_{S_{r,v}} \left( \widehat{\alpha}_{S_{r,v}} \right) <  \cC_{S_{r,v}} \left( \alpha^\star\right) )} \prod_{k \in S_{r,v}} \mathbb{P}_{y_k\rvert \bx_k}\left(\mathrm{d} y_k\right) \\
&\xrightarrow[]{v \to \infty} 0,
\end{align*}
where the limit follows by \eqref{eq:consPop} and Lemma \ref{SetofRespondents}. Since $$  \mathbb{P}_{mpq} \left( \cC_{S_{r,v}} \left( \widehat{\alpha}_{S_{r,v}} \right) <  \cC_{S_{r,v}} \left( \alpha^{\star}\right)\right) = \mathbb{E}_{pq} \left[ \mathbb{P}_{m} \left( \cC_{S_{r,v}} \left( \widehat{\alpha}_{S_{r,v}} \right) <  \cC_{S_{r,v}} \left( \alpha^{\star}\right)\right)\right],$$
an application of Lebesgue dominated convergence gives the result. 
\subsection*{Proof of Theorem \ref{Theo2}.} \label{Proof_Theo2}
Write $$\sqrt{n_v} \left( \widehat{\mu}_{\widehat{\alpha},v} - \mu_v\right) = \sqrt{n_v} \left( \widehat{\mu}_{\alpha^{\star},v} - \mu_v\right)  + \sqrt{n_v} \left( \widehat{\mu}_{\widehat{\alpha},v} - \widehat{\mu}_{\alpha^{\star},v}\right) .$$ We will show that the second term vanishes in probability.
Observe that 
\small
\begin{align*}\left\{\sqrt{n_v} \rvert \widehat{\mu}_{\widehat{\alpha},v} -\widehat{\mu}_{\alpha^{\star},v}  \rvert > \epsilon \right\} &= \left(\left\{\sqrt{n_v} \rvert \widehat{\mu}_{\widehat{\alpha},v} -\widehat{\mu}_{\alpha^{\star},v}  \rvert > \epsilon \right\} \cap \{ \widehat{\alpha} = \alpha^{\star} \}\right) \\
&\biguplus (\left\{\sqrt{n_v} \rvert \widehat{\mu}_{\widehat{\alpha},v} -\widehat{\mu}_{\alpha^{\star},v}  \rvert > \epsilon \right\}\cap \{ \widehat{\alpha} \neq \alpha^{\star} \}).
\end{align*}
\normalsize
Moreover, on the event $\{ \widehat{\alpha} = \alpha \}$, we have $\widehat{\mu}_{\widehat{\alpha},v} =\widehat{\mu}_{\alpha^{\star},v} $ so that $$\{\sqrt{n}_v \rvert \widehat{\mu}_{\widehat{\alpha},v} -\widehat{\mu}_{\alpha^{\star},v}  \rvert > \epsilon \} \cap \{ \widehat{\alpha} = \alpha^{\star} \} = \emptyset.$$
Thus,
\begin{align*}
\mathbb{P}_{mpq}\left(\sqrt{n_v} \rvert \widehat{\mu}_{\widehat{\alpha},v} -\widehat{\mu}_{\alpha^{\star},v}  \rvert > \epsilon \right)& =    \mathbb{P}_{mpq}\left(\sqrt{n_v} \rvert \widehat{\mu}_{\widehat{\alpha},v} -\widehat{\mu}_{\alpha^{\star},v}  \rvert > \epsilon  \rvert  \ \widehat{\alpha}\neq \alpha^{\star} \right)  \mathbb{P}_{mpq} \left(\widehat{\alpha} \neq \alpha^{\star}\right) \\
&\leq \mathbb{P}_{mpq} \left(\widehat{\alpha} \neq \alpha^{\star}\right)
\end{align*} 
converges to $0$ by assumption.

\subsection*{Proof of Theorem \ref{Theo3}.}\label{proof_Theo3}
We adapt the proof of \cite{kim2009unified} to the unweighted case. Let
\begin{equation}\label{UStat}
\widehat{U}\left(\pmb{\beta}_{\alpha}\right)=\sum_{k \in S_v}\frac{r_k\bx_{k,\alpha}\left(y_k-\bx_{k,\alpha}^{\top}\pmb{\beta}_{\alpha}\right)}{N_v}.
\end{equation}
Note that $\widehat{\pmb{\beta}}_{\alpha,v}$ given by \eqref{ols} is the solution of $\widehat{U}(\pmb{\beta}_{\alpha})=\textbf{0}_{\alpha}$. 
Let
\begin{equation*}
\widehat{\mu}_{v}(\pmb{\beta}_{\alpha})=\frac{1}{N_v}\left(\sum_{k \in S_{r,v}}\frac{y_k}{\pi_k}+\sum_{k \in S_{m,v}}\frac{\bx_{k,\alpha}^{\top}\pmb{\beta}_{\alpha}}{\pi_k}\right),
\end{equation*}
which is the function of $\pmb{\beta}_{\alpha}$. We further define \begin{equation*}
\widetilde{\mu}_{v}\left(\pmb{\gamma},\pmb{\beta}_{\alpha}\right)=\widehat{\mu}_{v}\left(\pmb{\beta}_{\alpha}\right)+\pmb{\gamma}^{\top}\widehat{U}\left(\pmb{\beta}_{\alpha}\right),
\end{equation*}
which is seen as a function of $\pmb{\beta}_{\alpha}$ and $\pmb{\gamma}$.  Our goal is to find a particular choice of $\pmb{\gamma}$, called $\pmb{\gamma}^{\star}$ such that $\widetilde{\mu}_{v}(\pmb{\gamma}^{\star},\widehat{\pmb{\beta}}_{\alpha.v})-\widetilde{\mu}_{v}(\pmb{\gamma}^{\star},\pmb{\beta}_{\alpha,v})=o_\mathbb{P}(n_v^{-1/2})$. If so, the effect of estimating $\pmb{\beta}_{\alpha}$ can be ignored by choosing $\pmb{\gamma}=\pmb{\gamma}^{\star}$. By noting $\widetilde{\mu}_{v}(\pmb{\gamma},\widehat{\pmb{\beta}}_{\alpha.v})=\widehat{\mu}_{\alpha,v}$ for every $p_{\alpha}$-dimensional vector $\pmb{\gamma}$, we have $\widehat{\mu}_{\alpha,v}-\widetilde{\mu}_{v}(\pmb{\gamma}^{\star},\pmb{\beta}_{\alpha,v})=o_\mathbb{P}(n_v^{-1/2})$.  To find $\pmb{\gamma}^{\star}$, we use the theory of \cite{randles1982asymptotic}, having proved $\sqrt{N_v}(\widehat{\pmb{\beta}}_{\alpha,v}-\pmb{\beta}_{\alpha})=\mathcal{O}_\P(1)$ in Lemma \ref{BetaEst}, then $\widetilde{\mu}_{v}(\pmb{\gamma}^{\star},\widehat{\pmb{\beta}}_{\alpha,v})-\widetilde{\mu}_{v}(\pmb{\gamma}^{\star},\pmb{\beta}_{\alpha})=o_\mathbb{P}(n_v^{-1/2})$ holds if 
\begin{equation*}
\mathbb{E}_{mp}\left[\frac{\partial \widetilde{\mu}_{v}(\pmb{\gamma},\pmb{\beta}_{\alpha})}{\partial\pmb{\beta}_{\alpha}}\bigg|_{\pmb{\beta}_{\alpha}=\pmb{\beta}_{\alpha}} \right]=\mathbb{E}_{mp}\left[ \frac{\partial \widehat{\mu}_{v}(\pmb{\beta}_{\alpha})}{\partial\pmb{\beta}_{\alpha}}\bigg|_{\pmb{\beta}_{\alpha}=\pmb{\beta}_{\alpha}}\right]-\pmb{\gamma}^{\top}\mathbb{E}_{mp}\left[ \frac{\partial \widehat{U}(\pmb{\beta}_{\alpha})}{\partial\pmb{\beta}_{\alpha}}\bigg|_{\pmb{\beta}_{\alpha}=\pmb{\beta}_{\alpha}}\right]=0.
\end{equation*}
The solution of $\pmb{\gamma}^{\star}$ is given by
\begin{align*}
\pmb{\gamma}^{\star}&=-\left(\mathbb{E}_{mp}\left[\frac{\partial \widehat{U}(\pmb{\beta}_{\alpha})^{\top{}}}{\partial\pmb{\beta}_{\alpha}}\bigg|_{\pmb{\beta}_{\alpha}=\pmb{\beta}_{\alpha}} \right]\right)^{-1}\mathbb{E}_{mp}\left[ \frac{\partial \widehat{\mu}_{v}(\pmb{\beta}_{\alpha})}{\partial\pmb{\beta}_{\alpha}}\bigg|_{\pmb{\beta}_{\alpha}=\pmb{\beta}_{\alpha}}\right]\\&=\left(\sum_{k \in U_v}\frac{r_k\pi_k\bx_{k,\alpha}\mathbf{x}_{k,\alpha}^{\top}}{N_v} \right)^{-1}\sum_{k \in U_v}\frac{(1-r_k)\bx_{k,\alpha}}{N_v}=\textbf{c}_{\alpha,v}.
\end{align*}
Next we will show that $\widetilde{\mu}_v(\pmb{\gamma}^{\star},\pmb{\beta}_{\alpha})=\widetilde{\mu}_{\alpha,v}$. Write $\widetilde{\mu}_{v}(\pmb{\gamma},\pmb{\beta}_{\alpha})$ as
\begin{align*}
\widetilde{\mu}_{v}\left(\pmb{\gamma},\pmb{\beta}_{\alpha}\right)&=\widehat{\mu}_{v}(\pmb{\beta}_{\alpha})+\pmb{\gamma}^{\top}\widehat{U}(\pmb{\beta}_{\alpha})\\
&=\frac{1}{N_v}\left(\sum_{k \in S_{r,v}}\frac{y_k}{\pi_k}+\sum_{k\in S_{m,v}}\frac{\bx_{k,\alpha}^{\top}\pmb{\beta}_{\alpha}}{\pi_k}+\pmb{\gamma}^{\top}\sum_{k \in S_{r,v}}\frac{\pi_k\left(y_k-\bx_{k,\alpha}^{\top}\pmb{\beta}_{\alpha}\right)\bx_{k,\alpha}}{\pi_k}\right)\\
&=\frac{1}{N_v}\left(\sum_{k \in S_{r,v}}\frac{y_k+\pi_k\pmb{\gamma}^{\top}\left(y_k-\bx_{k,\alpha}^{\top}\pmb{\beta}_{\alpha}\right)\bx_{k,\alpha}}{\pi_k}+\sum_{k \in S_{m,v}}\frac{\bx_{k,\alpha}^{\top}\pmb{\beta}_{\alpha}}{\pi_k} \right)\\
&=\frac{1}{N_v}\sum_{k\in S_{v}} \frac{r_ky_k+r_k\pi_k\pmb{\gamma}^{\top}\bx_{k,\alpha}\left(y_k-\bx_{k,\alpha}^{\top}\pmb{\beta}_{\alpha}\right)+(1-r_k)\bx_{k,\alpha}^{\top}\pmb{\beta}_{\alpha}}{\pi_k}\\
&=\frac{1}{N_v}\sum_{k\in S_{v}} \frac{r_k\left(y_k-\bx_{k,\alpha}^{\top}\pmb{\beta}_{\alpha}\right)+r_k\pi_k\pmb{\gamma}^{\top}\bx_{k,\alpha}\left(y_k-\bx_{k,\alpha}^{\top}\pmb{\beta}_{\alpha}\right)+\bx_{k,\alpha}^{\top}\pmb{\beta}_{\alpha}}{\pi_k}\\
&=\frac{1}{N_v}\sum_{k\in S_{v}} \frac{r_k\left(1+\pi_k\pmb{\gamma}^{\top}\bx_{k,\alpha}\right)\left(y_k-\bx_{k,\alpha}^{\top}\pmb{\beta}_{\alpha}\right)+\bx_{k,\alpha}^{\top}\pmb{\beta}_{\alpha}}{\pi_k}.
\end{align*}
Recall that
\begin{equation*}
\widetilde{\mu}_{\alpha,v}=\frac{1}{N_v}\sum_{k\in S_{v}} \frac{r_k\left(1+\pi_k\textbf{c}_{\alpha,v}^{\top}\bx_{k,\alpha}\right)\left(y_k-\bx_{k,\alpha}^{\top}\pmb{\beta}_{\alpha}\right)+\bx_{k,\alpha}^{\top}\pmb{\beta}_{\alpha}}{\pi_k}.
\end{equation*}
Taking $\pmb{\gamma}^{\star}=\textbf{c}_{\alpha,v}$ leads to $\widetilde{\mu}_v(\textbf{c}_{\alpha,v},\pmb{\beta}_{\alpha})=\widetilde{\mu}_{\alpha,v}$, we finally conclude $\widehat{\mu}_{\alpha,v}-\widetilde{\mu}_{\alpha,v}=o_{\mathbb{P}}(n_v^{-1/2}).$
\subsection*{Proof of Theorem \ref{Theo4}.} \label{proof_Theo4}
 Write
\begin{equation*}
n_v\left|\widehat{V}_{T,v}(\alpha)-V_{T,v}(\alpha)\right|\leq n_v\left|\widehat{V}_{1,v}(\alpha)-V_{1,v}(\alpha)\right|+n_v\left|\widehat{V}_{2,v}(\alpha)-V_{2,v}(\alpha)\right|:= (a) + (b).
\end{equation*}
\paragraph{Consistency of (a).} We further decompose $n_v|\widehat{V}_1(\alpha)-V_1(\alpha)|$ as
\begin{align*}
n_v\left|\widehat{V}_{1,v}(\alpha)-V_{1,v}(\alpha)\right|&\leq n_v\left|\widehat{V}_{1,v}(\alpha)-\widetilde{V}_{1,v}(\alpha)\right|+n_v\left|\widetilde{V}_{1,v}(\alpha)-\bar{V}_{1,v}(\alpha)\right|\\
&\quad+n_v\left|\bar{V}_{1,v}(\alpha)-V_{1,v}(\alpha)\right|\\
&:=A_v+B_v+C_v,
\end{align*}
where
\begin{equation*}
\widetilde{V}_{1,v}(\alpha)=\frac{1}{N_v^2}\sum_{k \in S_v}\sum_{l \in S_v}\frac{\eta_{k,\alpha}}{\pi_k}\frac{\eta_{l,\alpha}}{\pi_l}\frac{\Delta_{kl}}{\pi_{kl}},
\end{equation*}
\begin{equation}\label{BarV1}
\bar{V}_{1,v}(\alpha)=\mathbb{V}_p\left(\tilde{\mu}_{\alpha,v}-\mu_v\right)=\frac{1}{N_v^2}\sum_{k \in U_v}\sum_{l \in U_v}\frac{\eta_{k,\alpha}}{\pi_k}\frac{\eta_{l,\alpha}}{\pi_l}\Delta_{kl},
\end{equation}

\paragraph{Consistency of $A_v$.} Let 
\begin{equation*}
c_{kl}=\frac{\Delta_{kl}I_kI_l}{\pi_k\pi_l\pi_{kl}},\qquad k,l\in U_v.
\end{equation*}
Using the equality, for arbitrary $k,l \in U_v$,
\begin{equation*}
\widehat{\eta}_{k,\alpha}\widehat{\eta}_{l,\alpha}-\eta_{k,\alpha}\eta_{l,\alpha}=\left(\widehat{\eta}_{k,\alpha}-\eta_{k,\alpha}\right)\left(\widehat{\eta}_{l,\alpha}-\eta_{l,\alpha}\right)+\eta_{k,\alpha}\left(\widehat{\eta}_{l,\alpha}-\eta_{l,\alpha}\right)+\eta_{l,\alpha}\left(\widehat{\eta}_{k,\alpha}-\eta_{k,\alpha}\right),
\end{equation*}
$A_v$ can be decomposed as
\begin{align*}
A_v&\leq\frac{n_v}{N_v^2}\sum_{k \in U_v}\sum_{l \in U_v}\left|c_{kl}\left(\widehat{\eta}_{k,\alpha}\widehat{\eta}_{l,\alpha}-\eta_{k,\alpha}\eta_{l,\alpha}\right)\right|\\
& \leq \frac{n_v}{N_v^2}\sum_{k \in U_v}\sum_{l \in U_v}\left|c_{kl}\left(\widehat{\eta}_{k,\alpha}-\eta_{k,\alpha}\right)\left(\widehat{\eta}_{l,\alpha}-\eta_{l,\alpha}\right)\right|\\
&\quad+\frac{n_v}{N_v^2}\sum_{k\in U_v}\sum_{l \in U_v}\left|c_{kl}\left(\widehat{\eta}_{k,\alpha}-\eta_{k,\alpha}\right)\eta_{l,\alpha}\right|+\frac{n_v}{N_v^2}\sum_{k\in U_v}\sum_{l \in U_v}\left|c_{kl}\eta_{k,\alpha}\left(\widehat{\eta}_{l,\alpha}-\eta_{l,\alpha}\right)\right|\\&:=A_{1,v}+A_{2,v}+A_{3,v}.
\end{align*}
Next, we will treat $A_{1,v}$ and $A_{2,v}$, and $A_{3,v} $separately.
 For $A_{1,v}$, since Lemma \ref{ConstEtaEst} shows
\begin{equation*}
\sum_{k \in U_v}\frac{\left(\widehat{\eta}_{k,\alpha}-\eta_{k,\alpha}\right)^2}{N_v}=o_{\mathbb{P}}(1),
\end{equation*} 
we have 
\begin{align*}
A_{1,v}&\leq \frac{n_v}{N_v\lambda^2}\sum_{k \in U_v}\frac{\left(\widehat{\eta}_{k,\alpha}-\eta_{k,\alpha}\right)^2}{N_v}\\
&\quad+\frac{n_v}{N_v\lambda^2\lambda^{\star}}\max_{k \neq l \in U_v}|\Delta_{kl}|\sum_{k \in U_v}\sum_{\substack{l \in U_v\\k \neq l}}\frac{|(\widehat{\eta}_{k,\alpha}-\eta_{k,\alpha})(\widehat{\eta}_{l,\alpha}-\eta_{l,\alpha})|}{N_v}\\
&\leq \frac{n_v}{N_v\lambda^2}\sum_{k \in U_v}\frac{(\widehat{\eta}_{k,\alpha}-\eta_{k,\alpha})^2}{N_v}\\
&\quad+\frac{n_v}{N_v\lambda^2\lambda^{\star}}\max_{k \neq l \in U_v}|\Delta_{kl}|\sum_{k \in U_v}\sum_{l \in U_v}\frac{| (\widehat{\eta}_{k,\alpha}-\eta_{k,\alpha})(\widehat{\eta}_{l,\alpha}-\eta_{l,\alpha})|}{N_v}\\
&\leq \left(\frac{n_v}{N_v\lambda^2}+\frac{n_v}{\lambda^2\lambda^{\star}}\max_{k \neq l \in U_v}|\Delta_{kl}|\right)\sum_{k \in U_v}\frac{(\widehat{\eta}_{k,\alpha}-\eta_{k,\alpha})^2}{N_v}=o_{\mathbb{P}}(1).
\end{align*}
By symmetry, for $A_{2,v}$, using $\sum_{k \in U_v}\sum_{l\in U_v}a_kb_l\leq N_v\left(\sum_{k\in U_v}a_k^2\right)^{1/2}\left(\sum_{l\in U_v}b_l\right)^{1/2}$, we have
\begin{align*}
A_{2,v}&\leq \frac{n_v}{N_v\lambda^2}\sum_{k \in U_v}\frac{|(\widehat{\eta}_{k,\alpha}-\eta_{k,\alpha})\eta_{k,\alpha}|}{N_v}+\frac{n_v}{N_v^2\lambda^2\lambda^{\star}}\max_{k \neq l \in U_v}|\Delta_{kl}|\sum_{k \in U_v}\sum_{\substack{l \in U_v\\k \neq l}}|(\widehat{\eta}_{k,\alpha}-\eta_{k,\alpha})\eta_{l,\alpha}|\\
&\leq \frac{n_v}{N_v\lambda^2}\sum_{k \in U_v}\frac{|(\widehat{\eta}_{k,\alpha}-\eta_{k,\alpha})\eta_{k,\alpha}|}{N_v}+\frac{n_v}{N_v^2\lambda^2\lambda^{\star}}\max_{k \neq l \in U_v}|\Delta_{kl}|\sum_{k \in U_v}\sum_{l \in U_v}|(\widehat{\eta}_{k,\alpha}-\eta_{k,\alpha})\eta_{l,\alpha}|\\
&\leq \left(\frac{n_v}{N_v\lambda^2}+\frac{n_v}{\lambda^2\lambda^{\star}}\max_{k \neq l \in U_v}|\Delta_{kl}|\right)\sqrt{\sum_{k \in U_v}\frac{\eta_{k,\alpha}^2}{N_v}}\sqrt{\sum_{k \in U_v}\frac{(\widehat{\eta}_{k,\alpha}-\eta_{k,\alpha})^2}{N_v}}=o_{\mathbb{P}}(1).
\end{align*}
As a result, we have $A_{v}=o_{\mathbb{P}}(1)$. Furthermore, we have $n_v|\widehat{V}_{1,v}(\alpha)-V_{1,v}(\alpha)|=o_{\mathbb{P}}(1).$ 
\paragraph{Consistency of $B_v$.} Write
\begin{align*}
    &\mathbb{E}_p\left[\left(n_v(\widetilde{V}_{1,v}(\alpha)-\bar{V}_{1,v}(\alpha)) \right)^2 \right]\\
    &\quad\leq \frac{2n_v^2}{N_v^4}\sum_{k \in U_v}\sum_{l \in U_v}\frac{1-\pi_k}{\pi_k}\frac{1-\pi_l}{\pi_l}\frac{\eta_{k,\alpha}^2}{\pi_k}\frac{\eta_{l,\alpha}^2}{\pi_l}\Delta_{kl}\\
    &\quad\quad+\frac{2n_v^2}{N_v^4}\sum_{i \in U_v}\sum_{\substack{j \in U_v\\j \neq i}}\sum_{k \in U_v}\sum_{\substack{l \in U_v\\l \neq k}}\frac{\eta_{i,\alpha}\eta_{j,\alpha}\eta_{k,\alpha}\eta_{l,\alpha}}{\pi_{i}\pi_j\pi_{k}\pi_l}\E_p\left[\frac{\left(I_iI_j-\pi_{ij} \right)\left(I_kI_l-\pi_{kl} \right) }{\pi_{ij}\pi_{kl}}\right]\Delta_{ij}\Delta_{kl}\\
    &\quad:=B_{1,v}+B_{2,v}.
\end{align*}
 We will use the same argument of Theorem 3 in \cite{breidt2000local}. For $B_{1,v}$, we have
\begin{align*}
    B_{1,v}\leq \left(\frac{2}{N_v\lambda^3} +\frac{2n_v\max_{k\neq l \in U_v}|\Delta_{kl}|}{N_v\lambda^4}\right)\sum_{k \in U_v}\frac{\eta_{k,\alpha}^4}{N_v}.
\end{align*}
Recall from Lemma \ref{FiniteEta} gives
\begin{equation*}
    \limsup_{v \rightarrow \infty}\sum_{k \in U_v}\frac{\eta_{k,\alpha}^4}{N_v}<\infty,
\end{equation*}
almost surely. It follows that
$B_{1,v}$ converges to 0 almost surely. For $B_{2,v}$, we have
\begin{align*}
    B_{2,v}&\leq L_v+\frac{2\left(n_v\max_{k \neq l \in U_v}|\Delta_{kl}| \right)^2}{\lambda^4\lambda^{\star2}}\\
    &\quad\times\max_{(i,j,k,l)\in D_{4,N_v}}\Bigg|\E_p\left [\frac{\left(I_iI_j-\pi_{ij}\right)\left(I_kI_l-\pi_{kl} \right)}{\pi_{ij}\pi_{kl}}\right] \Bigg|\sum_{k \in U_v}\frac{\eta_{k,\alpha}^4}{N_v}.
\end{align*}
Here, $L_v$ converges to 0 almost surely. It follows that $B_{2,v}$ converges to 0 almost surely. As a result, for any $\epsilon>0$, an application of Chebyshev's inequality gives
\begin{equation*}
     \lim_{ v\rightarrow \infty}\P_p\left(\left|n_v(\widetilde{V}_{1,v}(\alpha)-\bar{V}_{1,v}(\alpha)) \right|>\epsilon \right)=0,
\end{equation*}
almost surely.
It follows that 
\begin{equation*}
     \lim_{ v\rightarrow \infty}\P_{mpq}\left(\left|n_v(\widetilde{V}_{1,v}(\alpha)-\bar{V}_{1,v}(\alpha)) \right|>\epsilon \right)=0,
\end{equation*}
almost surely.
This concludes $B_v=o_\P(1).$
\paragraph{Consistency of $C_v$.} 
$C_v$ can be decomposed as
\begin{equation*}
C_v=n_v\left|\bar{V}_{1,v}(\alpha)-V_{1,v}(\alpha)\right|\leq \left|n_v\bar{V}_{1,v}(\alpha)-c_{1,\alpha}^{\star}\right|+\left|c_{1,\alpha}^{\star}-n_vV_{1,v}(\alpha)\right|,
\end{equation*}
where $c_{1,\alpha}^{\star}=f^{\star}c_{1,\alpha}$ with $f^{\star}=\lim_{v \rightarrow \infty}n_v/N_v$.
On the one hand, we have
\begin{equation*}
    \left|n_v\bar{V}_{1,v}(\alpha)-c_{1,\alpha}^{\star}\right|=o_\P(1),
\end{equation*}
by assumption. The remaining part is to show
\begin{align*}
    \lim_{v \rightarrow \infty}\left|c_{1,\alpha}^{\star}-n_vV_{1,v}(\alpha)\right|=\lim_{v \rightarrow \infty}\left|\mathbb{E}_{mq}[n_v\bar{V}_{1,v}(\alpha)-c_{1,\alpha}^{\star}] \right|=0
\end{align*}
almost surely. We aim to apply Vitali convergence theorem, that is to show that $|n_v\bar{V}_{1,v}(\alpha)-c_{1,\alpha}^{\star}|=o_\P(1)$ and that $n_v\bar{V}_{1,v}(\alpha)$ is uniformly integrable. From this, it will follow that $\lim_{v \rightarrow \infty}\mathbb{E}_{ mq}[|n_v\bar{V}_{1,v}(\alpha)-c_{1,\alpha}^{\star}|]=0$ almost surely. To show the uniform integrability of $n_v\bar{V}_{1,v}(\alpha)$, it suffices to prove that $\E_{mq}[(n_v\bar{V}_{1,v}(\alpha))^2]$ is almost surely bounded, which follows from Lemma \ref{BoundedV1V2}.

\paragraph{Consistency of (b).} The term
$n_v|\hat{V}_{2,v}(\alpha)-V_{2,v}(\alpha)|$ can be decomposed as \begin{equation*}
n_v\left|\hat{V}_{2,v}(\alpha)-V_{2,v}(\alpha)\right|\leq n_v\left|\hat{V}_{2,v}(\alpha)-\bar{V}_{2,v}(\alpha)\right|+n_v\left|\bar{V}_{2,v}(\alpha)-V_{2,v}(\alpha)\right|:=D_v+E_v,
\end{equation*}
where \begin{equation}\label{BarV2}
\bar{V}_{2,v}(\alpha)=\mathbb{V}_m\left(\mathbb{E}_p\left[\tilde{\mu}_{\alpha,v}-\mu_v\right]\right)=\sigma^2\sum_{k \in U_v}\frac{1-r_k+r_k\left(\pi_k\textbf{c}_{\alpha,v}^{\top}\bx_{k,\alpha}\right)^2}{N_v^2}
\end{equation}
\paragraph{Consistency of $D_v$.}
Combining the results $|\hat{\sigma}^2_{\alpha,v}-\sigma^2|=o_{\mathbb{P}}(1)$ in Lemma \ref{SigmaEst} and
\begin{equation*}
\left|\sum_{k\in S_v}\frac{1-r_k+r_k\left(\pi_k\widehat{\textbf{c}}_{\alpha,v}^{\top}\bx_{k,\alpha}\right)^2}{\pi_kN_v}-\sum_{k \in U_v}\frac{1-r_k+r_k\left(\pi_k\textbf{c}_{\alpha,v}^{\top}\bx_{k,\alpha}\right)^2}{N_v} \right|=o_{\mathbb{P}}(1),
\end{equation*}
from \eqref{eq:cOP}, we have
\begin{equation*}
\left|\sum_{k \in U_v}\frac{1-r_k+r_k\left(\pi_k\textbf{c}_{\alpha,v}^{\top}\bx_{k,\alpha}\right)^2}{N_v}\right|\leq 1+\Arrowvert\textbf{c}_{\alpha,v}\Arrowvert_2^2\sum_{k \in U_v}\frac{\Arrowvert\bx_{k,\alpha}\Arrowvert_2^2}{N_v}=\mathcal{O}_\P(1).
\end{equation*}
It follows that
\begin{align*}
n_v\left|\widehat{V}_{2,v}(\alpha)-\bar{V}_{2,v}(\alpha)\right|&\leq\frac{n_v}{N_v}\Bigg|\left(\widehat{\sigma}_{\alpha,v}^2-\sigma^2\right)\Bigg(\sum_{k\in S_v}\frac{1-r_k+r_k(\pi_k\widehat{\textbf{c}}_{\alpha,v}^{\top}\bx_{k,\alpha})^2}{\pi_kN_v}\\
&\quad\quad-\sum_{k \in U_v}\frac{1-r_k+r_k(\pi_k\textbf{c}_{\alpha,v}^{\top}\bx_{k,\alpha})^2}{N_v}\Bigg)\Bigg|\\
&\quad+\frac{n_v}{N_v}\sigma^2\Bigg|\sum_{k\in S_v}\frac{1-r_k+r_k\left(\pi_k\widehat{\textbf{c}}_{\alpha,v}^{\top}\bx_{k,\alpha}\right)^2}{\pi_kN_v}\\
&\quad\quad -\sum_{k \in U_v}\frac{1-r_k+r_k\left(\pi_k\textbf{c}_{\alpha,v}^{\top}\bx_{k,\alpha}\right)^2}{N_v} \Bigg|\\
&\quad+\frac{n_v}{N_v}\left|\sum_{k \in U_v}\frac{1-r_k+r_k\left(\pi_k\textbf{c}_{\alpha,v}^{\top}\bx_{k,\alpha}\right)^2}{N_v}\right|\left|\widehat{\sigma}_{\alpha,v}^2-\sigma^2\right|=o_\P(1).
\end{align*}
\paragraph{Consistency of $E_v$.} The term
$E_v$ can be decomposed as
\begin{equation*}
E_v=n_v\left|\bar{V}_{2,v}(\alpha)-V_{2,v}(\alpha)\right|\leq \left|n_v\bar{V}_{2,v}(\alpha)-c_{2,\alpha}^{\star}\right|+\left|c_{2,\alpha}^{\star}-n_vV_{2,v}(\alpha)\right|,
\end{equation*}
where $c_{2,\alpha}^{\star}=f^{\star}c_{2,\alpha}$ with $f^{\star}=\lim_{v \rightarrow \infty}n_v/N_v$.
On the one hand, we have
\begin{equation*}
    \left|n_v\bar{V}_{2,v}(\alpha)-c_{2,\alpha}^{\star}\right|=o_\P(1),
\end{equation*}
by assumption. The remaining part is to show
\begin{align*}
    \lim_{v \rightarrow \infty}\left|c_{2,\alpha}^{\star}-n_vV_{2,v}(\alpha)\right|=\lim_{v \rightarrow \infty}\left|\mathbb{E}_{q}[n_v\bar{V}_{2,v}(\alpha)-c_{2,\alpha}^{\star}] \right|=0
\end{align*}
almost surely. We proceed again by Vitali's convergence theorem. The uniform integrability of $n_v\bar{V}_{2,v}(\alpha)$ follows from Lemma \ref{BoundedV1V2}.
\subsection*{Proof of Theorem \ref{Theo5}.} \label{proof_Theo5}
Write
\begin{align} \label{eq:proofTheo5}
\dfrac{\widehat{V}_{T,v}(\widehat{\alpha})}{V_{T,v}(\alpha^{\star})} = 1 + \dfrac{n_v \left(\widehat{V}_{T,v}(\widehat{\alpha})-\widehat{V}_{T,v}(\alpha^{\star})\right) }{n_v V_{T,v}(\alpha^{\star})} + \dfrac{n_v \left(\widehat{V}_{T,v}(\alpha^{\star})-V_{T,v}(\alpha^{\star})\right) }{n_v V_{T,v}(\alpha^{\star})} := 1 + A_v + B_v.
\end{align}
It is therefore enough to show that $A_v = o_\mathbb{P}(1)$ and $B_v = o_\mathbb{P}(1)$. Notice that, by assumption, satisfies $n_v \lim_{v\to \infty} V_{T,v}(\alpha^{\star})>0$ so that it is enough to show that the numerators of $A_v$ and $B_v$ converge to $0$ in probability. 
For $A_v$, let $\epsilon>0$ and write 
\begin{align*}
\P_{mpq}&\left( n_v \bigg\rvert\widehat{V}_{T,v}(\widehat{\alpha})-\widehat{V}_{T,v}(\alpha^{\star})\bigg\rvert >\epsilon\right) =     \P_{mpq}\left( n_v \bigg\rvert\widehat{V}_{T,v}(\widehat{\alpha})-\widehat{V}_{T,v}(\alpha^{\star}) \bigg\rvert >\epsilon \ , \ \widehat{\alpha}=\alpha^{\star}\right)  \\&+\P_{mpq}\left( n_v \bigg\rvert\widehat{V}_{T,v}(\widehat{\alpha})-\widehat{V}_{T,v}(\alpha^{\star}) \bigg\rvert >\epsilon \ , \ \widehat{\alpha}\neq\alpha^{\star}\right)  \\
&\leq \P_{mpq}\left( \widehat{\alpha}\neq\alpha^{\star}\right) \xrightarrow[v \to \infty]{}0
\end{align*}
almost surely since the model selection criterion is consistent. 

~\\
For $B_v$, given that $\alpha^{\star} \in \cC$, it follows directly that $B_v = o_\mathbb{P}(1)$ by Theorem \ref{Theo4}.

\subsection*{Proof of Theorem \ref{Theo6}.} \label{proof_Theo6}

To prove the asymptotic normality of $\widehat{\mu}_{\alpha^{\star},v} - \mu_v$, we  verify the conditions of Theorem 2 of \cite{chen2007asymptotic}. We decompose 
\begin{equation}
\widehat{\mu}_{\alpha^{\star},v} - \mu_v =\underbrace{ \dfrac{1}{N_v}\sum_{k\in U_v} \nu_k \eta_{k,\alpha^{\star}}}_{:= U_v} + \underbrace{\dfrac{1}{N_v}\sum_{k \in U_v} \left(\eta_{k,\alpha} - y_k\right)}_{:= V_v},
\end{equation}
with $\nu_k := I_k/\pi_k-1$ for $k\in U_v$.  Define $\cB_v := \sigma \left(\left(\bx_k, r_k, y_k\right)_{k\in U_v}\right).$  We need to verify the three conditions of Theorem 2 of \cite{chen2007asymptotic}, which we label as (i), (ii), and (iii), respectively.

\paragraph{Verification of (i).} We wish to show that $V_v$ is asymptotically normal and $\cB_v$-measurable. Measurability follows immediately by noting, for an arbitrary $k \in U_v$, $$\eta_{k,\alpha}-y_k = \epsilon_k (r_k + r_k \pi_k\textbf{c}_{\alpha^{\star},v}^{\top}\bx_{k,\alpha^\star} -1):= \epsilon_k w_k$$ with $\textbf{c}_{\alpha^{\star},v}$ depending only on the covariates and response indicators. Moreover, note that $\E_m [\eta_{k,\alpha}-y_k ]=0$. We will start by establishing a conditional central limit theorem via the conditional Lyapunov condition. More precisely, we need to show that there exists $\delta>0$ such that 
\begin{equation*}
L_v:=	\dfrac{	1}{\left[ \sum_{k\in U_v} \V_m \left( \dfrac{\eta_k-y_k}{N_v}\right)\right]^{1+\delta/2}}\sum_{k\in U_v} \E_m \left[\left(\dfrac{\eta_k-y_k}{N_v}\right)^{2+\delta}\right] \xrightarrow[v \to \infty]{\mathbb{P}} 0. 
\end{equation*}
Given our assumptions, it is convenient to show it for $\delta = 2.$ We start by noting that $$\sum_{k\in U_v} \V_m \left( \dfrac{\eta_k-y_k}{N_v}\right) = \dfrac{\sigma^2}{N_v^2} \sum_{k\in U_v}\left( (1-r_k) + r_k (\pi_k \mathbf{c}_{\alpha^{\star},v}^\top\bx_{k,\alpha^{\star}})^2\right) \geq \dfrac{\sigma^2}{N_v^2} N_{m,v} \asymp_\P \dfrac{1}{N_v}.$$
Moreover, note that $w_k^4 = 1$ if $r_k = 0$ and $w_k^4 = (
\pi_k \mathbf{c}_{\alpha^{\star},v}^\top\bx_{k,\alpha^{\star}})^4$ otherwise. Thus, using $(a+b)^4 \leq 8 (a^4+b^4)$, Cauchy-Schwartz inequality and \ref{S3}, we get 

\begin{align*}
\sum_{k\in U_v} \E_m \left[ \left(\dfrac{\eta_k-y_k}{N_v}\right)^{4}\right] &\leqslant \dfrac{8M_0}{N_v^4} \sum_{k\in U_v} \left( 1 + \rVert \mathbf{c}_{\alpha^{\star},v} \rVert^4_2 \rVert \mathbf{x}_{k, \alpha^{\star}} \rVert^4_2 \right) = \mathcal{O}_\P\left(\dfrac{1}{N_v^3}\right),
\end{align*}
by using \eqref{eq:cOP} and \ref{S2}. This shows that $L_v = \mathcal{O}_\P\left(N_v^{-1}\right)$ and thus the Lyapunov condition holds. Therefore, (i) with $$\sigma_{1v}^2 := \V_m \left( \sum_{k\in U_v} \dfrac{\eta_{k,\alpha^{\star}}-y_k}{N_v}\right)$$holds. Conditional Gaussianity follows. Moreover, since the asymptotic distribution does not depend on the conditioning, a dominated convergence argument can be used to extend it to an unconditional central limit theorem.

\paragraph{Verification of (ii).} We have
\begin{align*}
&\E \left[ U_v\rvert \cB_v\right] = \dfrac{1}{N_v} \sum_{k\in U_v} \E \left[ \nu_k\rvert \cB_v\right] \eta_{k, \alpha^{\star}} = 0,\\&  \V \left( U_v \rvert \cB_v\right) = \dfrac{1}{N_v^2}\sum_{k\in U_v}\sum_{l \in U_v} \Delta_{kl} \dfrac{\eta_{k, \alpha^{\star}}}{\pi_k}\dfrac{\eta_{l, \alpha^{\star}}}{\pi_l} := \sigma_{2v}^2,
\end{align*}
since, for all $k \in U_v$, $\eta_k$ is $\cB_v$-measurable. By assumption, $\sigma_{2v}^{-1} U_v \rvert \cB_v \xrightarrow[v \to \infty]{\cL}\cN\left(0,1\right)$, which can be equivalently stated as $$ \sup_{t \in \R} \bigg\rvert \P \left(\sigma_{2v}^{-1} U_v \leqslant t \rvert \cB_n\right) - \Phi(t)\bigg\rvert \xrightarrow[n \to \infty]{\P}0,$$ where $\Phi$ denotes cumulative distribution function of a standard normal random variable. Assumption (ii) is verified. 

\paragraph{Verification of (iii).} We have 
\begin{align*}
N_v \V_m \left( \sum_{k\in U_v} \dfrac{\eta_{k,\alpha^{\star}}-y_k}{N_v}\right) &= \dfrac{\sigma^2}{N_v}\sum_{k\in U_v} (1-r_k) +  \sigma^2\mathbf{c}_{\alpha^{\star},v}^\top \dfrac{1}{N_v}\sum_{k\in U_v} r_k \pi_k^2 \bx_{k,\alpha^{\star}}\bx_{k,\alpha^{\star}}^\top \mathbf{c}_{\alpha^{\star},v}\\
&\xrightarrow[v \to \infty]{\P}\sigma^2 \left\{ \E \left[ 1-p(\bx_{\alpha^{\star}}) \right] + \mathbf{\overline{c}_{\alpha^{\star}}}^\top\boldsymbol{C}_2\mathbf{\overline{c}_{\alpha^{\star}}}\right\}
\end{align*}
where $$\mathbf{\overline{c}_{\alpha^\star}} = \boldsymbol{C}_1^{-1} \E_\bx \left[ (1-p(\bx)) \bx_{\alpha^{\star}}\right].$$
Thus, we have $$ \dfrac{\sigma_{1v}^2}{\sigma_{2v}^2} =  \xrightarrow[v \to \infty]{\P} \dfrac{\sigma^2 \left\{ \E \left[ 1-p(\bx_{\alpha^{\star}}) \right] + \mathbf{\overline{c}_{\alpha^{\star}}}^\top\boldsymbol{C}_2\mathbf{\overline{c}_{\alpha^{\star}}}\right\}}{c_3}.$$
Therefore, (iii) follows.

\paragraph{Putting pieces together.} Since the three conditions of \cite{chen2007asymptotic} are satisfied, it follows that $\widehat{\mu}_{\alpha^{\star},v} - \mu_v$ is asymptotically normal.  Moreover, by Theorem \ref{Theo3}, the asymptotic distribution of $\widehat{\mu}_{\widehat{\alpha},v} - \mu_v$ is the same as that of $\widehat{\mu}_{\alpha^{\star},v} - \mu_v$, and is thus asymptotically normal as well. An application of Theorem \ref{Theo5} proves the unit variance.

\subsection{Consistency of estimators}
\subsubsection{Consistency of single estimators}
\begin{lemma}\label{BetaEst}
Let $\alpha \in \cC$ and $\left(\widehat{\pmb{\beta}}_{\alpha,v}\right)_{v\in \mathbb{N}}$ be a sequence of least-squares estimators given by \eqref{ols}.  Assume \ref{S1}-\ref{S3} and \ref{D1}-\ref{D3}. Then,
\begin{equation*}
\big\rVert\widehat{\pmb{\beta}}_{\alpha,v}-\pmb{\beta}_{\alpha}\big\rVert_2=\mathcal{O}_\P\left(\frac{1}{\sqrt{n_v}} \right).
\end{equation*}
\end{lemma}
    \begin{proof}
Let $$\widetilde{\pmb{\beta}}_{\alpha,v}=\bigg(\sum_{k \in U_v}\frac{p(\bx_k)\pi_k\bx_{k,\alpha}\bx_{k,\alpha}^{\top}}{N_v}\bigg)^{-1}\sum_{k \in U_v}\frac{p(\bx_k)\pi_k\bx_ky_k}{N_v}.$$ We decompose $\widehat{\pmb{\beta}}_{\alpha,v}-\pmb{\beta}_{\alpha}$ as
\begin{equation*}
\widehat{\pmb{\beta}}_{\alpha,v}-\pmb{\beta}_{\alpha}=\underbrace{\widehat{\pmb{\beta}}_{\alpha,v}-\widetilde{\pmb{\beta}}_{\alpha,v}}_{(a)}+\underbrace{\widetilde{\pmb{\beta}}_{\alpha,v}-\pmb{\beta}_{\alpha}}_{(b)}.
\end{equation*}
We need to show that both terms (a) and (b) converge to 0 with rate $\mathcal{O}_\P(n_v^{-1/2})$.
\vspace{4mm}
~\\
\textbf{Convergence of (a).} The triangle inequality  gives
\begin{align}\label{decomp}
\big\Arrowvert\widehat{\pmb{\beta}}_{\alpha,v}-\widetilde{\pmb{\beta}}_{\alpha,v}\big\Arrowvert_2&=\Bigg\Arrowvert\left(\frac{\boldsymbol{A}_{r,\alpha}}{N_v}\right)^{-1}\sum_{k \in S_{r,v}}\frac{\bx_{k,\alpha}y_k}{N_v} \\
&\quad\quad-\left(\sum_{k \in U_v}\frac{p(\bx_k)\pi_k\bx_{k,\alpha}\bx_{k,\alpha}^{\top}}{N_v}\right)^{-1}\sum_{k \in U_v}\frac{p(\bx_k)\pi_k\bx_{k,\alpha}y_k}{N_v}\Bigg\Arrowvert_2\nonumber\\
&= \Bigg\Arrowvert \left(\frac{\boldsymbol{A}_{r,\alpha}}{N_v} \right)^{-1}\left( \sum_{ k\in S_{r,v}}\frac{\bx_{k,\alpha}y_k}{N_v}-\sum_{k \in U_v}\frac{p(\bx_k)\pi_k\bx_{k,\alpha}y_k}{N_v}\right) \nonumber\\
&\quad\quad+  \left(\left(\frac{\boldsymbol{A}_{r,\alpha}}{N_v} \right)^{-1}- \left(\sum_{k \in U_v} \frac{p(\bx_k)\pi_k\bx_{k,\alpha}\bx_{k,\alpha}^{\top}}{N_v} \right)^{-1}\right) \sum_{k \in U_v}\frac{p(\bx_k)\pi_k\bx_{k,\alpha}y_k}{N_v}\Bigg\Arrowvert_2
\nonumber\\&\leq \Bigg\Arrowvert \left(\frac{\boldsymbol{A}_{r,\alpha}}{N_v} \right)^{-1}\left( \sum_{ k\in S_{r,v}}\frac{\bx_{k,\alpha}y_k}{N_v}-\sum_{k \in U_v}\frac{p(\bx_k)\pi_k\bx_{k,\alpha}y_k}{N_v}\right)\Bigg\Arrowvert_2 \nonumber\\
&\quad\quad+ \Bigg\Arrowvert \left(\left(\frac{\boldsymbol{A}_{r,\alpha}}{N_v} \right)^{-1}- \left(\sum_{k \in U_v} \frac{p(\bx_k)\pi_k\bx_{k,\alpha}\bx_{k,\alpha}^{\top}}{N_v} \right)^{-1}\right) \\
&\quad\quad\times \sum_{k \in U_v}\frac{p(\bx_k)\pi_k\bx_{k,\alpha}y_k}{N_v}\Bigg\Arrowvert_2 \nonumber\\
&:=A_{1,v}+A_{2,v} \nonumber
\end{align}
Next, we consider $A_{1,v}$ and $A_{2,v}$ separately.
\vspace{2mm}
~\\
\textbf{Treatment of $A_{1,v}$.} For $A_{1,v}$, using the \hyperlink{SMI}{Schwarz matrix inequality}, we obtain 
\begin{align*}
A_{1,v} &\leq \Bigg\Arrowvert \left(\frac{\boldsymbol{A}_{r,\alpha}}{N_v} \right)^{-1}\Bigg\Arrowvert_{op} \Bigg\Arrowvert  \sum_{ k\in S_{r,v}}\frac{\bx_{k,\alpha}y_k}{N_v}-\sum_{k \in U_v}\frac{p(\bx_k)\pi_k\bx_{k,\alpha}y_k}{N_v}\Bigg\Arrowvert_2
\end{align*}
We will apply Lemma \ref{HTEstimator} to obtain
\begin{equation}\label{VectorXY}
    \Bigg\Arrowvert  \sum_{ k\in S_{r,v}}\frac{\bx_{k,\alpha}y_k}{N_v}-\sum_{k \in U_v}\frac{p(\bx_k)\pi_k\bx_{k,\alpha}y_k}{N_v}\Bigg\Arrowvert_2=\mathcal{O}_\P\left(\frac{1}{\sqrt{N_v}} \right).
\end{equation}
 Setting $w_k:=r_k$, $\boldsymbol{T}_k:=\bx_{k,\alpha}y_k$, as $\{r_k\}_{k\in U_v}$ is i.i.d. random variables with $\E_q[r_k^2]=p_k\leq 1$ almost surely, it remains to verify
\begin{equation*}
    \limsup_{v \rightarrow \infty}\sum_{k \in U_v}\frac{\E_m\left[\Arrowvert\textbf{x}_{k,\alpha}y_k\Arrowvert_2^2\right]}{N_v}<\infty
\end{equation*}
almost surely. Recall from \ref{S2} that the covariates have bounded support, from which we get
\begin{align*}
    \sum_{k \in U_v}\frac{\E_m\left[\Arrowvert\textbf{x}_{k,\alpha}y_k\Arrowvert_2^2\right]}{N_v}&=\sum_{k \in U_v}\frac{\Arrowvert\textbf{x}_{k,\alpha}\Arrowvert_2^2\E_m\left[y_k^2\right]}{N_v}=\sum_{k \in U_v}\frac{\Arrowvert\textbf{x}_{k,\alpha}\Arrowvert_2^2\E_m\left[\left(\textbf{x}_{k,\alpha}^{\top}\pmb{\beta}_{\alpha}+\epsilon_k \right)^2\right]}{N_v}\\
    &\leq C_0^2\sum_{k \in U_v}\frac{2\sigma^2+2\left(\textbf{x}_{k,\alpha}^{\top}\pmb{\beta}_{\alpha}\right)^2}{N_v}\leq 2C_0^2\sigma^2+2C_0^4\Arrowvert\pmb{\beta}_{\alpha}\Arrowvert_2^2 <\infty.
\end{align*}
It follows that
\begin{equation*}
    \limsup_{v \rightarrow \infty}\sum_{k \in U_v}\frac{\Arrowvert\textbf{x}_{k,\alpha}y_k\Arrowvert_2^2}{N_v}<\infty,
\end{equation*}
almost surely, from which \eqref{VectorXY} follows. On the other hand, Lemma \ref{BoundnessA-1} shows
\begin{equation*}
    \Bigg\Arrowvert \left(\frac{\boldsymbol{A}_{r,\alpha}}{N_v} \right)^{-1}\Bigg\Arrowvert_{op}=\mathcal{O}_\P(1).
\end{equation*}
As a result, we obtain $A_{1,v}=\mathcal{O}_\P(n_v^{-1/2}).$
\vspace{2mm}
~\\
\textbf{Treatment of $A_{2,v}$.}
The term $A_{2,v}$ can be decomposed as
\begin{equation*}
A_{2,v}\leq \Bigg\Arrowvert \left(\frac{\boldsymbol{A}_{r,\alpha}}{N_v} \right)^{-1}- \left(\sum_{k \in U_v} \frac{p(\bx_k)\pi_k\bx_{k,\alpha}\bx_{k,\alpha}^{\top}}{N_v} \right)^{-1}\Bigg\Arrowvert_{op}\Bigg\Arrowvert \sum_{k \in U_v}\frac{p(\bx_k)\pi_k\bx_{k,\alpha}y_k}{N_v}\Bigg\Arrowvert_2.
\end{equation*}
On the one hand, using Lemma \ref{HTEstimator} with $w_k:=r_k$ and  $\boldsymbol{T}_k:=\textbf{x}_{k,\alpha}\textbf{x}_{k,\alpha}^{\top}$. Since
\begin{equation}\label{FiniteXX}
     \limsup_{v \rightarrow \infty}\sum_{k \in U_v}\frac{\E_m\left[\Arrowvert \textbf{x}_{k,\alpha}\textbf{x}_{k,\alpha}^{\top}\Arrowvert_F^2\right]}{N_v}= \limsup_{v \rightarrow \infty}\sum_{k \in U_v}\frac{\Arrowvert\textbf{x}_{k,\alpha}\Arrowvert_2^4}{N_v}\leq C_0^4
\end{equation}
almost surely, it follows that
\begin{equation}\label{HTMatrix}
    \Bigg\Arrowvert \frac{\boldsymbol{A}_{r,\alpha}}{N_v} - \sum_{k \in U_v} \frac{p(\bx_k)\pi_k\bx_{k,\alpha}\bx_{k,\alpha}^{\top}}{N_v} \Bigg\Arrowvert_{op}=\mathcal{O}_\P\left( \frac{1}{\sqrt{n_v}}\right).
\end{equation}
On the other hand, recall that Lemma \ref{BoundnessAp-1} shows
\begin{equation*}
    \Bigg\Arrowvert \left(\sum_{k \in U_v}\frac{p(\bx_k)\pi_k\bx_{k,\alpha}\bx_{k,\alpha}^{\top}}{N_v} \right)^{-1}\Bigg\Arrowvert_{op}=\mathcal{O}_\P(1).
\end{equation*}
An application of Lemma \ref{MatrixInverse} gives
\begin{equation*}
\Bigg\Arrowvert \left(\frac{\boldsymbol{A}_{r,\alpha}}{N_v} \right)^{-1}- \left(\sum_{k \in U_v} \frac{p(\bx_k)\pi_k\bx_{k,\alpha}\bx_{k,\alpha}^{\top}}{N_v} \right)^{-1}\Bigg\Arrowvert_{op}=\mathcal{O}_\P\left( \frac{1}{\sqrt{n_v}}\right).
\end{equation*}
Meanwhile,
\begin{equation*}
\Bigg\Arrowvert \sum_{k \in U_v}\frac{p(\bx_k)\pi_k\bx_{k,\alpha}y_k}{N_v}\Bigg\Arrowvert_2\leq \sum_{k \in U_v}\frac{\Arrowvert\bx_{k,\alpha}\Arrowvert_2|y_k|}{N_v}\leq C_0\sum_{k \in U_v}\frac{\rvert y_k\rvert}{N_v},
\end{equation*}
which implies
\begin{equation*}
    \limsup_{v \rightarrow \infty}\Bigg\Arrowvert \sum_{k \in U_v}\frac{p(\bx_k)\pi_k\bx_{k,\alpha}y_k}{N_v}\Bigg\Arrowvert_2<\infty.
\end{equation*} 
It follows that $A_{2,v}=\mathcal{O}_\P(n_v^{-1/2})$ and, therefore,
\begin{equation*}
\big\Arrowvert\widehat{\pmb{\beta}}_{\alpha,v}-\widetilde{\pmb{\beta}}_{\alpha,v}\big\Arrowvert_2=\mathcal{O}_\P\left( \frac{1}{\sqrt{n_v}}\right).
\end{equation*}
~\\
\textbf{Convergence of (b).} Recall from Lemma \ref{FiniteAp-1}, we have
 \begin{equation*}
    \limsup_{v \rightarrow \infty}\Bigg\Arrowvert \left(\sum_{k \in U_v}\frac{p(\bx_k)\pi_k\bx_{k,\alpha}\bx_{k,\alpha}^{\top}}{N_v} \right)^{-2}\Bigg\Arrowvert_{op}<\infty
    \end{equation*}
    almost surely.
An application of Chebyshev's inequality shows that for any $\epsilon>0$, we have
\begin{align}\label{OracleBeta}
&\mathbb{P}_{mpq}\left( \big\Arrowvert\widetilde{\pmb{\beta}}_{\alpha,v} -\pmb{\beta}_{\alpha}\big\Arrowvert_2^2>\epsilon\right)\\
\quad&
\quad\leq\frac{1}{\epsilon^2}\mathbb{E}_{mpq}\left[\left(\widetilde{\pmb{\beta}}_{\alpha} -\pmb{\beta}_{\alpha}\right)^{\top}\left(\widetilde{\pmb{\beta}}_{\alpha} -\pmb{\beta}_{\alpha}\right)\right]\nonumber\\
&\quad=\frac{1}{\epsilon^2}\mathbb{E}_m\left[\sum_{k \in U_v}\frac{\epsilon_k\bx_{k,\alpha}^{\top}p(\bx_k)\pi_k}{N_v}\left(\sum_{k \in U_v}\frac{\pi_kp(\bx_k)\bx_{k,\alpha}\bx_{k,\alpha}^{\top}}{N_v} \right)^{-2}\sum_{k \in U_v}\frac{\epsilon_k\bx_{k,\alpha}p(\bx_k)\pi_k}{N_v}\right]\nonumber\\
&\quad\quad\overset{(*)}{=}\frac{\sigma^2}{\epsilon^2}\sum_{k \in U_v}\frac{\bx_{k,\alpha}^{\top}p(\bx_k)\pi_k}{N_v}\left(\sum_{k \in U_v}\frac{\pi_kp(\bx_k)\bx_{k,\alpha}\bx_{k,\alpha}^{\top}}{N_v} \right)^{-2}\frac{\bx_{k,\alpha}p(\bx_k)\pi_k}{N_v} \nonumber\\
&\quad\overset{(**)}{\leq} \frac{\sigma^2}{N_v\epsilon^2}\sum_{k \in U_v}\frac{\Arrowvert\bx_{k,\alpha}\Arrowvert_2^2}{N_v}\Bigg\Arrowvert\left(\sum_{k \in U_v}\frac{\pi_kp(\bx_k)\bx_{k,\alpha}\bx_{k,\alpha}^{\top}}{N_v} \right)^{-2}\Bigg\Arrowvert_{op}\nonumber\\
&\quad\leq \frac{\sigma^2C_0^2}{N_v\epsilon^2}\Bigg\Arrowvert\left(\sum_{k \in U_v}\frac{\pi_kp(\bx_k)\bx_{k,\alpha}\bx_{k,\alpha}^{\top}}{N_v} \right)^{-2}\Bigg\Arrowvert_{op}.
\end{align}
Here, (*) is due to the fact that $\mathbb{E}_m[\epsilon_k\epsilon_l]=0$ for $k \neq l \in U_v$ and (**) from the inequality $\bz^{\top}\boldsymbol{A}\bz\leq \Arrowvert\bz\Arrowvert_2^2\Arrowvert\boldsymbol{A}\Arrowvert_{op}$.
The last line of \eqref{OracleBeta} converges to 0 almost surely.  Putting all things together, we obtain
\begin{equation*}
\big\Arrowvert\widehat{\pmb{\beta}}_{\alpha,v} -\pmb{\beta}_{\alpha}\big\Arrowvert_2=\mathcal{O}_\P\left( \frac{1}{\sqrt{n_v}}\right).
\end{equation*}
\end{proof}
\begin{lemma}\label{SMatrix}
Assume \ref{S1}-\ref{S2} and \ref{D1}-\ref{D3}. If \eqref{wrong} holds, then there exists a constant $K_0$ such that
\begin{equation*}
\lim_{v \rightarrow \infty}\mathbb{P}_{mpq}\left(\lambda_{min}(N_v^{-1}\textbf{S}_v)\geq K_0\right)=1
\end{equation*}
almost surely.
\end{lemma}

\begin{proof}
\textbf{Step 1.} Prove that
\begin{equation*}
\Arrowvert N_v^{-1}\textbf{S}_v-\textbf{S}_{U_v}\Arrowvert_{op}=\mathcal{O}_\P(n_v^{-1/2}).
\end{equation*}
Write
\begin{align*}
&\Arrowvert N_v^{-1}\textbf{S}_v-\textbf{S}_{U_v}\Arrowvert_{op}\\&\quad=\Bigg\Arrowvert \frac{1}{N_v}\textbf{A}_{r,\alpha^c}\\
&\quad\quad-\sum_{k \in S_{r,v}}\frac{\bx_{k,\alpha^c}\bx_{k,\alpha}^{\top}}{N_v}\left(\frac{\textbf{A}_{r,\alpha}}{N_v} \right)^{-1} \sum_{k \in S_{r,v}} \frac{\bx_{k,\alpha}\bx_{k,\alpha^c}^{\top}}{N_v}-\Bigg(\sum_{k \in U_v}\frac{\bx_{k,\alpha^c}\bx_{k,\alpha^c}^{\top}p(\bx_k)\pi_k}{N_v}\\&\quad\quad\quad-\sum_{k \in U_v}\frac{\bx_{k,{\alpha}}\bx_{k,\alpha^c}^{\top}p(\bx_k)\pi_k}{N_v}\left(\sum_{k \in U_v}\frac{\bx_{k,\alpha}\bx_{k,\alpha}^{\top}p(\bx_k)\pi_k}{N_v} \right)^{-1}\sum_{k \in U_v} \frac{\bx_{k,{\alpha}}\bx_{k,\alpha^c}^{\top}p(\bx_k)\pi_k}{N_v} \Bigg)\Bigg\Arrowvert_{op}\\
&\quad\leq \Bigg\Arrowvert\frac{1}{N_v}\textbf{A}_{r,\alpha^c}- \sum_{k \in U_v}\frac{\bx_{k,\alpha^c}\bx_{k,\alpha^c}^{\top}p(\bx_k)\pi_k}{N_v}\Bigg\Arrowvert_{op}\\
&\quad\quad +\Bigg\Arrowvert\sum_{k \in S_{r,v}}\frac{\bx_{k,\alpha^c}\bx_{k,\alpha}^{\top}}{N_v}\left(\frac{\textbf{A}_{r,\alpha}}{N_v} \right)^{-1} \sum_{k \in S_{r,v}} \frac{\bx_{k,\alpha}\bx_{k,\alpha^c}^{\top}}{N_v}\\&\quad\quad\quad-\sum_{k \in U_v}\frac{\bx_{k,\alpha^c}\bx_{k,\alpha}^{\top}p(\bx_k)\pi_k}{N_v}\left(\sum_{k \in U_v}\frac{\bx_{k,\alpha}\bx_{k,\alpha}^{\top}p(\bx_k)\pi_k}{N_v} \right)^{-1}\sum_{k \in U_v} \frac{\bx_{k,{\alpha}}\bx_{k,\alpha^c}^{\top}p(\bx_k)\pi_k}{N_v} \Bigg\Arrowvert_{op}\\
&\quad:=A_v+B_v.
\end{align*}
 Recall that \eqref{HTMatrix} shows $A_v=\mathcal{O}_\P(n_v^{-1/2})$. For $B_v$, we decompose
 
\begin{align*}
B_v&\leq \Bigg\Arrowvert\sum_{ k\in S_{r,v}}\frac{\bx_{k,\alpha^c}\bx_{k,\alpha}^{\top}}{N_v}\Bigg( \left(\frac{\textbf{A}_{r,\alpha}}{N_v}\right)^{-1}\sum_{k \in S_{r,v}}\frac{\bx_{k,\alpha}\bx_{k,\alpha^c}^{\top}}{N_v}\\
&\quad\quad-\left(\sum_{k \in U_v}\frac{p(\bx_k)\pi_k\bx_{k,\alpha}\bx_{k,\alpha}^{\top}}{N_v} \right)^{-1}\sum_{k \in U_v} \frac{p(\bx_k)\pi_k\bx_{k,{\alpha}}\bx_{k,\alpha^c}^{\top}}{N_v}\Bigg)\Bigg\Arrowvert_{op}\\
&\quad+\Bigg\Arrowvert\left( \sum_{k\in S_{r,v}}\frac{\bx_{k,\alpha^c}\bx_{k,\alpha}^{\top}}{N_v}-\sum_{k\in U_v}\frac{p(\bx_k)\pi_k\bx_{k,\alpha^c}\bx_{k,\alpha}^{\top}}{N_v}\right)\\
&\quad\quad\times\left(\sum_{k \in U_v}\frac{p(\bx_k)\pi_k\bx_{k,\alpha}\bx_{k,\alpha}^{\top}}{N_v} \right)^{-1}\sum_{k \in U_v} \frac{p(\bx_k)\pi_k\bx_{k,{\alpha}}\bx_{k,\alpha^c}^{\top}}{N_v}\Bigg\Arrowvert_{op}\\
&:=B_{1,v}+B_{2,v}.
\end{align*}
We need to prove $B_{1,v}$ and $B_{2,v}$ converges to 0 with $\mathcal{O}_\P(n_v^{-1/2})$.
\vspace{2mm}
~\\
\textbf{Convergence of $B_{1,v}$.} We have

\begin{align*}
B_{1,v}&\leq \Bigg\Arrowvert\sum_{k \in S_{r,v}}\frac{\bx_{k,\alpha^c}\bx_{k,\alpha}^{\top}}{N_v}\Bigg\Arrowvert_{op}\\
&\quad\times\Bigg\Arrowvert\left(\frac{\textbf{A}_{r,\alpha}}{N_v}\right)^{-1}\sum_{k \in S_{r,v}}\frac{\bx_{k,\alpha}\bx_{k,\alpha^c}^{\top}}{N_v}\\
&\quad\quad-\left(\sum_{k \in U_v}\frac{p(\bx_k)\pi_k\bx_{k,\alpha}\bx_{k,\alpha}^{\top}}{N_v} \right)^{-1}\sum_{k \in U_v} \frac{p(\bx_k)\pi_k\bx_{k,{\alpha}}\bx_{k,\alpha^c}^{\top}}{N_v} \Bigg\Arrowvert_2.
\end{align*}
Using \ref{S2}, we get
\begin{align*}
    \Bigg\Arrowvert\sum_{k \in S_{r,v}}\frac{\bx_{k,\alpha^c}\bx_{k,\alpha}^{\top}}{N_v}\Bigg\Arrowvert_{op}\leq \sum_{k \in U_v}\frac{\Arrowvert\bx_{k,\alpha}\Arrowvert_2\Arrowvert\bx_{k,\alpha^c}\Arrowvert_2}{N_v}\leq C_0^2,
\end{align*}
almost surely.
Next, using the \hyperlink{SMI}{Schwarz matrix inequality}, we decompose
\begin{align}\label{DecompBv}
&\Bigg\Arrowvert\left(\frac{\textbf{A}_{r,\alpha}}{N_v}\right)^{-1}\sum_{k \in S_{r,v}}\frac{\bx_{k,\alpha}\bx_{k,\alpha^c}^{\top}}{N_v}-\left(\sum_{k \in U_v}\frac{p(\bx_k)\pi_k\bx_{k,\alpha}\bx_{k,\alpha}^{\top}}{N_v} \right)^{-1}\sum_{k \in U_v} \frac{p(\bx_k)\pi_k\bx_{k,{\alpha}}\bx_{k,\alpha^c}^{\top}}{N_v} \Bigg\Arrowvert_{op}\nonumber\\
&\quad\leq \Bigg\Arrowvert\left(\frac{\textbf{A}_{r,\alpha}}{N_v}\right)^{-1}\left(\sum_{k \in S_{r,v}}\frac{\bx_{k,\alpha}\bx_{k,\alpha^c}^{\top}}{N_v}-\sum_{k \in U_v}\frac{p(\bx_k)\pi_k\bx_{k,\alpha}\bx_{k,\alpha^c}^{\top}}{N_v}\right)\Bigg\Arrowvert_{op}\nonumber\\
&\quad\quad+\Bigg\Arrowvert\left(\left(\frac{\textbf{A}_{r,\alpha}}{N_v}\right)^{-1}-\left(\sum_{k \in U_v}\frac{p(\bx_k)\pi_k\bx_{k,\alpha}\bx_{k,\alpha}^{\top}}{N_v} \right)^{-1}\right)\sum_{k \in U_v}\frac{p(\bx_k)\pi_k\bx_{k,\alpha}\bx_{k,\alpha^c}^{\top}}{N_v}\Bigg\Arrowvert_{op}\nonumber\\
&\quad\leq \Bigg\Arrowvert\left(\frac{\textbf{A}_{r,\alpha}}{N_v}\right)^{-1}\Bigg\Arrowvert_{op}\Bigg\Arrowvert\sum_{k \in S_{r,v}}\frac{\bx_{k,\alpha}\bx_{k,\alpha^c}^{\top}}{N_v}-\sum_{k \in U_v}\frac{p(\bx_k)\pi_k\bx_{k,\alpha}\bx_{k,\alpha^c}^{\top}}{N_v}\Bigg\Arrowvert_{op}\nonumber\\
&\quad\quad+\Bigg\Arrowvert\left(\frac{\textbf{A}_{r,\alpha}}{N_v}\right)^{-1}-\left(\sum_{k \in U_v}\frac{p(\bx_k)\pi_k\bx_{k,\alpha}\bx_{k,\alpha}^{\top}}{N_v} \right)^{-1}\Bigg\Arrowvert_{op}\Bigg\Arrowvert\sum_{k \in U_v}\frac{p(\bx_k)\pi_k\bx_{k,\alpha}\bx_{k,\alpha^c}^{\top}}{N_v}\Bigg\Arrowvert_{op}.
\end{align}
Next, we evaluate every component of \eqref{DecompBv}. Lemma \ref{BoundnessAp-1} and Lemma \ref{BoundnessA-1} show that
\begin{equation*}
    \Bigg\Arrowvert \left(\sum_{k \in U_v}\frac{p(\bx_k)\pi_k\bx_{k,\alpha}\bx_{k,\alpha}^{\top}}{N_v} \right)^{-1}\Bigg\Arrowvert_{op}=\mathcal{O}_\P(1),\qquad\qquad\Bigg\Arrowvert\left(\frac{\textbf{A}_{r,\alpha}}{N_v}\right)^{-1}\Bigg\Arrowvert_{op}=\mathcal{O}_\P(1),
\end{equation*}
respectively.
We proceed to apply Lemma \ref{HTEstimator} with $w_k:=r_k$ and $\boldsymbol{T}_k=\textbf{x}_{k,\alpha}\textbf{x}_{k,\alpha^c}^{\top}$. Since
\begin{align*}
     \limsup_{v \rightarrow \infty}\sum_{k \in U_v}\frac{\E_m\left[\Arrowvert\textbf{x}_{k,\alpha}\textbf{x}^{\top}_{k,\alpha^c}\Arrowvert_F^2\right]}{N_v}= \limsup_{v \rightarrow \infty}\sum_{k \in U_v}\frac{\Arrowvert\textbf{x}_{k,\alpha}\Arrowvert_2^2\Arrowvert\textbf{x}_{k,\alpha^c}\Arrowvert_2^2}{N_v}\leq C_0^4
\end{align*}
almost surely, it follows that
\begin{equation*}
    \Bigg\Arrowvert\sum_{k \in S_{r,v}}\frac{\bx_{k,\alpha}\bx_{k,\alpha^c}^{\top}}{N_v}-\sum_{k \in U_v}\frac{p(\bx_k)\pi_k\bx_{k,\alpha}\bx_{k,\alpha^c}^{\top}}{N_v}\Bigg\Arrowvert_{op}=\mathcal{O}_\P\left(\frac{1}{\sqrt{n_v}}\right).
\end{equation*}
Recall again from \eqref{HTMatrix} that
\begin{equation*}
    \Bigg\Arrowvert\frac{\textbf{A}_{r,\alpha}}{N_v}-\sum_{k \in U_v}\frac{p(\bx_k)\pi_k\bx_{k,\alpha}\bx_{k,\alpha}^{\top}}{N_v} \Bigg\Arrowvert_{op}=\mathcal{O}_\P\left(\frac{1}{\sqrt{n_v}}\right).
\end{equation*}
so that, using Lemma \ref{MatrixInverse} gives
\begin{equation*}
    \Bigg\Arrowvert\left(\frac{\textbf{A}_{r,\alpha}}{N_v}\right)^{-1}-\left(\sum_{k \in U_v}\frac{p(\bx_k)\pi_k\bx_{k,\alpha}\bx_{k,\alpha}^{\top}}{N_v} \right)^{-1}\Bigg\Arrowvert_{op}=\mathcal{O}_\P\left(\frac{1}{\sqrt{n_v}}\right).
\end{equation*}
Finally, by noting
\begin{align*}
    \Bigg\Arrowvert\sum_{k \in U_v}\frac{p(\bx_k)\pi_k\bx_{k,\alpha}\bx_{k,\alpha^c}^{\top}}{N_v}\Bigg\Arrowvert_{op}\leq \sum_{k \in U_v}\frac{p(\bx_k)\pi_k\Arrowvert\bx_{k,\alpha}\Arrowvert_2\Arrowvert\bx_{k,\alpha^c}\Arrowvert }{N_v} \leq C_0^2
\end{align*}
almost surely, we conclude that \eqref{DecompBv} is $\mathcal{O}_\P(n_v^{-1/2}).$
\vspace{2mm}
~\\
\textbf{Treatment of $B_{2,v}$.}
Using the \hyperlink{SMI}{Schwarz matrix inequality}, we have
\begin{align*}
B_{2,v}&=\Bigg\Arrowvert\left( \sum_{k\in S_{r,v}}\frac{\bx_{k,\alpha^c}\bx_{k,\alpha}^{\top}}{N_v}-\sum_{k\in U_v}\frac{p(\bx_k)\pi_k\bx_{k,\alpha^c}\bx_{k,\alpha}^{\top}}{N_v}\right)\\
&\quad\times\left(\sum_{k \in U_v}\frac{p(\bx_k)\pi_k\bx_{k,\alpha}\bx_{k,\alpha}^{\top}}{N_v} \right)^{-1}\sum_{k \in U_v} \frac{p(\bx_k)\pi_k\bx_{k,{\alpha}}\bx_{k,\alpha^c}^{\top}}{N_v}\Bigg\Arrowvert_{op}\\
&\leq \Bigg\Arrowvert \sum_{k\in S_{r,v}}\frac{\bx_{k,\alpha^c}\bx_{k,\alpha}^{\top}}{N_v}-\sum_{k\in U_v}\frac{p(\bx_k)\pi_k\bx_{k,\alpha^c}\bx_{k,\alpha}^{\top}}{N_v}\Bigg\Arrowvert_{op}\\
&\quad\times\Bigg\Arrowvert\left(\sum_{k \in U_v}\frac{p(\bx_k)\pi_k\bx_{k,\alpha}\bx_{k,\alpha}^{\top}}{N_v} \right)^{-1}\Bigg\Arrowvert_{op}\Bigg\Arrowvert\sum_{k \in U_v} \frac{p(\bx_k)\pi_k\bx_{k,{\alpha}}\bx_{k,\alpha^c}^{\top}}{N_v}\Bigg\Arrowvert_{op}=\mathcal{O}_\P\left(\frac{1}{\sqrt{n_v}}\right).
\end{align*}
Putting all things together, we obtain 
\begin{equation*}
\Arrowvert N_v^{-1}\textbf{S}_v-\textbf{S}_{U_v}\Arrowvert_{op}=\mathcal{O}_\P\left(\frac{1}{\sqrt{n_v}}\right).
\end{equation*}
\vspace{4mm}
~\\
\textbf{Step 2.} On the one hand, Step 1 gives for any $\epsilon>0$,
\begin{equation*}
\lim_{v \rightarrow \infty}\mathbb{P}_{mpq}\left(\Arrowvert N_v^{-1}\textbf{S}_v-\textbf{S}_{U_v} \Arrowvert_{op}\leq\epsilon\right)=1
\end{equation*}
almost surely. 
On the other hand, by Weyl's inequality,
\begin{equation*}
\left|\lambda_{min}(N_v^{-1}\textbf{S}_{v})-\lambda_{min}(\textbf{S}_{U_v})\right|\leq \Arrowvert N_v^{-1}\textbf{S}_v-\textbf{S}_{U_v} \Arrowvert_{op},
\end{equation*}
from which it follows that
\begin{equation*}
\lim_{v \rightarrow \infty}\mathbb{P}_{pq}\left(\left|\lambda_{min}(N_v^{-1}\textbf{S}_{v})-\lambda_{min}(\textbf{S}_{U_v})\right|\leq\epsilon\right)=1.
\end{equation*}
As a result, by setting $K_0=C_0-\epsilon$, we have
\begin{align*}
&\lim_{v \rightarrow \infty}\mathbb{P}_{pq}\left(\left|\lambda_{min}(N_v^{-1}\textbf{S}_{v})-\lambda_{min}(\textbf{S}_{U_v})\right|\leq\epsilon\right)\\&\quad=\lim_{v \rightarrow \infty}\mathbb{P}_{pq}\left(-\epsilon\leq\lambda_{min}(N_v^{-1}\textbf{S}_{v})-\lambda_{min}(\textbf{S}_{U_v})\leq\epsilon\right)\\
&\quad\leq \lim_{v \rightarrow \infty}\mathbb{P}_{pq}\left(\lambda_{min}(\textbf{S}_{U_v})-\epsilon\leq\lambda_{min}(N_v^{-1}\textbf{S}_{v})\right)\\
&\quad\leq \lim_{v \rightarrow \infty}\mathbb{P}_{pq}\left(C_0-\epsilon\leq\lambda_{min}(N_v^{-1}\textbf{S}_{v})\right)\\
&\quad= \lim_{v \rightarrow \infty}\mathbb{P}_{pq}\left(K_0\leq\lambda_{min}(N_v^{-1}\textbf{S}_{v})\right)=1
\end{align*}
almost surely. Furthermore, using the Lebesgue convergence theorem, we obtain
\begin{equation*}
\lim_{v \rightarrow \infty}\mathbb{P}_{mpq}\left(\Arrowvert N_v^{-1}\textbf{S}_v-\textbf{S}_{U_v} \Arrowvert_{op}\leq\epsilon\right)=1
\end{equation*}
almost surely. 

\end{proof}
\begin{lemma}\label{SigmaEst}
Let $\alpha \in \cC$ and $(\widehat{\sigma}_{\alpha,v}^2)_{v \in \mathbb{N}}$ be the sequence of the estimator defined in \eqref{sigmaest}. Assume \ref{S1}-\ref{S3} and \ref{D1}-\ref{D3}. We have
\begin{equation*}
|\widehat{\sigma}_{\alpha,v}^2-\sigma^2|=\mathcal{O}_\P\left( \frac{1}{\sqrt{n_v}}\right).
\end{equation*}
\end{lemma}
\begin{proof}
Let 
\begin{equation*}
\widehat{S}_v^2=\frac{1}{n_{r,v}}\sum_{k \in S_{r,v}}\left(y_k-\bx_{k,\alpha}^{\top}\pmb{\beta}_{\alpha}\right)^2.
\end{equation*}
Then, $|\sigma_{\alpha,v}^2-\sigma^2|$ can be decomposed as
\begin{equation*}
\left|\widehat{\sigma}_{\alpha,v}^2-\sigma^2\right|\leq\frac{n_{r,v}}{n_{r,v}-p_{\alpha}}\left|\widehat{S}_v^2-\sigma^2\right|+\frac{p_{\alpha}}{n_{r,v}-p_{\alpha}}\sigma^2.
\end{equation*}
Since $\sigma^2p_{\alpha}/(n_{r,v}-p_{\alpha})=o_\P(n_v^{-1/2})$,
the remaining part is to show $|\widehat{S}_v^2-\sigma|=\mathcal{O}_\P(n_v^{-1/2}).$
~\\
Note that $|\widehat{S}_v^2-\sigma^2|$ can be further decomposed as
\begin{equation*}
|\widehat{S}_v^2-\sigma^2|\leq\underbrace{|\widehat{S}_v^2-\widehat{S}_{r}^2|}_{(a)}+\underbrace{|\widehat{S}_{r}^2-\widetilde{S}_v^2|}_{(b)}+\underbrace{|\widetilde{S}_v^2-\sigma^2|}_{(c)},
\end{equation*}
where 
\begin{equation*}
\widehat{S}_{r}^2=\sum_{k \in S_v}\frac{r_k\epsilon_k^2}{N_v}\left( \frac{n_{r,v}}{N_v}\right)^{-1}
\end{equation*}
and
\begin{equation*}
\widetilde{S}_{v}^2=\sum_{k \in S_v}\frac{p(\bx_k)\pi_k\epsilon_k^2}{N_v}\left( \sum_{k \in U_v}\frac{p(\bx_k)\pi_k}{N_v}\right)^{-1}.
\end{equation*}
We need to show the terms (a), (b), and (c) converge to 0 with rate $\mathcal{O}_\P(n_v^{-1/2})$.
\vspace{4mm}
~\\
\textbf{Convergence of (a).}
$|\widehat{S}_v^2-\widehat{S}_r^2|$ can be expressed as
\begin{align*}
\left|\widehat{S}_v^2-\widehat{S}_r^2\right|&=\left|\frac{N_v}{n_{r,v}}\sum_{k \in S_v}\frac{r_k\left(y_k-\bx_{k,\alpha}^{\top}\widehat{\pmb{\beta}}_{\alpha,v}\right)^2}{N_v}-\frac{N_v}{n_{r,v}}\sum_{k \in S_v}\frac{r_k\left(y_k-\bx_{k,\alpha}^{\top}\pmb{\beta}_{\alpha}\right)^2}{N_v} \right|\\
&=\left|\left(\frac{n_{r,v}}{N_v}\right)^{-1}\right| \underbrace{\left|\sum_{k \in S_v}\frac{r_k\left(y_k-\bx_{k,\alpha}^{\top}\widehat{\pmb{\beta}}_{\alpha,v}\right)^2}{N_v}-\sum_{k \in S_v}\frac{r_k\left(y_k-\bx_{k,\alpha}^{\top}\pmb{\beta}_{\alpha}\right)^2}{N_v}\right|}_{:=A_v}.
\end{align*}
Since Lemma \ref{SetofRespondents} shows
\begin{equation*}
\lim_{v \rightarrow \infty}\mathbb{P}_{mpq}\left(\frac{n_{r,v}}{n_v}\geq \xi \right)=1
\end{equation*}
we deduce $(n_{r,v}/N_v)^{-1}=\mathcal{O}_\P(1).$ Using
\begin{equation*}
\left(\bx_{k,\alpha}^{\top}\widehat{\pmb{\beta}}_{\alpha,v}\right)^2-\left(\bx_{k,\alpha}^{\top}\pmb{{\beta}}_{\alpha}\right)^2=\left(\bx_{k,\alpha}^{\top}\left(\widehat{\pmb{\beta}}_{\alpha,v}-\pmb{\beta}_{\alpha}\right)\right)^2+2\bx_{k,\alpha}^{\top}\pmb{\beta}_{\alpha}\bx_{k,\alpha}^{\top}\left(\widehat{\pmb{\beta}}_{\alpha,v}-\pmb{\beta}_{\alpha}\right),
\end{equation*}
 we decompose $A_v$ as
\begin{align*}
&\left|\sum_{k \in S_v}\frac{r_k\left(y_k-\bx_{k,\alpha}^{\top}\widehat{\pmb{\beta}}_{\alpha,v}\right)^2}{N_v}-\sum_{k \in S_v}\frac{r_k\left(y_k-\bx_{k,\alpha}^{\top}\pmb{\beta}_{\alpha}\right)^2}{N_v}\right|\\&\quad\leq \sum_{k\in U_v}\frac{\left|\left(y_k-\bx_{k,\alpha}^{\top}\widehat{\pmb{\beta}}_{\alpha,v})^2-(y_k-\bx_{k,\alpha}^{\top}\pmb{\beta}_{\alpha}\right)^2\right|}{N_v}\\
&\quad= \sum_{k\in U_v}\frac{\left|-2\bx_{k,\alpha}^{\top}\left(\widehat{\pmb{\beta}}_{\alpha,v}-\pmb{\beta}_{\alpha}\right)+\left(\bx_{k,\alpha}^{\top}\widehat{\pmb{\beta}}_{\alpha,v}\right)^2-\left(\bx_{k,\alpha}^{\top}\pmb{{\beta}}_{\alpha}\right)^2\right|}{N_v}\\
&\quad= \sum_{k\in U_v}\frac{\left|-2\bx_{k,\alpha}^{\top}\left(\widehat{\pmb{\beta}}_{\alpha,v}-\pmb{\beta}_{\alpha}\right)+\left(\bx_{k,\alpha}^{\top}\left(\widehat{\pmb{\beta}}_{\alpha,v}-\pmb{\beta}_{\alpha}\right)^2\right)+2\bx_{k,\alpha}^{\top}\pmb{\beta}_{\alpha}\bx_{k,\alpha}^{\top}\left(\widehat{\pmb{\beta}}_{\alpha,v}-\pmb{\beta}_{\alpha}\right)\right|}{N_v}\\
&\quad\leq \dfrac{2}{N_v}\sum_{k \in U_v}|y_k|\left|\bx_{k,\alpha}^{\top}\left(\widehat{\pmb{\beta}}_{\alpha,v}-\pmb{\beta}_{\alpha}\right)\right|\\
&\quad\quad+\dfrac{1}{N_v}\sum_{k \in U_v}\left(\bx_{k,\alpha}^{\top}\left(\widehat{\pmb{\beta}}_{\alpha,v}-\pmb{\beta}_{\alpha}\right)\right)^2+\dfrac{2}{N_v}\sum_{k \in U_v}\left|\bx_{k,\alpha}^{\top}\pmb{\beta}_{\alpha}\bx_{k,\alpha}^{\top}\left(\widehat{\pmb{\beta}}_{\alpha,v}-\pmb{\beta}_{\alpha}\right)\right|\\
&\quad:=A_{1,v}+A_{2,v}+A_{3,v}.
\end{align*}
 By the law of the large numbers, we have
\begin{align*}
    \sum_{k \in U_v}\frac{y_k^2}{N_v}\xrightarrow[v\rightarrow \infty]{a.s.}\mathbb{E}_{\bx m}[y_1^2]&=\mathbb{E}_{\bx m}\left[\left( \textbf{x}_{1,\alpha}^{\top}\pmb{\beta}_{\alpha}+\epsilon_1\right)^2\right]\\
    &\leq \mathbb{E}_{\bx m}\left[ 2\left(\textbf{x}_{1,\alpha}^{\top}\pmb{\beta}_{\alpha}\right)+\epsilon_1^2\right]=2\sigma^2+2C_0^2\Arrowvert\pmb{\beta}_{\alpha}\Arrowvert_2^2.
\end{align*}
For $A_{1,v}$, using Cauchy-Schwarz inequality, we have
\begin{align*}
A_{1,v}&\leq 2\sqrt{\sum_{k \in U_v}\frac{y_k^2}{N_v}}\sqrt{\sum_{k \in U_v}\frac{\left(\bx_{k,\alpha}^{\top}\left(\widehat{\pmb{\beta}}_{\alpha,v}-\pmb{\beta}_{\alpha}\right)\right)^2}{N_v}}\\&\leq 2\Arrowvert\widehat{\pmb{\beta}}_{\alpha,v}-\pmb{\beta}_{\alpha} \Arrowvert_2\sqrt{\sum_{k \in U_v}\frac{y_k^2}{N_v}}\sqrt{\sum_{k \in U_v}\frac{\Arrowvert\bx_{k,\alpha}\Arrowvert_2^2}{N_v}}\\
&\leq 2C_0\Arrowvert\widehat{\pmb{\beta}}_{\alpha,v}-\pmb{\beta}_{\alpha} \Arrowvert_2\sqrt{\sum_{k \in U_v}\frac{y_k^2}{N_v}}=\mathcal{O}_\P\left(\frac{1}{\sqrt{n_v}}\right),
\end{align*}
which also implies $A_{2,v}=\mathcal{O}_\P(n_v^{-1/2})$. For $A_{3,v}$, using Cauchy-Schwarz inequality, we have
\begin{align*}
A_{3,v}&\leq 2\sqrt{\sum_{k \in U_v}\frac{\left(\bx_{k,\alpha}^{\top}\pmb{\beta}_{\alpha}\right)^2}{N_v}}\sqrt{\sum_{k \in U_v}\frac{\left(\bx_{k,\alpha}^{\top}\left(\widehat{\pmb{\beta}}_{\alpha,v}-\pmb{\beta}_{\alpha}\right)\right)^2}{N_v}}\\
&\leq 2C_0^2\Arrowvert\widehat{\pmb{\beta}}_{\alpha,v}-\pmb{\beta}_{\alpha} \Arrowvert_2\Arrowvert\pmb{\beta}_{\alpha}\Arrowvert_2=\mathcal{O}_\P\left(\frac{1}{\sqrt{n_v}}\right).
\end{align*}
Finally, we conclude $|\widehat{S}_v-\widehat{S}_r|=\mathcal{O}_\P(n_v^{-1/2})$.
\vspace{4mm}
~\\
\textbf{Treatment of term (b).} Term (b) can be decomposed as
\begin{align*}
\left|\widehat{S}_r^2-\widetilde{S}_v^2\right|&\leq\left|\left(\frac{n_{r,v}}{N_v}\right)^{-1}\left(\sum_{k \in S_v}\frac{r_k\epsilon_k^2}{N_v}-\sum_{k \in U_v}\frac{p(\bx_k)\pi_k\epsilon_k^2}{N_v}\right)\right|\\
&\quad+\left|\left(\left(\frac{n_{r,v}}{N_v}\right)^{-1}-\left(\sum_{k \in U_v}\frac{p(\bx_k)\pi_k}{N_v} \right)^{-1} \right)\sum_{k \in U_v}\frac{p(\bx_k)\pi_k\epsilon_k^2}{N_v}\right|\\
&:=B_{1,v}+B_{2,v}.
\end{align*}
For $B_{1,v}$, we apply Lemma \ref{HTEstimator} with $\boldsymbol{T}_{k}:=\epsilon_k^2$ and $w_k:=r_k$. Since
\begin{equation*}
     \limsup_{v \rightarrow \infty}\sum_{k \in U_v}\frac{\E_m\left[\epsilon_k^4\right]}{N_v}\leq M_0
\end{equation*}
almost surely, it follows that
\begin{equation*}
\left|\sum_{k \in S_v}\frac{r_k\epsilon_k^2}{N_v}-\sum_{k \in U_v}\frac{p(\bx_k)\pi_k\epsilon_k^2}{N_v}\right|=\mathcal{O}_\P\left(\frac{1}{\sqrt{n_v}}\right).
\end{equation*}
Recall from Lemma \ref{SetofRespondents} that $\left(n_{r,v}/N_v\right)^{-1}=\mathcal{O}_\P(1)$,  so that $B_{1,v}=\mathcal{O}_\P(n_v^{-1/2})$. 
For $B_{2,v}$, we have
\begin{equation*}
\left|\left(\frac{n_{r,v}}{N_v}\right)^{-1}-\left(\sum_{k \in U_v}\frac{p(\bx_k)\pi_k}{N_v} \right)^{-1}\right|=\left|\left(\frac{n_{r,v}}{N_v}\right)^{-1}\right|\left|\left(\sum_{k \in U_v}\frac{p(\bx_k)\pi_k}{N_v} \right)^{-1}\right|\left|\frac{n_{r,v}}{N_v}-\sum_{k \in U_v}\frac{p(\bx_k)\pi_k}{N_v} \right|.
\end{equation*}
Applying Lemma \ref{HTEstimator}, set $\boldsymbol{T}_{k}:=1$ and $w_k:=r_k$ for $k \in U_v$, we have
\begin{equation*}
\left|\frac{n_{r,v}}{N_v}-\sum_{k \in U_v}\frac{p(\bx_k)\pi_k}{N_v} \right|=\mathcal{O}_\P\left(\frac{1}{\sqrt{n_v}}\right).
\end{equation*}
By \ref{D2} and positivity,
\begin{equation*}
\left(\sum_{k \in U_v}\frac{p(\bx_k)\pi_k}{N_v} \right)^{-1}\leq (\rho \lambda)^{-1}
\end{equation*}
so that $B_{2,v}=\mathcal{O}_\P(n_v^{-1/2})$ and $ |\widehat{S}_r^2-\widetilde{S}_v^2|=\mathcal{O}_\P(n_v^{-1/2})$.
\vspace{4mm}
~\\
\textbf{Treatment of term (c).} By Chebyshev's inequality
\begin{align}\label{epsfinal}
\mathbb{P}_{ mpq}\left(\left|\widetilde{S}_v^2-\sigma^2\right|>\epsilon \right)&\leq \frac{1}{\epsilon^2}\mathbb{E}_{ mpq}\left[\left(\widetilde{S}_v^2-\sigma^2\right)^2\right]
\nonumber\\&\leq \frac{1}{N_v}\sum_{k \in U_v}\frac{p(\bx_k)^2\pi_k^2\mathbb{E}_m[\epsilon_k^4]}{N_v}\left(\sum_{k \in U_v}\frac{p(\bx_k)\pi_k}{N_v} \right)^{-2}\nonumber\\
&\leq \frac{1}{N_v}\sum_{k \in U_v}\frac{M_0}{N_v}\left(\sum_{k \in U_v}\frac{p(\bx_k)\pi_k}{N_v} \right)^{-2}\nonumber\\
&\leq \frac{M_0}{N_v\rho^2\lambda^2}.
\end{align}
Combining the terms (a), (b), and (c), we conclude
$ |\widehat{S}_v^2-\sigma^2|=\mathcal{O}_\P(n_v^{-1/2})$ and thus $|\widehat{\sigma}_{\alpha,v}^2-\sigma^2|=\mathcal{O}_\P(n_v^{-1/2})$.
\end{proof}
\begin{lemma}\label{Cest}
Let $\alpha \in \cC$ and $(\widehat{\textbf{c}}_{\alpha,v})_{v\in \mathbb{N}}$ be a sequence of estimator given by \eqref{cwidehat}.  Assume \ref{S1}-\ref{S3} and \ref{D1}-\ref{D3}.  Then,
\begin{equation*}
\Arrowvert\widehat{\textbf{c}}_{\alpha,v}-\textbf{c}_{\alpha,v}\Arrowvert_2=\mathcal{O}_\P\left(\frac{1}{\sqrt{n_v}}\right).
\end{equation*}
\end{lemma}
\begin{proof}
Let \begin{equation}\label{tc}
\widetilde{\textbf{c}}_{\alpha,v}=\left(\sum_{k \in U_v}\frac{p(\bx_k)\pi_k\bx_{k,\alpha}\bx_{k,\alpha}^{\top}}{N_v}\right)^{-1}\sum_{k \in U_v}\frac{(1-p(\bx_k))\bx_{k,\alpha}}{N_v}.
\end{equation} 
Write
\begin{equation*}
\Arrowvert\widehat{\textbf{c}}_{\alpha,v}-\textbf{c}_{\alpha,v}\Arrowvert_2=\underbrace{\widehat{\Arrowvert\textbf{c}}_{\alpha,v}-\widetilde{\textbf{c}}_{\alpha,v}\Arrowvert_2}_{(a)}+\underbrace{\Arrowvert\widetilde{\textbf{c}}_{\alpha,v}-\textbf{c}_{\alpha,v}\Arrowvert_2}_{(b)}.
\end{equation*}
~\\
\textbf{Treatment of (a).}
We decompose 
\begin{align*}
&\Arrowvert\widehat{\textbf{c}}_{\alpha,v}-\widetilde{\textbf{c}}_{\alpha,v}\Arrowvert_2\\
\quad&=\Bigg\Arrowvert \left(\sum_{k \in S_{v}}\frac{r_k\bx_{k,\alpha}\bx_{k,\alpha}^{\top}}{N_v}\right)^{-1}\sum_{k \in S_v}\frac{(1-r_k)\bx_{k,\alpha}}{\pi_kN_v}\\
&\quad\quad-\left(\sum_{k \in U_v}\frac{p(\bx_k)\pi_k\bx_{k,\alpha}\bx_{k,\alpha}^{\top}}{N_v}\right)^{-1}\sum_{k \in U_v}\frac{(1-p(\bx_k))\bx_{k,\alpha}}{N_v}\Bigg\Arrowvert_2\\
&\quad= \Bigg\Arrowvert \left(\sum_{k \in S_v}\frac{r_k\bx_{k,\alpha}\bx_{k,\alpha}^{\top}}{N_v}\right)^{-1}\left(\sum_{k \in S_v}\frac{(1-r_k)\bx_{k,\alpha}}{\pi_kN_v}-\sum_{k \in U_v}\frac{(1-p(\bx_k))\bx_{k,\alpha}}{N_v}\right)\\
&\quad\quad+\left(\left(\sum_{k \in S_v}\frac{r_k\bx_{k,\alpha}\bx_{k,\alpha}^{\top}}{N_v}\right)^{-1}-\left(\sum_{k \in U_v}\frac{p(\bx_k)\pi_k\bx_{k,\alpha}\bx_{k,\alpha}^{\top}}{N_v} \right)^{-1} \right)\sum_{k \in U_v}\frac{(1-p(\bx_k))\bx_{k,\alpha}}{N_v}\Bigg\Arrowvert_2\\
&\quad\leq \Bigg\Arrowvert \left(\sum_{k \in S_v}\frac{r_k\bx_{k,\alpha}\bx_{k,\alpha}^{\top}}{N_v}\right)^{-1}\left(\sum_{k \in S_v}\frac{(1-r_k)\bx_{k,\alpha}}{\pi_kN_v}-\sum_{k \in U_v}\frac{(1-p(\bx_k))\bx_{k,\alpha}}{N_v}\right)\Bigg\Arrowvert_2\\
&\quad\quad+\Bigg \Arrowvert\left(\left(\sum_{k \in S_v}\frac{r_k\bx_{k,\alpha}\bx_{k,\alpha}^{\top}}{N_v}\right)^{-1}-\left(\sum_{k \in U_v}\frac{p(\bx_k)\pi_k\bx_{k,\alpha}\bx_{k,\alpha}^{\top}}{N_v} \right)^{-1} \right)\sum_{k \in U_v}\frac{(1-p(\bx_k))\bx_{k,\alpha}}{N_v}\Bigg\Arrowvert_2\\
&:=A_{1,v}+A_{2,v}.
\end{align*}
For $A_{1,v}$, we have
\begin{equation*}
A_{1,v}\leq \Bigg\Arrowvert \left(\sum_{k \in S_v}\frac{r_k\bx_{k,\alpha}\bx_{k,\alpha}^{\top}}{N_v}\right)^{-1}\Bigg\Arrowvert_{op}\Bigg\Arrowvert\sum_{k \in S_v}\frac{(1-r_k)\bx_{k,\alpha}}{\pi_kN_v}-\sum_{k \in U_v}\frac{(1-p(\bx_k))\bx_{k,\alpha}}{N_v}\Bigg\Arrowvert_2.
\end{equation*}
By Lemma \ref{BoundnessApi-1},
\begin{equation*}
    \Bigg\Arrowvert \left(\sum_{k \in S_v}\frac{r_k\bx_{k,\alpha}\bx_{k,\alpha}^{\top}}{N_v}\right)^{-1}\Bigg\Arrowvert_{op}=\mathcal{O}_\P(1),
\end{equation*}
Moreover, applying Lemma \ref{HTEstimator}, by setting $\boldsymbol{T}_k:=\pi_k^{-1}\textbf{x}_{k,\alpha}$, $w_k:=1-r_k$, since
\begin{equation*}
    \limsup_{v\rightarrow\infty}\sum_{k \in U_v}\frac{\E_m\left[\Arrowvert\pi_k^{-1}\textbf{x}_{k,\alpha}\Arrowvert_F^2\right]}{N_v}\leq \frac{C_0^2}{\lambda^2}
\end{equation*}
almost surely, it follows that
\begin{equation*}
    \Bigg\Arrowvert\sum_{k \in U_v}\frac{(1-r_k)\textbf{x}_{k,\alpha}}{\pi_kN_v}-\sum_{k \in U_v}\frac{(1-p_k)\textbf{x}_{k,\alpha}}{N_v}\Bigg\Arrowvert_{2}^2=\mathcal{O}_\P\left(\frac{1}{\sqrt{n_v}}\right).
\end{equation*}
Hence, $A_{1,v}=o_\P(1).$ For $A_{2,v}$, we obtain
\begin{align*}
A_{2,v}&\leq \Bigg\Arrowvert\left(\sum_{k \in S_v}\frac{r_k\bx_{k,\alpha}\bx_{k,\alpha}^{\top}}{N_v}\right)^{-1}-\left(\sum_{k \in U_v}\frac{p(\bx_k)\pi_k\bx_{k,\alpha}\bx_{k,\alpha}^{\top}}{N_v} \right)^{-1} \Bigg\Arrowvert_{op}\\
&\quad\times\Bigg\Arrowvert\sum_{k \in U_v}\frac{(1-p(\bx_k))\bx_{k,\alpha}}{N_v}\Bigg\Arrowvert_2.
\end{align*}
Equation \eqref{HTMatrix} gives
\begin{equation*}
\Bigg\Arrowvert\sum_{k \in S_v}\frac{r_k\bx_{k,\alpha}\bx_{k,\alpha}^{\top}}{N_v}-\sum_{k \in U_v}\frac{p(\bx_k)\pi_k\bx_{k,\alpha}\bx_{k,\alpha}^{\top}}{N_v}  \Bigg\Arrowvert_{op}=\mathcal{O}_\P\left(\frac{1}{\sqrt{n_v}}\right).
\end{equation*}
 Lemma \ref{BoundnessAp-1} and Lemma \ref{BoundnessA-1} state
\begin{equation*}
     \Bigg\Arrowvert \left(\sum_{k \in U_v}\frac{\bx_{k,\alpha}\bx_{k,\alpha}^{\top}p(\bx_k)\pi_k}{N_v}\right)^{-1}\Bigg\Arrowvert_{op}=\mathcal{O}_\P(1), \qquad\Bigg\Arrowvert \left(\sum_{k \in S_v}\frac{r_k\bx_{k,\alpha}\bx_{k,\alpha}^{\top}}{N_v}\right)^{-1}\Bigg\Arrowvert_{op}=\mathcal{O}_\P(1)
\end{equation*}
respectively. Using Lemma \ref{MatrixInverse}, we obtain
\begin{equation*}
    \Bigg\Arrowvert\left(\sum_{k \in S_v}\frac{r_k\bx_{k,\alpha}\bx_{k,\alpha}^{\top}}{N_v}\right)^{-1}-\left(\sum_{k \in U_v}\frac{p(\bx_k)\pi_k\bx_{k,\alpha}\bx_{k,\alpha}^{\top}}{N_v} \right)^{-1} \Bigg\Arrowvert_{op}=\mathcal{O}_\P\left(\frac{1}{\sqrt{n_v}}\right).
\end{equation*}
Finally, note that \begin{equation*}
    \Bigg\Arrowvert\sum_{k \in U_v}\frac{(1-p(\bx_k))\bx_{k,\alpha}}{N_v}\Bigg\Arrowvert_2\leq \sum_{k \in U_v}\frac{\Arrowvert\bx_{k,\alpha}\Arrowvert_2}{N_v}\leq C_0
\end{equation*}
almost surely, we conclude that $A_{2,v}=\mathcal{O}_\P(n_v^{-1/2}).$
\vspace{4mm}
~\\
\textbf{Treatment of (b).} The term (b) can be bounded by
\begin{align*}
&\Arrowvert\widetilde{\textbf{c}}_{\alpha,v}-\textbf{c}_{\alpha,v}\Arrowvert_2\\
\quad&\leq \Bigg\Arrowvert \left(\sum_{k \in U_v}\frac{r_k\pi_k\bx_{k,\alpha}\bx_{k,\alpha}^{\top}}{N_v}\right)^{-1}\Bigg\Arrowvert_{op}\Bigg\Arrowvert\sum_{k \in U_v}\frac{(1-r_k)\bx_{k,\alpha}}{N_v}-\sum_{k \in U_v}\frac{(1-p(\bx_k))\bx_{k,\alpha}}{N_v}\Bigg\Arrowvert_2\nonumber\\
&\quad\quad+\Bigg \Arrowvert\left(\sum_{k \in U_v}\frac{r_k\pi_k\bx_{k,\alpha}\bx_{k,\alpha}^{\top}}{N_v}\right)^{-1}-\left(\sum_{k \in U_v}\frac{p(\bx_k)\pi_k\bx_{k,\alpha}\bx_{k,\alpha}^{\top}}{N_v} \right)^{-1} \Bigg\Arrowvert_{op}\\
&\quad\quad\quad\times\Bigg\Arrowvert\sum_{k \in U_v}\frac{(1-p(\bx_k))\bx_{k,\alpha}}{N_v}\Bigg\Arrowvert_2\nonumber\\
&\quad:=B_{1,v}+B_{2,v}.
\end{align*}
For $B_{1,v}$, applying  Lemma \ref{REstimator}, by setting $w_k:=1-r_k$, $\boldsymbol{T}_k:=\textbf{x}_{k,\alpha}$ and recalling that
\begin{equation*}
    \limsup_{v \rightarrow\infty }\sum_{k \in U_v}\frac{\E_m\left[\Arrowvert\textbf{x}_{k,\alpha}\Arrowvert_2^2\right]}{N_v}\leq C_0^2<\infty,
\end{equation*}
we have
\begin{equation*}
    \Bigg\Arrowvert\sum_{k \in U_v}\frac{(1-r_k)\bx_{k,\alpha}}{N_v}-\sum_{k \in U_v}\frac{(1-p(\bx_k))\bx_{k,\alpha}}{N_v}\Bigg\Arrowvert_2=\mathcal{O}_\P\left(\frac{1}{\sqrt{n_v}}\right).
\end{equation*}
Also, using Lemma \ref{BoundnessApi-1}, we obtain $B_{1,v}=\mathcal{O}_\P(n_v^{-1/2}).$ For $B_{2,v}$, applying Lemma \ref{REstimator}, setting $w_k:=r_k$, $\boldsymbol{T}_k:=\pi_k\textbf{x}_{k,\alpha}\textbf{x}_{k,\alpha}^{\top}$, \eqref{FiniteXX} implies
 \begin{align}\label{RMatrix}
       & \Bigg\Arrowvert\sum_{k \in U_v}\frac{r_k\pi_k\bx_{k,\alpha}\bx_{k,\alpha}^{\top}}{N_v}-\sum_{k \in U_v}\frac{p(\bx_k)\pi_k\bx_{k,\alpha}\bx_{k,\alpha}^{\top}}{N_v}\Bigg\Arrowvert_{op}=\mathcal{O}_\P\left(\frac{1}{\sqrt{n_v}}\right).
\end{align}
An application of Lemma \ref{MatrixInverse} gives
\begin{equation*}
    \Bigg \Arrowvert\left(\sum_{k \in U_v}\frac{r_k\pi_k\bx_{k,\alpha}\bx_{k,\alpha}^{\top}}{N_v}\right)^{-1}-\left(\sum_{k \in U_v}\frac{p(\bx_k)\pi_k\bx_{k,\alpha}\bx_{k,\alpha}^{\top}}{N_v} \right)^{-1} \Bigg\Arrowvert_{op}=\mathcal{O}_\P\left(\frac{1}{\sqrt{n_v}}\right).
\end{equation*}
As a result, $B_{2,v}=\mathcal{O}_\P(n_v^{-1/2})$.
Putting all things together, we obtain $A_v=\mathcal{O}_\P(n_v^{-1/2})$ and $B_v=\mathcal{O}_\P(n_v^{-1/2})$. This concludes $\Arrowvert\widehat{\textbf{c}}_{\alpha,v}-\textbf{c}_{\alpha,v}\Arrowvert_2=\mathcal{O}_\P(n_v^{-1/2})$.
\end{proof}
\subsubsection{Consistency of plug-in estimators}
\begin{lemma}\label{Cfunction}
Let $\alpha \in \cC$ and $(\widehat{\textbf{c}}_{\alpha,v})_{v\in\mathbb{N}}$ be a sequence defined in \eqref{cwidehat}. Assume \ref{S1}-\ref{S3} and \ref{D1}-\ref{D3}. We have
\begin{equation*}
\Bigg|\sum_{k \in S_v}\frac{1-r_k+r_k\left(\pi_k\widehat{\textbf{c}}_{\alpha,v}^{\top}\bx_{k,\alpha}\right)^2}{N_v\pi_k}-\sum_{k \in U_v}\frac{1-r_k+r_k\left(\pi_k\textbf{c}_{\alpha,v}^{\top}\bx_{k,\alpha}\right)^2}{N_v}\Bigg|=\mathcal{O}_\P\left(\frac{1}{\sqrt{n_v}}\right).
\end{equation*}
\end{lemma}
\begin{proof}
Consider the following decomposition,
\begin{align*}
\Bigg|&\sum_{k \in S_v}\frac{1-r_k+r_k\left(\pi_k\widehat{\textbf{c}}_{\alpha,v}^{\top}\bx_{k,\alpha}\right)^2}{N_v\pi_k}-\sum_{k \in U_v}\frac{1-r_k+r_k\left(\pi_k\widetilde{\textbf{c}}_{\alpha,v}^{\top}\bx_{k,\alpha}\right)^2}{N_v}\Bigg|\\&\quad\leq\Bigg|\sum_{k \in S_v}\frac{1-r_k+r_k\left(\pi_k\widehat{\textbf{c}}_{\alpha,v}^{\top}\bx_{k,\alpha}\right)^2}{N_v\pi_k}-\sum_{k \in S_v}\frac{1-r_k+r_k\left(\pi_k\widetilde{\textbf{c}}_{\alpha,v}^{\top}\bx_{k,\alpha}\right)^2}{N_v\pi_k}\Bigg|\\&\quad+\Bigg|\sum_{k \in S_v}\frac{1-r_k+r_k\left(\pi_k\widetilde{\textbf{c}}_{\alpha,v}^{\top}\bx_{k,\alpha}\right)^2}{N_v\pi_k}-\sum_{k \in U_v}\frac{1-r_k+r_k\left(\pi_k\widetilde{\textbf{c}}_{\alpha,v}^{\top}\bx_{k,\alpha}\right)^2}{N_v}\Bigg|\\
&\quad+\Bigg|\sum_{k \in U_v}\frac{1-r_k+r_k\left(\pi_k\widetilde{\textbf{c}}_{\alpha,v}^{\top}\bx_{k,\alpha}\right)^2}{N_v}-\sum_{k \in U_v}\frac{1-r_k+r_k\left(\pi_k\textbf{c}_{\alpha,v}^{\top}\bx_{k,\alpha}\right)^2}{N_v}\Bigg|\\
&:=A_v+B_v+C_v,
\end{align*}
where $\widetilde{\textbf{c}}_{\alpha,v}$ is given by \eqref{tc}. Next, we will show $A_v$, $B_v$ and $C_v$ converge to 0 with rate $\mathcal{O}_\P(n_v^{-1/2})$.
~\\
\textbf{Treatment of $A_v$.}
Using
\begin{equation*}
\left(\widehat{\textbf{c}}_{\alpha,v}^{\top}\bx_{k,\alpha}\right)^2-\left(\widetilde{\textbf{c}}_{\alpha,v}^{\top}\bx_{k,\alpha}\right)^2=\left(\left(\widehat{\textbf{c}}_{\alpha,v}-\widetilde{\textbf{c}}_{\alpha,v}\right)\bx_{k,\alpha}\right)^2+2\left(\widehat{\textbf{c}}_{\alpha,v}-\widetilde{\textbf{c}}_{\alpha,v}\right)\bx_{k,\alpha}\widetilde{\textbf{c}}_{\alpha,v}^{\top}\bx_{k,\alpha},
\end{equation*}
we obtain
\begin{align}\label{decomA}
A_v&=\left|\sum_{k \in S_v}r_k\frac{\left(\pi_k\widehat{\textbf{c}}_{\alpha,v}^{\top}\bx_{k,\alpha}\right)^2-\left(\pi_k\widetilde{\textbf{c}}_{\alpha,v}^{\top}\bx_{k,\alpha}\right)^2}{N_v\pi_k}\right|\nonumber\\
&\leq \sum_{k \in U_v}\frac{\left|\left(\widehat{\textbf{c}}_{\alpha,v}^{\top}\bx_{k,\alpha}\right)^2-\left(\widetilde{\textbf{c}}_{\alpha,v}^{\top}\bx_{k,\alpha}\right)^2\right|}{N_v\lambda}\nonumber\\&\leq \sum_{k\in U_v}\frac{\left(\left(\widehat{\textbf{c}}_{\alpha,v}^{\top}-\widetilde{\textbf{c}}_{\alpha,v}\right)\bx_{k,\alpha}\right)^2+2\left|\widetilde{\textbf{c}}_{\alpha,v}^{\top}\bx_{k,\alpha}\left(\widehat{\textbf{c}}_{\alpha,v}^{\top}-\widetilde{\textbf{c}}_{\alpha,v}\right)\bx_{k,\alpha}\right|}{N_v\lambda}\nonumber\\&:=A_{1,v}+A_{2,v}.
\end{align}
For $A_{1,v}$, having proved $\Arrowvert\widehat{\textbf{c}}_{\alpha,v}-\widetilde{\textbf{c}}_{\alpha,v}\Arrowvert_2=\mathcal{O}_\P(n_v^{-1/2})$ in Lemma \ref{Cest}, we have
\begin{align*}
A_{1,v}&\leq \frac{1}{\lambda}\sum_{k \in U_v}\frac{\left(\left(\widehat{\textbf{c}}_{\alpha,v}^{\top}-\widetilde{\textbf{c}}_{\alpha,v}\right)\bx_{k,\alpha}\right)^2}{N_v}\\
&\leq \frac{1}{\lambda}\Arrowvert \widehat{\textbf{c}}_{\alpha,v}^{\top}-\widetilde{\textbf{c}}_{\alpha,v}\Arrowvert_2^2\sum_{k \in U_v}\frac{\Arrowvert\bx_{k,\alpha}\Arrowvert_2^2}{N_v}\\
&\leq \frac{C_0^2}{\lambda}\Arrowvert \widehat{\textbf{c}}_{\alpha,v}^{\top}-\widetilde{\textbf{c}}_{\alpha,v}\Arrowvert_2^2=\mathcal{O}_\P\left(\frac{1}{\sqrt{n_v}}\right).
\end{align*}
For $A_{2,v}$, recall from Lemma \ref{finitetildec}, we have
\begin{equation*}
     \limsup_{v \rightarrow \infty}\Arrowvert\widetilde{\textbf{c}}_{\alpha,v}\Arrowvert_2<\infty
\end{equation*}
almost surely.
It follows that
\begin{align*}
A_{2,v}&\leq\frac{2}{\lambda}\sqrt{\sum_{k \in U_v}\frac{\left(\widetilde{\textbf{c}}_{\alpha,v}^{\top}\bx_{k,\alpha}\right)^2}{N_v}}\sqrt{\sum_{k \in U_v}\frac{\left(\left(\widehat{\textbf{c}}_{\alpha,v}^{\top}-\widetilde{\textbf{c}}_{\alpha,v}\right)\bx_{k,\alpha}\right)^2}{N_v}}\\
&\leq \frac{2}{\lambda}\Arrowvert \widetilde{\textbf{c}}_{\alpha,v}\Arrowvert_2\Arrowvert \widehat{\textbf{c}}_{\alpha,v}^{\top}-\widetilde{\textbf{c}}_{\alpha,v}\Arrowvert_2\sum_{k \in U_v}\frac{\Arrowvert\bx_{k,\alpha}\Arrowvert_2^2}{N_v}\\
&\leq\frac{2C_0^2}{\lambda}\Arrowvert \widetilde{\textbf{c}}_{\alpha,v}\Arrowvert_2\Arrowvert \widehat{\textbf{c}}_{\alpha,v}^{\top}-\widetilde{\textbf{c}}_{\alpha,v}\Arrowvert_2=\mathcal{O}_\P\left(\frac{1}{\sqrt{n_v}}\right).
\end{align*}
As a result, we have $A_v=\mathcal{O}_\P(n_v^{-1/2}).$
\vspace{4mm}
~\\
\textbf{Treatment of $B_v$.}  Let $a_k=1-r_k+r_k\left(\pi_k\widetilde{\textbf{c}}_{\alpha,v}\bx_{k,\alpha}\right)^2$ for $k \in U_v$. Applying Lemma \ref{HTEstimator}, with $\boldsymbol{T}_k:=a_k$ and $w_k:=1$ for $k \in U_v$, it remains to show
\begin{equation*}
    \limsup_{v\rightarrow \infty}\sum_{k \in U_v}\frac{\E_m\left[a_k^2\right]}{N_v}<\infty
\end{equation*}
almost surely. To this aim, write
\begin{align*}
\limsup_{v\rightarrow \infty}\sum_{k \in U_v}\frac{\E_m\left[a_k^2\right]}{N_v}&\leq \limsup_{v\rightarrow \infty}\frac{1}{N_v}\sum_{k \in U_v}\left(2(1-r_k)^2+2r_k^2\left(\pi_k\widetilde{\textbf{c}}_{\alpha,v}^{\top}\bx_{k,\alpha}\right)^2\right)\\&\leq\limsup_{v\rightarrow \infty} \frac{1}{N_v}\sum_{k \in U_v}\left(2+2\Arrowvert\widetilde{\textbf{c}}_{\alpha,v}\Arrowvert_2^2\Arrowvert\bx_{k,\alpha}\Arrowvert_2^2\right)\\
&\leq \limsup_{v\rightarrow \infty}2\left(1+C_0^2\Arrowvert\widetilde{\textbf{c}}_{\alpha,v}\Arrowvert_2^2 \right)<\infty
\end{align*}
almost surely. This concludes $B_v=\mathcal{O}_\P(n_v^{-1/2}).$ \vspace{4mm}
~\\
\textbf{Treatment of $C_v$.}  Note that
\begin{equation*}
\left({\textbf{c}}_{\alpha,v}^{\top}\bx_{k,\alpha}\right)^2-\left(\widetilde{\textbf{c}}_{\alpha,v}^{\top}\bx_{k,\alpha}\right)^2=\left(\left({\textbf{c}}_{\alpha,v}-\widetilde{\textbf{c}}_{\alpha,v}\right)\bx_{k,\alpha}\right)^2+2\left({\textbf{c}}_{\alpha,v}-\widetilde{\textbf{c}}_{\alpha,v}\right)\bx_{k,\alpha}\widetilde{\textbf{c}}_{\alpha,v}^{\top}\bx_{k,\alpha}.
\end{equation*}
Using the result $\Arrowvert\widetilde{\textbf{c}}_{\alpha,v}-\textbf{c}_{\alpha,v}\Arrowvert_2=\mathcal{O}_\P(n_v^{-1/2})$ in Lemma \ref{Cest} and following the same decomposition in \eqref{decomA}, we obtain
\begin{align}\label{Cfuntilde}
&\left|\dfrac{1}{N_v}\sum_{k \in U_v}1-r_k+r_k\left(\pi_k\widetilde{\textbf{c}}_{\alpha,v}^{\top}\bx_{k,\alpha}\right)^2-\dfrac{1}{N_v}\sum_{k \in U_v}1-r_k+r_k\left(\pi_k\textbf{c}_{\alpha,v}^{\top}\bx_{k,\alpha}\right)^2\right|\nonumber\\
&\quad\leq \dfrac{1}{N_v}\sum_{k \in U_v}\left|\left(\widetilde{\textbf{c}}_{\alpha,v}^{\top}\bx_{k,\alpha}\right)^2-\left({\textbf{c}}_{\alpha,v}^{\top}\bx_{k,\alpha}\right)^2\right|\nonumber\\
&\quad\leq \dfrac{1}{N_v}\sum_{k\in U_v}\left(\left(\widetilde{\textbf{c}}_{\alpha,v}^{\top}-{\textbf{c}}_{\alpha,v}\right)\bx_{k,\alpha}\right)^2+2\left|\widetilde{\textbf{c}}_{\alpha,v}^{\top}\bx_{k,\alpha}\left(\widetilde{\textbf{c}}_{\alpha,v}^{\top}-{\textbf{c}}_{\alpha,v}\right)\bx_{k,\alpha}\right|\nonumber\\
&\quad\leq\Arrowvert\widetilde{\textbf{c}}_{\alpha,v}-\textbf{c}_{\alpha,v}\Arrowvert_2\dfrac{1}{N_v}\sum_{k \in U_v}\Arrowvert\bx_{k,\alpha}\Arrowvert_2^2\left(2\Arrowvert\widetilde{\textbf{c}}_{\alpha,v}\Arrowvert_2+\Arrowvert\widetilde{\textbf{c}}_{\alpha,v}-\textbf{c}_{\alpha,v}\Arrowvert_2\right)\nonumber\\
&\quad \leq C_0^2\Arrowvert\widetilde{\textbf{c}}_{\alpha,v}-\textbf{c}_{\alpha,v}\Arrowvert_2\left(2\Arrowvert\widetilde{\textbf{c}}_{\alpha,v}\Arrowvert_2+\Arrowvert\widetilde{\textbf{c}}_{\alpha,v}-\textbf{c}_{\alpha,v}\Arrowvert_2\right)=\mathcal{O}_\P\left(\frac{1}{\sqrt{n_v}}\right).
\end{align}
Overall
\begin{equation*}
\Bigg|\sum_{k \in S_v}\frac{1-r_k+r_k\left(\pi_k\widehat{\textbf{c}}_{\alpha,v}^{\top}\bx_{k,\alpha}\right)^2}{N_v\pi_k}-\sum_{k \in U_v}\frac{1-r_k+r_k\left(\pi_k\textbf{c}_{\alpha,v}^{\top}\bx_{k,\alpha}\right)^2}{N_v}\Bigg|=\mathcal{O}_\P\left(\frac{1}{\sqrt{n_v}}\right).
\end{equation*}
\end{proof}
\begin{lemma}\label{ConstEtaEst}
Let $\alpha \in \cA$ and assume \ref{D1}-\ref{D3} and \ref{S1}-\ref{S3}. We have 
\begin{equation*}
\dfrac{1}{N_v}\sum_{k \in U_v}\left(\widehat{\eta}_{k,\alpha}-\eta_{k,\alpha}\right)^2=o_\P(1),
\end{equation*}
where $\widehat{\eta}_{k,\alpha}$ is given by \eqref{etawidehat}.
\end{lemma}
\begin{proof}
For an arbitrary $k \in U_v$, we decompose 
\begin{align*}
\eta_{k,\alpha}-\widehat{\eta}_{k,\alpha}&=\bx_{k,\alpha}^{\top}\pmb{\beta}_{\alpha}+r_k\left(1+\pi_k\textbf{c}_{\alpha,v}^{\top}\bx_{k,\alpha}\right)\left(y_k-\bx_{k,\alpha}^{\top}\pmb{\beta}_{\alpha,v}\right)\\&\quad-\left(\bx_{k,\alpha}^{\top}\widehat{\pmb{\beta}}_{\alpha,v}+r_k\left(1+\pi_k\widehat{\textbf{c}}_{\alpha,v}^{\top}\bx_{k,\alpha}\right)\left(y_k-\bx_{k,\alpha}^{\top}\widehat{\pmb{\beta}}_{\alpha,v}\right)\right)\\
&\quad+r_k\pi_k\widehat{\textbf{c}}_{\alpha,v}^{\top}\bx_{k,\alpha}\bx_{k,\alpha}^{\top}\widehat{\pmb{\beta}}_{\alpha}-r_k\pi_k\textbf{c}_{\alpha,v}^{\top}\bx_{k,\alpha}\bx_{k,\alpha}^{\top}\pmb{\beta}\\
&=(1-r_k)\bx_{k,\alpha}^{\top}\left(\pmb{\beta}_{\alpha}-\widehat{\pmb{\beta}}_{\alpha,v}\right)-r_k\pi_k\left(\textbf{c}_{\alpha,v}-\widehat{\textbf{c}}_{\alpha,v}\right)^{\top}\bx_{k,\alpha}y_k\\
&\quad+r_k\pi_k(\widehat{\textbf{c}}_{\alpha,v}-\textbf{c}_{\alpha,v})^{\top}\bx_{k,\alpha}\bx_{k,\alpha}^{\top}\left(\widehat{\pmb{\beta}}_{\alpha,v}-\pmb{\beta}_{\alpha}\right)\\
&\quad+r_k\pi_k\textbf{c}_{\alpha,v}^{\top}\bx_{k,\alpha}\bx_{k,\alpha}^{\top}\left(\widehat{\pmb{\beta}}_{\alpha,v}-\pmb{\beta}_{\alpha}\right)+r_k\pi_k(\widehat{\textbf{c}}_{\alpha,v}-\textbf{c}_{\alpha,v})^{\top}\bx_{k,\alpha}\bx_{k,\alpha}^{\top}\pmb{\beta}_{\alpha},
\end{align*}
since
\begin{align*}
&r_k\pi_k\widehat{\textbf{c}}_{\alpha,v}^{\top}\bx_{k,\alpha}\bx_{k,\alpha}^{\top}\widehat{\pmb{\beta}}_{\alpha}-r_k\pi_k\textbf{c}_{\alpha,v}^{\top}\bx_{k,\alpha}\bx_{k,\alpha}^{\top}\pmb{\beta}\\&\quad=r_k\pi_k\left(\widehat{\textbf{c}}_{\alpha,v}-\textbf{c}_{\alpha,v}\right)^{\top}\bx_{k,\alpha}\bx_{k,\alpha}^{\top}\left(\widehat{\pmb{\beta}}_{\alpha,v}-\pmb{\beta}_{\alpha}\right)\\&\quad\quad+r_k\pi_k\textbf{c}_{\alpha,v}^{\top}\bx_{k,\alpha}\bx_{k,\alpha}^{\top}\left(\widehat{\pmb{\beta}}_{\alpha,v}-\pmb{\beta}_{\alpha}\right)+r_k\pi_k\left(\widehat{\textbf{c}}_{\alpha,v}-\textbf{c}_{\alpha,v}\right)^{\top}\bx_{k,\alpha}\bx_{k,\alpha}^{\top}\pmb{\beta}_{\alpha}.
\end{align*}
On the other hand, according to Lemma \ref{BoundnessApi-1}, we obtain,
\begin{equation}\label{eq:cOP}
\Arrowvert\textbf{c}_{\alpha,v}\Arrowvert_2\leq \Bigg\Arrowvert\left(\sum_{k \in U_v}\frac{r_k\pi_k\bx_{k,\alpha}\bx_{k,\alpha}^{\top}}{N_v} \right)^{-1} \Bigg\Arrowvert_{op}\Bigg\Arrowvert \sum_{k\in U_v}\frac{(1-r_k)\bx_{k,\alpha}}{N_v} \Bigg\Arrowvert_2=\mathcal{O}_\P(1).
\end{equation}
Combining the results that $\Arrowvert\widehat{\pmb{\beta}}_{\alpha,v}-\pmb{\beta}_{\alpha,v}\Arrowvert=o_\P(1)$ in Lemma \ref{BetaEst} and $\Arrowvert\widehat{\textbf{c}}_{\alpha,v}-\textbf{c}_{\alpha,v}\Arrowvert_2=o_\P(1)$ in Lemma \ref{Cest}, an application of Cauchy-Schwarz inequality gives
\begin{align*}
&\sum_{k \in U_v}\frac{(\widehat{\eta}_{k,\alpha}-\eta_{k,\alpha})^2}{N_v}\\&\quad=\sum_{k \in U_v}\frac{(\eta_{k,\alpha}-\widehat{\eta}_{k,\alpha})^2}{N_v}\\&=\frac{1}{N_v}\sum_{k \in U_v}\Bigg((1-r_k)\bx_{k,\alpha}^{\top}\left(\pmb{\beta}_{\alpha}-\widehat{\pmb{\beta}}_{\alpha,v}\right)-r_k\pi_k\left(\textbf{c}_{\alpha,v}-\widehat{\textbf{c}}_{\alpha,v}\right)^{\top}\bx_{k,\alpha}y_k\\&\quad+r_k\pi_k\left(\widehat{\textbf{c}}_{\alpha,v}-\textbf{c}_{\alpha,v}\right)^{\top}\bx_{k,\alpha}\bx_{k,\alpha}^{\top}\left(\widehat{\pmb{\beta}}_{\alpha,v}-\pmb{\beta}_{\alpha}\right)\\
&\quad\quad+r_k\pi_k\textbf{c}_{\alpha,v}^{\top}\bx_{k,\alpha}\bx_{k,\alpha}^{\top}\left(\widehat{\pmb{\beta}}_{\alpha,v}-\pmb{\beta}_{\alpha}\right)+r_k\pi_k\left(\widehat{\textbf{c}}_{\alpha,v}-\textbf{c}_{\alpha,v}\right)^{\top}\bx_{k,\alpha}\bx_{k,\alpha}^{\top}\pmb{\beta}_{\alpha}\Bigg)^2\\
&\quad\leq \frac{5}{N_v}\sum_{k \in U_v}\Bigg(\left(\bx_{k,\alpha}^{\top}\left(\widehat{\pmb{\beta}}_{\alpha,v}-\pmb{\beta}_{\alpha}\right)\right)^2+\left(\left(\widehat{\textbf{c}}_{\alpha,v}-\textbf{c}_{\alpha,v}\right)^{\top}\bx_{k,\alpha}y_k\right)^2\\
&\quad\quad+\left(\left(\widehat{\textbf{c}}_{\alpha,v}-\textbf{c}_{\alpha,v}\right)^{\top}\bx_{k,\alpha}\bx_{k,\alpha}^{\top}\left(\widehat{\pmb{\beta}}_{\alpha,v}-\pmb{\beta}_{\alpha}\right)\right)^2+\left(\textbf{c}_{\alpha,v}^{\top}\bx_{k,\alpha}\bx_{k,\alpha}^{\top}\left(\widehat{\pmb{\beta}}_{\alpha,v}-\pmb{\beta}_{\alpha}\right)\right)^2\\
&\quad\quad+\left(\left(\widehat{\textbf{c}}_{\alpha,v}-\textbf{c}_{\alpha,v}\right)^{\top}\bx_{k,\alpha}\bx_{k,\alpha}^{\top}\pmb{\beta}_{\alpha}\right)^2\Bigg)\\
&\quad\leq 5\Arrowvert\widehat{\pmb{\beta}}_{\alpha,v}-\pmb{\beta}_{\alpha}\Arrowvert_2^2 \sum_{k \in U_v}\frac{\Arrowvert\bx_{k,\alpha}\Arrowvert_2^2}{N_v}+5\Arrowvert\widehat{\textbf{c}}_{\alpha,v}-\textbf{c}_{\alpha,v}\Arrowvert_2^2\sqrt{\sum_{k \in U_v}\frac{y_k^4}{N_v}}\sqrt{\sum_{k \in U_v}\frac{\Arrowvert\bx_{k,\alpha}\Arrowvert_2^4}{N_v}}\\
&\quad\quad+5\Arrowvert\widehat{\pmb{\beta}}_{\alpha,v}-\pmb{\beta}_{\alpha}\Arrowvert_2^2\Arrowvert\widehat{\textbf{c}}_{\alpha,v}-\textbf{c}_{\alpha,v}\Arrowvert_2^2\sum_{k \in U_v}\frac{\Arrowvert\bx_{k,\alpha}\Arrowvert_2^4}{N_v}\\
&\quad\quad+5\Arrowvert\widehat{\pmb{\beta}}_{\alpha,v}-\pmb{\beta}_{\alpha}\Arrowvert_2^2\Arrowvert\textbf{c}_{\alpha,v}\Arrowvert_2^2\sum_{k \in U_v}\frac{\Arrowvert\bx_{k,\alpha}\Arrowvert_2^4}{N_v}\\
&\quad\quad+5\Arrowvert\pmb{\beta}_{\alpha}\Arrowvert_2^2\Arrowvert\widehat{\textbf{c}}_{\alpha,v}-\textbf{c}_{\alpha,v}\Arrowvert_2^2\sum_{k \in U_v}\frac{\Arrowvert\bx_{k,\alpha}\Arrowvert_2^4}{N_v}=o_\P(1).
\end{align*}
Finally, we conclude
\begin{equation*}
\sum_{k \in U_v}\frac{(\widehat{\eta}_{k,\alpha}-\eta_{k,\alpha})^2}{N_v}=o_\P(1).
\end{equation*}
\end{proof}
\subsection{Additional lemmas}
\subsubsection{Technical lemmas}
\begin{lemma}\label{HTEstimator}
Let $(w_k)_{k \in U_v}$ be i.i.d. random variable satisfying $\E_q[w_k^2]\leq b$ for some $b \geq 0$ almost surely. Let $(\widehat{\boldsymbol{Z}}_{\pi,v})_{v \in \mathbb{N}}$ and $(\boldsymbol{Z}_v)_{v\in \mathbb{N}}$ be a sequence of matrices in $\R^{d_1 \times d_2}$ defined by
\begin{equation*} \widehat{\boldsymbol{Z}}_{\pi,v}=\sum_{k \in S_v}\frac{w_k\boldsymbol{T}_k}{N_v},
\end{equation*}
and
\begin{equation*}
    \boldsymbol{Z}_{v}=\sum_{k \in U_v}\frac{\mathbb{E}_q[w_k]\boldsymbol{T}_k}{N_v},
\end{equation*}
respectively. 
Assume \ref{D1}-\ref{D3}. If
  $\{\boldsymbol{T}_k\}_{k \in U_v}$ satisfies
    \begin{equation*}
        \limsup_{v \rightarrow \infty}\dfrac{1}{N_v}\sum_{k \in U_v}\E_m\left[\Arrowvert\boldsymbol{T}_k\Arrowvert_F^2\right]<\infty
    \end{equation*}
almost surely, then we have
\begin{equation*}
    \big\Arrowvert\widehat{\boldsymbol{Z}}_{\pi,v}-\boldsymbol{Z}_v\big\Arrowvert_{op}=\mathcal{O}_{\mathbb{P}}\left(\frac{1}{\sqrt{N_v}} \right).
\end{equation*}
\end{lemma}
\begin{proof}
   Recall that for arbitrary matrices $\boldsymbol{B}$, we have $\rVert \boldsymbol{B}\rVert_{op} \leq \rVert \boldsymbol{B}\rVert_{F}.$ For any $\epsilon>0$, an application of Chebyshev's inequality gives
   \begin{align*}
       \P_{mpq}\left( \Arrowvert\widehat{\boldsymbol{Z}}_{\pi,v}-\boldsymbol{Z}_v\Arrowvert_{op}>\epsilon\right)&\leq \frac{1}{\epsilon^2}\E_{mpq}\left[\Arrowvert \widehat{\boldsymbol{Z}}_{\pi,v}-\boldsymbol{Z}_v\Arrowvert_{op}^2\right]\\
       &\leq \frac{1}{\epsilon^2}\E_{mpq}\left[\Arrowvert \widehat{\boldsymbol{Z}}_{\pi,v}-\boldsymbol{Z}_v\Arrowvert_{F}^2\right].
   \end{align*}
   We proceed to evaluate
   \begin{align}\label{tempTrace1}
       \E_{mpq}\left[\Arrowvert \widehat{\boldsymbol{T}}_{\pi,v}-\boldsymbol{T}_v\Arrowvert_{F}^2\right]&=\E_{mpq}\left[\Bigg\Arrowvert\sum_{k \in S_v}\frac{w_k\boldsymbol{T}_k}{N_v}-\sum_{k \in U_v}\frac{\E_q\left[w_k\right]\boldsymbol{T}_k}{N_v} \Bigg\Arrowvert_F^2 \right]\nonumber\\
       &=\E_{mpq}\left[\mathrm{Tr}\left(\sum_{k \in U_v}\sum_{l \in U_v}\frac{(I_kw_k-\pi_k\E_q\left[w_k\right])(I_lw_l-\pi_l\E_q\left[w_l\right])\boldsymbol{T}_k^{\top}\boldsymbol{T}_l}{N_v^2} \right) \right]\nonumber\\
        &=\E_{mpq}\left[\sum_{k \in U_v}\sum_{l \in U_v}\frac{(I_kw_k-\pi_k\E_q\left[w_k\right])(I_lw_l-\pi_l\E_q\left[w_l\right])\mathrm{Tr}\left(\boldsymbol{T}_k^{\top}\boldsymbol{T}_l\right)}{N_v^2}  \right].
   \end{align}
   When $k=l$, we have $\E_{pq}[(I_kw_k-\pi_k\E_q\left[w_k\right])(I_lw_l-\pi_l\E_q\left[w_l\right])]=\pi_k(1-\pi_k)\V_q[w_k]$. When $k \neq l \in U_v$, we have $\E_{pq}[(I_kw_k-\pi_k\E_q\left[w_k\right])(I_lw_l-\pi_l\E_q\left[w_l\right])]=\Delta_{kl}\E_q\left[w_k\right]\E_q\left[w_l\right].$ As a result, \eqref{tempTrace1} reduces to
   \begin{align*}
       &\E_{mpq}\left[\sum_{k \in U_v}\sum_{l \in U_v}\frac{(I_kw_k-\pi_k\E_q\left[w_k\right])(I_lw_l-\pi_l\E_q\left[w_l\right])\mathrm{Tr}\left(\boldsymbol{T}_k^{\top}\boldsymbol{T}_l\right)}{N_v^2}  \right]\\
       &\quad=\sum_{k \in U_v}\frac{\pi_k(1-\pi_k)\V_q\left( w_k\right)}{N_v^2}\E_m\left[\mathrm{Tr}\left(\boldsymbol{T}_k^{\top}\boldsymbol{T}_k\right)\right]\\
    &\quad\quad+\sum_{k \in U_v}\sum_{\substack{l \in U_v\\l \neq k}}\frac{\Delta_{kl}\E_q\left[w_k\right]\E_q\left[w_k\right]}{N_v^2}\E_m\left[\mathrm{Tr}\left(\boldsymbol{T}_k^{\top}\boldsymbol{T}_l\right)\right]\\
       &\quad \leq\frac{b}{N_v} \sum_{k \in U_v}\frac{\E_m\left[\Arrowvert\boldsymbol{T}_k\Arrowvert_F^2\right]}{N_v}+\max_{k \neq l \in U_v}|\Delta_{kl}|\sum_{k \in U_v}\sum_{l \in U_v}\frac{|\E_q\left[w_k\right]\E_q\left[w_k\right]|}{N_v^2}\E_m\left[\big|\mathrm{Tr}\left(\boldsymbol{T}_k^{\top}\boldsymbol{T}_l\right)\big|\right]\\
        &\quad \leq\frac{b}{N_v} \sum_{k \in U_v}\frac{\E_m\left[\Arrowvert\boldsymbol{T}_k\Arrowvert_F^2\right]}{N_v}\\
        &\quad\quad+\frac{n_v\max_{k \neq l \in U_v}|\Delta_{kl}|}{n_v}\sum_{k \in U_v}\sum_{l \in U_v}\frac{|\E_q\left[w_k\right]\E_q\left[w_k\right]|}{N_v^2}\E_m\left[\Arrowvert\boldsymbol{T}_k\Arrowvert_F\Arrowvert\boldsymbol{T}_l\Arrowvert_F\right]\\
        &\quad \leq\frac{b}{N_v} \sum_{k \in U_v}\frac{\E_m\left[\Arrowvert\boldsymbol{T}_k\Arrowvert_F^2\right]}{N_v}+\frac{\bar{\Delta}}{n_v}\sum_{k \in U_v}\sum_{l \in U_v}\frac{|\E_q\left[w_k\right]\E_q\left[w_k\right]|}{N_v^2}\E_m^{1/2}\left[\Arrowvert\boldsymbol{T}_k\Arrowvert_F^2\right]\E_m^{1/2}\left[\Arrowvert\boldsymbol{T}_l\Arrowvert_F^2\right]\\
         &\quad \leq\frac{b}{N_v} \sum_{k \in U_v}\frac{\E_m\left[\Arrowvert\boldsymbol{T}_k\Arrowvert_F^2\right]}{N_v}+\frac{\bar{\Delta}}{n_v}\sum_{k \in U_v}\frac{\E_q^2\left[w_k\right]}{N_v}\E_m\left[\Arrowvert\boldsymbol{T}_k\Arrowvert_F^2\right]\\
          &\quad \leq\frac{b}{N_v} \sum_{k \in U_v}\frac{\E_m\left[\Arrowvert\boldsymbol{T}_k\Arrowvert_F^2\right]}{N_v}+\frac{\bar{\Delta}}{n_v}\sum_{k \in U_v}\frac{\E_q\left[w_k^2\right]}{N_v}\E_m\left[\Arrowvert\boldsymbol{T}_k\Arrowvert_F^2\right]\\
           &\quad \leq\frac{b}{N_v} \sum_{k \in U_v}\frac{\E_m\left[\Arrowvert\boldsymbol{T}_k\Arrowvert_F^2\right]}{N_v}+\frac{b\bar{\Delta}}{n_v}\sum_{k \in U_v}\frac{\E_m\left[\Arrowvert\boldsymbol{T}_k\Arrowvert_F^2\right]}{N_v}.
   \end{align*}
The result follows.
\end{proof}
\begin{lemma}\label{REstimator}
Let $(w_k)_{k \in U_v}$ be i.i.d. random variable satisfying $\E_q[w_k^2]\leq b$ for some $b \geq 0$ almost surely. Let $(\widehat{\boldsymbol{Z}}_{p,v})_{v \in \mathbb{N}}$ and $(\boldsymbol{Z}_v)_{v\in \mathbb{N}}$ be a sequence of matrices in $\R^{d_1 \times d_2}$ defined by 
\begin{equation*}
    \widehat{\boldsymbol{Z}}_{p,v}=\sum_{k \in U_v}\frac{w_k\boldsymbol{T}_k}{N_v},
\end{equation*}
and
\begin{equation*}
    \boldsymbol{Z}_{v}=\sum_{k \in U_v}\frac{\mathbb{E}_q[w_k]\boldsymbol{T}_k}{N_v},
\end{equation*}
respectively. If
  $(\boldsymbol{T}_k)_{k \in U_v}$ satisfies
    \begin{equation*}
        \limsup_{v \rightarrow \infty}\sum_{k \in U_v}\frac{\E_m\left[\Arrowvert\boldsymbol{T}_k\Arrowvert_F^2\right]}{N_v}<\infty
    \end{equation*}
almost surely, then we have
\begin{equation*}
    \big\Arrowvert\widehat{\boldsymbol{Z}}_{p,v}-\boldsymbol{Z}_v\big\Arrowvert_{op}=\mathcal{O}_{\mathbb{P}}\left(\frac{1}{\sqrt{N_v}} \right).
\end{equation*}
\end{lemma}
\begin{proof}
    By evaluating $\E_{mpq}[\Arrowvert  \widehat{\boldsymbol{Z}}_{p,v}-\boldsymbol{Z}_v\Arrowvert_{op}^2]$, we have
    \begin{align*}
        \E_{mpq}\left[\Arrowvert  \widehat{\boldsymbol{Z}}_{p,v}-\boldsymbol{Z}_v\Arrowvert_{op}^2\right] &\leq \E_{mq}\left[ \Bigg\Arrowvert\sum_{k \in U_v}\frac{\left(w_k-\E_m\left[w_k \right] \right)\boldsymbol{T}_k}{N_v}\Bigg\Arrowvert_{F}^2\right]\\
        &=\E_{mq}\left[\sum_{k \in U_v}\sum_{l \in U_v}\frac{\left(w_k-\E_m\left[w_k \right] \right)\left(w_l-\E_m\left[w_l \right] \right)\mathrm{Tr}\left(\boldsymbol{T}_k^{\top}\boldsymbol{T}_l \right)}{N_v^2} \right]\\
        &=\sum_{k \in U_v}\frac{\V_q(w_k)\E_m\left[\mathrm{Tr}\left(\boldsymbol{T}_k^{\top}\boldsymbol{T}_k \right)\right] }{N_v^2}\\
         &\leq \sum_{k \in U_v}\frac{\E_q(w_k^2)\E_m\left[\Arrowvert\boldsymbol{T}_k\Arrowvert_F^2 \right] }{N_v^2}\\
        &\leq \frac{b}{N_v}\sum_{k \in U_v}\frac{\E_m\left[\Arrowvert\boldsymbol{T}_k\Arrowvert_F^2 \right]}{N_v},
    \end{align*}
from which the result follows.
\end{proof}
\begin{lemma} \label{MatrixOverfitting}
For any $n \in \N$, the map $\cL_{2,n} : 
\cP(\{ 1, 2,..., p\}) \longrightarrow \R_+$ defined by $$\alpha \ \longmapsto  \sigma^2 \left( \sum_{k\in S_m} \dfrac{\bx_{k,\alpha}^\top }{\pi_k} \right) \boldsymbol{A}_{r, \alpha}^{-1} \left( \sum_{k\in S_m}\dfrac{\bx_{k,\alpha} }{\pi_k} \right)$$
is, almost surely, a strictly increasing set function, that is, $\alpha_1 \subset \alpha_2$ implies $\cL_2(\alpha_1) <\cL_2(\alpha_2) $.
\end{lemma}
\begin{proof}
Let $\alpha_1, \alpha_2 \in \cP\{ 1, 2,..., p\})$ such that $\alpha_1\subset \alpha_2.$ Then, 
\begin{align*}
& \cL_{2,n} \left(\alpha_1\right) < \cL_{2,n} \left(\alpha_2\right)\\&\quad\Leftrightarrow \left(\sum_{k \in S_{m}}\frac{\bx_{k,\alpha_1}^{\top}}{N\pi_k} \right)\boldsymbol{A}_{r, \alpha_1}^{-1}\left(\sum_{k \in S_{m}}\frac{\bx_{k,\alpha_1}}{N\pi_k} \right) < \left(\sum_{k \in S_{m}}\frac{\bx_{k,\alpha_2}^{\top}}{N\pi_k} \right)\boldsymbol{A}_{r, \alpha_2}^{-1} \left(\sum_{k \in S_{m}}\frac{\bx_{k,\alpha_2}}{N\pi_k} \right).
\end{align*}
This amounts to showing that one quadratic form is almost surely less than the other. We start by expressing $\boldsymbol{A}_{r, \alpha_2}^{-1} $ in terms of $\boldsymbol{A}_{r, \alpha_1}^{-1} $. Let $\boldsymbol{A}_{r,\alpha_1, \alpha_2-\alpha_1} := \bx_{r,\alpha_1}^\top \bx_{r,\alpha_2-\alpha_1}$ for arbitrary $\alpha_1 \subset  \alpha_2$. Note that 
$$
\boldsymbol{A}_{r,\alpha_2} =
\begin{pmatrix}
\boldsymbol{A}_{r,\alpha_1} & \boldsymbol{A}_{r,\alpha_1, \alpha_2 - \alpha_1} \\
\rule{0pt}{3ex}\boldsymbol{A}_{r,\alpha_2 - \alpha_1, \alpha_1} & \boldsymbol{A}_{r,\alpha_2 - \alpha_1}
\end{pmatrix}.
$$

Using a \hyperlink{BMIF}{block-matrix inversion formula} (see, e.g., \cite{horn2012matrix}, page 25), one can show that  
$$
\boldsymbol{A}_{r,\alpha_2}^{-1} = \begin{pmatrix}
\boldsymbol{A}_{r,\alpha_1}^{-1} + \boldsymbol{A}_{r,\alpha_1}^{-1} \boldsymbol{A}_{r,\alpha_1, \alpha_2-\alpha_1} \boldsymbol{S}^{-1}  \boldsymbol{A}_{r,\alpha_2-\alpha_1,\alpha_1}\boldsymbol{A}_{r,\alpha_1}^{-1} & -\boldsymbol{A}_{r,\alpha_1}^{-1}\boldsymbol{A}_{r,\alpha_1, \alpha_2-\alpha_1}  \boldsymbol{S}^{-1} \\
-\rule{0pt}{4ex}\boldsymbol{S}^{-1}\boldsymbol{A}_{r,\alpha_2-\alpha_1,\alpha_1}\boldsymbol{A}_{r,\alpha_1}^{-1} & \boldsymbol{S}^{-1}
\end{pmatrix}.
$$
Therefore, writing $$ \left(\sum_{k \in S_{m}}\frac{\bx_{k,\alpha_2}^{\top}}{N\pi_k} \right)  = \left(\sum_{k \in S_{m}}\frac{\bx_{k,\alpha_1}^{\top}}{N\pi_k},\sum_{k \in S_{m}}\frac{\bx_{k,\alpha_2-\alpha_1}^{\top}}{N\pi_k} \right),$$ some algebra shows that  
\begin{align}\label{prooflm2}
& \left(\sum_{k \in S_{m}}\frac{\bx_{k,\alpha_1}^{\top}}{N\pi_k},\sum_{k \in S_{m}}\frac{\bx_{k,\alpha_2-\alpha_1}^{\top}}{N\pi_k} \right)\boldsymbol{A}_{r,\alpha_2}^{-1} \left(\sum_{k \in S_{m}}\frac{\bx_{k,\alpha_1}}{N\pi_k},\sum_{k \in S_{m}}\frac{\bx_{k,\alpha_2-\alpha_1}}{N\pi_k} \right)\nonumber \\ &\quad=  \left(\sum_{k \in S_{m}}\frac{\bx_{k,\alpha_1}^{\top}}{N\pi_k} \right)\boldsymbol{A}_{r, \alpha_1}^{-1} \left(\sum_{k \in S_{m}}\frac{\bx_{k,\alpha_1}}{N\pi_k} \right)\nonumber\\&\quad\quad+ \left(\sum_{k \in S_{m}}\frac{\bx_{k,\alpha_2-\alpha_1}^{\top}}{N\pi_k}  -\sum_{k \in S_m}\frac{\bx_{k,\alpha_1}^{\top}\boldsymbol{A}_{r,\alpha_1}^{-1}\boldsymbol{A}_{r,\alpha_2-\alpha_1}}{N\pi_k}\right) \nonumber\\
&\quad\quad\quad\times\boldsymbol{S}^{-1} \left(\sum_{k \in S_{m}}\frac{\bx_{k,\alpha_2-\alpha_1}^{\top}}{N\pi_k}  -\sum_{k \in S_m}\frac{\bx_{k,\alpha_1}^{\top}\boldsymbol{A}_{r,\alpha_1}^{-1}\boldsymbol{A}_{r,\alpha_2-\alpha_1}}{N\pi_k}\right)^{\top}.
\end{align}
Hence, showing   $ \cL_{2,n} (\alpha_1) < \cL_{2,n}(\alpha_2)$ amounts to show that \begin{align} \label{prooflM2}
&\left(\sum_{k \in S_{m}}\frac{\bx_{k,\alpha_2-\alpha_1}^{\top}}{N\pi_k}  -\sum_{k \in S_m}\frac{\bx_{k,\alpha_1}^{\top}\boldsymbol{A}_{r,\alpha_1}^{-1}\boldsymbol{A}_{r,\alpha_2-\alpha_1}}{N\pi_k}\right) \nonumber\\
&\times\boldsymbol{S}^{-1} \left(\sum_{k \in S_{m}}\frac{\bx_{k,\alpha_2-\alpha_1}^{\top}}{N\pi_k}  -\sum_{k \in S_m}\frac{\bx_{k,\alpha_1}^{\top}\boldsymbol{A}_{r,\alpha_1}^{-1}\boldsymbol{A}_{r,\alpha_2-\alpha_1}}{N\pi_k}\right)^{\top}>0
\end{align} almost surely. To that aim, observe that $\boldsymbol{S} $ can be written as $$\boldsymbol{S} = \bx_{r,\alpha_2 - \alpha_1}^\top \left( \boldsymbol{I}_{n_r} -\bx_{r,\alpha_1} \boldsymbol{A}_{r,\alpha_1}^{-1}\bx_{r,\alpha_1}^\top \right) \bx_{r,\alpha_2-\alpha_1}.$$ The matrix $\boldsymbol{S}$ is a gram matrix of linearly independent vectors, thus positive definite, which shows \eqref{prooflM2} and thus $\cL_{2,n} (\alpha_1) < \cL_{2,n} (\alpha_2)$ holds almost surely. 
\end{proof}
\begin{lemma}\label{AlgebraEx}
    Let $\alpha \in \cA$ and assume that the intercept is included in $\alpha$. Then,
    \begin{equation*}
\mathbb{E}_\bx\left[(1-p(\bx))\bx_{\alpha}^{\top}\right]\left(\mathbb{E}_\bx\left[p(\bx)\bx_{\alpha}\bx_{\alpha}^{\top}\right]\right)^{-1}\mathbb{E}_\bx\left[p(\bx)\bx_{\alpha}\right]=\mathbb{E}_\bx\left[1-p(\bx)\right].
\end{equation*}
\end{lemma}
\begin{proof}
Write $\bx_{\alpha}^{\top}=\left(1,\bx_{\alpha-1}^{\top}\right)$, where $\bx_{\alpha-1}^{\top}$ is the covariate without the intercept. As a result,
\begin{align*}
\mathbb{E}_\bx\left[p(\bx)\bx_{\alpha}\bx_{\alpha}^{\top}\right]=\begin{pmatrix}
\mathbb{E}_\bx[p(\bx)]&\mathbb{E}_\bx[p(\bx)\bx_{\alpha-1}^{\top}]\\
\mathbb{E}_\bx[p(\bx)\bx_{\alpha-1}]&\mathbb{E}_\bx[p(\bx)\bx_{\alpha-1}\bx_{\alpha-1}^{\top}]
\end{pmatrix}.
\end{align*}
Using the \hyperlink{BMIF}{block matrix inverse formula}, we obtain
\begin{align*}
&\left(\mathbb{E}_\bx\left[p(\bx)\bx_{\alpha}\bx_{\alpha}^{\top}\right]\right)^{-1}&\\&=\begin{pmatrix}
\frac{1}{\mathbb{E}_\bx[p(\bx)]}+\frac{1}{\mathbb{E}_\bx^2[p(\bx)]}\mathbb{E}_\bx\left[p(\bx)\bx_{\alpha-1}^{\top}\right]\boldsymbol{M}^{-1}_{\alpha-1}\mathbb{E}_\bx\left[p(\bx)\bx_{\alpha-1}\right]& -\frac{1}{\mathbb{E}_\bx[p(\bx)]}\boldsymbol{B}_{\alpha-1}\\
-\frac{1}{\mathbb{E}_\bx[p(\bx)]}\boldsymbol{B}_{\alpha-1}^{\top}& \boldsymbol{M}_{\alpha-1}^{-1}
\end{pmatrix},
\end{align*}
where $\boldsymbol{M}_{\alpha-1}=\mathbb{E}_\bx\left[p(\bx)\bx_{\alpha-1}\bx_{\alpha-1}^{\top}\right]-\mathbb{E}_\bx\left[p(\bx)\bx_{\alpha-1}\right]\left(\mathbb{E}_\bx[p(\bx)]\right)^{-1}\mathbb{E}_\bx\left[p(\bx)\bx_{\alpha-1}^{\top}\right]$ is a positive and definite matrix and $\boldsymbol{B}_{\alpha-1}=\mathbb{E}_\bx\left[p(\bx)\bx_{\alpha-1}^{\top}\right]\boldsymbol{M}^{-1}_{\alpha-1}$. As a result, 
\begin{align*}
&\left(\mathbb{E}_\bx\left[p(\bx)\bx_{\alpha}\bx_{\alpha}^{\top}\right]\right)^{-1}\mathbb{E}_\bx\left[p(\bx)\bx_{\alpha}\right]\\&\quad=\begin{pmatrix}
\frac{1}{\mathbb{E}_\bx[p(\bx)]}+\frac{1}{\mathbb{E}_\bx^2[p(\bx)]}\mathbb{E}_\bx\left[p(\bx)\bx_{\alpha-1}^{\top}\right]\boldsymbol{M}^{-1}_{\alpha-1}\mathbb{E}_\bx\left[p(\bx)\bx_{\alpha-1}\right]& -\frac{1}{\mathbb{E}_\bx[p(\bx)]}\boldsymbol{B}_{\alpha-1}\\
-\frac{1}{\mathbb{E}_\bx[p(\bx)]}\boldsymbol{B}_{\alpha-1}^{\top}& \boldsymbol{M}_{\alpha-1}^{-1}
\end{pmatrix}\\
&\quad\quad\times\begin{pmatrix}
\mathbb{E}_\bx\left[p(\bx)\right]\\
\mathbb{E}_\bx\left[p(\bx)\bx_{\alpha-1}\right]
\end{pmatrix}\\
&\quad= \begin{pmatrix}
1\\
\textbf{0}_{\alpha-1}
\end{pmatrix}.
\end{align*}
This leads to
\begin{align*}
&\mathbb{E}_\bx\left[(1-p(\bx))\bx_{\alpha}^{\top}\right]\left(\mathbb{E}_\bx\left[p(\bx)\bx_{\alpha}\bx_{\alpha}^{\top}\right]\right)^{-1}\mathbb{E}_\bx\left[p(\bx)\bx_{\alpha}\right]\\&\quad=\left(\mathbb{E}_\bx\left[1-p(\bx)\right],\mathbb{E}_\bx\left[(1-p(\bx))\bx_{\alpha-1}^{\top}\right]\right)\begin{pmatrix}
1\\
\textbf{0}_{\alpha-1}
\end{pmatrix}=\mathbb{E}_\bx[1-p(\bx)].
\end{align*}
\end{proof}
\begin{lemma} \label{SetofRespondents}
Assume \ref{S1}-\ref{S2}. Assume also that there exists $\rho>0$ such that $p(\bx)>\rho$ almost surely. Then, there exists $\xi>0$ such that 
\begin{equation*}
\lim_{v \to \infty} \mathbb{P}_{mpq} \left( \dfrac{n_{r,v}}{N_v} \geq \xi\right)=1
\end{equation*}
almost surely.
\end{lemma}

\begin{proof}
Observe that $$\mathbb{E}_q \left[ \dfrac{n_{r,v}}{N_v} \right] = \dfrac{1}{n_v} \sum_{k\in S_v} p\left(\bx_k\right), \qquad \text{and,} \qquad \mathbb{V}_q \left( \dfrac{n_{r,v}}{n_v}\right) = \dfrac{1}{n_v^2} \sum_{k\in S_v}  p\left(\bx_k\right) \left\{ 1- p\left(\bx_k\right)\right\} \leq \dfrac{1}{4n_v}.$$
Therefore, Chebyshev's inequality shows that, for all $\epsilon>0$, $$ \mathbb{P}_q \left( \bigg\rvert \dfrac{n_{r,v}}{n_v} - \dfrac{1}{n_v} \sum_{k\in S_v} p\left(\bx_k\right) \bigg\rvert  >\epsilon \right)\leq \dfrac{1}{4\epsilon^2n_v}\xrightarrow[v \to \infty]{}0$$ 
almost surely.
On the one hand, $$\dfrac{1}{n_v} \sum_{k\in S_v} p\left(\bx_k\right) \geq \min_{k\in S_v} p\left(\bx_k\right) \geq \rho$$
almost surely. On the other hand, 
\begin{align*}
\mathbb{P}_q \left( \bigg\rvert\dfrac{n_{r,v}}{n_v} - \dfrac{1}{n_v} \sum_{k\in S_v} p\left(\bx_k\right)\bigg\rvert\leq\epsilon\right) &=  \mathbb{P}_q \left(  \dfrac{1}{n_v} \sum_{k\in S_v} p\left(\bx_k\right)-\epsilon\leq \dfrac{n_{r,v}}{n_v} \leq  \dfrac{1}{n_v} \sum_{k\in S_v} p\left(\bx_k\right) +\epsilon  \right)  \\&\leq \mathbb{P}_q \left( \rho-\epsilon\leq \dfrac{n_{r,v}}{n_v}   \right).
\end{align*}
Using that $\min_{k \in U_v}p(\bx_k)\geq\rho$, we set $\xi := \rho - \epsilon$ from which it follows that $$\lim_{v\rightarrow \infty}\mathbb{P}_q \left( \xi\leq \dfrac{n_{r,v}}{n_v}   \right) \geq         \lim_{v \rightarrow \infty}\mathbb{P}_q \left( \bigg\rvert\dfrac{n_{r,v}}{n_v} - \dfrac{1}{n_v} \sum_{k\in S_v} p\left(\bx_k\right)\bigg\rvert>\epsilon\right) =1 $$ almost surely. Furthermore,
\begin{equation*}
\lim_{v \to \infty} \mathbb{P}_{mpq} \left( \dfrac{n_{r,v}}{N_v} \geq \xi\right)=1
\end{equation*}
almost surely.
\end{proof}

\begin{lemma}\label{MatrixInverse}
Let $\{ \boldsymbol{A}_v \}_{v \in \mathbb{N}}$ and $\{ \boldsymbol{B}_v \}_{\mathbb{N}}$ be $p \times p$ invertible matrices such that $\Arrowvert \boldsymbol{A}_v^{-1} \Arrowvert_{op}=\mathcal{O}_\P(1)$ and $\Arrowvert \boldsymbol{B}_v^{-1} \Arrowvert_{op}=\mathcal{O}_\P(1)$. If there exists a sequence $\{u_v\}_{v \in \mathbb{N}}$ such that 
\begin{equation*}
\Arrowvert\boldsymbol{A}_v-\boldsymbol{B}_v\Arrowvert_{op}=\mathcal{O}_\P\left(u_v^{-1}\right).
\end{equation*}
Then, we have
\begin{equation*}
\Arrowvert\boldsymbol{A}_v^{-1}-\boldsymbol{B}_v^{-1}\Arrowvert_{op}=\mathcal{O}_\P\left(u_v^{-1}\right).
\end{equation*}
\end{lemma}
\begin{proof}
Recall the \hypertarget{SMI}{Schwarz matrix inequality} (see \cite{hansen2022econometrics}, B.15 in Page 981), for any squared matrices $\boldsymbol{A}$ and $\boldsymbol{B}$, for the spectral norm, $\Arrowvert\boldsymbol{A}\boldsymbol{B}\Arrowvert_{op}\leq \Arrowvert\boldsymbol{A}\Arrowvert_{op}\Arrowvert\boldsymbol{B}\Arrowvert_{op}$. As a result,
\begin{align*}
\Arrowvert\boldsymbol{A}_v^{-1}-\boldsymbol{B}_v^{-1}\Arrowvert_{op} &= \Arrowvert \boldsymbol{A}_v^{-1}  \left(\textbf{I}_p-\boldsymbol{A}_v\boldsymbol{B}_v^{-1}\right)\Arrowvert_{op}\leq \Arrowvert\boldsymbol{A}_v^{-1}\Arrowvert_{op} \Arrowvert\boldsymbol{I}_p-\boldsymbol{A}_v\boldsymbol{B}_v^{-1}\Arrowvert_{op}\\
&=\Arrowvert\boldsymbol{A}_v^{-1}\Arrowvert_{op} \Arrowvert\left(\boldsymbol{B}_v-\boldsymbol{A}_v\right)\boldsymbol{B}_v^{-1}\Arrowvert_{op}\\
&\leq\Arrowvert\boldsymbol{A}_v^{-1}\Arrowvert_{op} \Arrowvert\boldsymbol{B}_v^{-1}\Arrowvert_{op}\Arrowvert\textbf{B}_v-\boldsymbol{A}_v\Arrowvert_{op}=\mathcal{O}_\P\left(u_v^{-1}\right).
\end{align*}
This concludes the proof since $\Arrowvert\boldsymbol{B}_v-\boldsymbol{A}_v\Arrowvert_{op}=\Arrowvert\boldsymbol{A}_v-\boldsymbol{B}_v\Arrowvert_{op}$ for all $v \in \mathbb{N}$.
\end{proof}

\subsubsection{Tightness of various statistics of interest}
\begin{lemma} \label{BoundnessAp-1}
 Assume \ref{S1}-\ref{S2} and \ref{D1}-\ref{D3}. For $\alpha \in \cA$, we have
    \begin{equation*}
    \Bigg\Arrowvert \left(\sum_{k \in U_v}\frac{p(\bx_k)\pi_k\bx_{k,\alpha}\bx_{k,\alpha}^{\top}}{N_v} \right)^{-1}\Bigg\Arrowvert_{op}=\mathcal{O}_\P(1).
    \end{equation*}
\end{lemma}
\begin{proof}
Using \ref{D2} gives
    \begin{equation*}
         \Bigg\Arrowvert \sum_{k \in U_v}\frac{p(\bx_k)\pi_k\bx_{k,\alpha}\bx_{k,\alpha}^{\top}}{N_v} \Bigg\Arrowvert_{op}\geq \lambda\Bigg\Arrowvert \sum_{k \in U_v}\frac{p(\bx_k)\bx_{k,\alpha}\bx_{k,\alpha}^{\top}}{N_v} \Bigg\Arrowvert_{op},
    \end{equation*}
    On the other hand, using the law of large numbers and the continuous mapping theorem gives
    \begin{equation*}
        \lambda_{min}\left(\sum_{k \in U_v}\frac{p(\bx_k)\bx_{k,\alpha}\bx_{k,\alpha}^{\top}}{N_v} \right)\xrightarrow[v\rightarrow \infty]{a.s.}\lambda_{min}\left(\E_{\bx}\left[p(\bx_1)\bx_{1,\alpha}\bx_{1,\alpha}\right]\right),
    \end{equation*}
    which is a positive-definite matrix. 
    As a result, for any $\epsilon>0$, we have
    \begin{align*}
        &\lim_{v \rightarrow \infty}\P_{ mpq}\left(\Bigg|\lambda_{min}\left(\sum_{k \in U_v}\frac{p(\bx_k)\bx_{k,\alpha}\bx_{k,\alpha}^{\top}}{N_v}\right)- \lambda_{min}\left(\E_{\bx}\left[p(\bx_1)\bx_{1,\alpha}\bx_{1,\alpha}\right]\right)\Bigg|\leq \epsilon\right)\\
        &\quad \leq \lim_{v \rightarrow \infty}\P_{ mpq}\left(\lambda_{min}\left(\sum_{k \in U_v}\frac{p(\bx_k)\bx_{k,\alpha}\bx_{k,\alpha}^{\top}}{N_v}\right)\geq \lambda_{min}\left(\E_{\bx}\left[p(\bx_1)\bx_{1,\alpha}\bx_{1,\alpha}\right]\right)- \epsilon\right)=1.
    \end{align*}
    Furthermore,
    \begin{align*}
         \lim_{v \rightarrow \infty}\P_{ mpq}\left(\lambda_{min}\left(\sum_{k \in U_v}\frac{p(\bx_k)\pi_k\bx_{k,\alpha}\bx_{k,\alpha}^{\top}}{N_v}\right)\geq \lambda \lambda_{min}\left(\E_{\bx}\left[p(\bx_1)\bx_{1,\alpha}\bx_{1,\alpha}\right]\right)- \epsilon\right)=1
    \end{align*}
    almost surely.
    This concludes
    \begin{equation*}
        \Bigg\Arrowvert \left(\sum_{k \in U_v}\frac{p(\bx_k)\pi_k\bx_{k,\alpha}\bx_{k,\alpha}^{\top}}{N_v} \right)^{-1}\Bigg\Arrowvert_{op}=\lambda_{min}^{-1}\left(\sum_{k \in U_v}\frac{p(\bx_k)\pi_k\bx_{k,\alpha}\bx_{k,\alpha}^{\top}}{N_v}\right)=\mathcal{O}_\P(1).
    \end{equation*}
\end{proof}

\begin{lemma} \label{FiniteAp-1}
 Assume \ref{S1}-\ref{S2} and \ref{D1}-\ref{D3}. For $\alpha \in \cA$, we have
    \begin{equation*}
    \limsup_{v \rightarrow \infty}\Bigg\Arrowvert \left(\sum_{k \in U_v}\frac{p(\bx_k)\pi_k\bx_{k,\alpha}\bx_{k,\alpha}^{\top}}{N_v} \right)^{-2}\Bigg\Arrowvert_{op}<\infty.
    \end{equation*}
\end{lemma}
\begin{proof}
We have    \begin{equation*}
         \Bigg\Arrowvert \sum_{k \in U_v}\frac{p(\bx_k)\pi_k\bx_{k,\alpha}\bx_{k,\alpha}^{\top}}{N_v} \Bigg\Arrowvert_{op}\leq \lambda\Bigg\Arrowvert \sum_{k \in U_v}\frac{p(\bx_k)\bx_{k,\alpha}^{\top}\bx_{k,\alpha}}{N_v} \Bigg\Arrowvert_{op},
    \end{equation*}
Also,
    \begin{equation*}
        \lambda_{min}^{-2}\left(\sum_{k \in U_v}\frac{p(\bx_k)\bx_{k,\alpha}\bx_{k,\alpha}^{\top}}{N_v} \right)\xrightarrow[v\rightarrow \infty]{a.s.}\lambda_{min}^{-2}\left(\E_{\bx}\left[p(\bx_1)\bx_{1,\alpha}\bx_{1,\alpha}\right]\right),
    \end{equation*}
    which is a positive-definite matrix. As a result,
    \begin{equation*}
         \limsup_{v \rightarrow \infty}\lambda_{min}^{-2}\left(\sum_{k \in U_v}\frac{\bx_{k,\alpha}\bx_{k,\alpha}^{\top}p(\bx_k)}{N_v} \right)<\infty
    \end{equation*}
    almost surely.
    It follows that
    \begin{equation*}
        \limsup_{v \rightarrow \infty}\Bigg\Arrowvert \left(\sum_{k \in U_v}\frac{p(\bx_k)\pi_k\bx_{k,\alpha}\bx_{k,\alpha}^{\top}}{N_v} \right)^{-2}\Bigg\Arrowvert_{op}= \limsup_{v \rightarrow \infty}\lambda_{min}^{-2}\left(\sum_{k \in U_v}\frac{p(\bx_k)\pi_k\bx_{k,\alpha}\bx_{k,\alpha}^{\top}}{N_v} \right)<\infty
    \end{equation*}
    almost surely.

\end{proof}

\begin{lemma}\label{BoundnessApi-1}
    Assume \ref{S1}-\ref{S2} and \ref{D1}-\ref{D3}. For $\alpha \in \cA$, we have
    \begin{equation*}
    \Bigg\Arrowvert \left(\sum_{k \in U_v} \frac{r_k\pi_k\bx_{k,\alpha}\bx_{k,\alpha}^{\top}}{N_v}\right)^{-1}\Bigg\Arrowvert_{op}=\mathcal{O}_\P(1).
    \end{equation*}
\end{lemma}
\begin{proof}
 Recall from Lemma \ref{HighProb}, there exists universal $C_1$, $C_2$, and $\gamma>0$ such that
\begin{align*}
     \P_q\left(\lambda_{min}\left(\sum_{k \in U_v}\frac{r_k\pi_k\bx_{k,\alpha}\bx_{k,\alpha}^{\top}}{N_v} \right)\geq \gamma \right)\geq 1-C_1e^{-C_2N_v}
\end{align*}
almost surely for large $v$.
As a result, we obtain
\begin{align*}
     \lim_{v \rightarrow \infty}\P_{mpq}\left(\lambda_{min}\left(\sum_{k \in U_v}\frac{r_k\pi_k\bx_{k,\alpha}\bx_{k,\alpha}^{\top}}{N_v} \right)\geq \gamma \right)=1.
\end{align*}
It follows that
\begin{equation*}
    \Bigg\Arrowvert \left(\sum_{k \in U_v} \frac{r_k\pi_k\bx_{k,\alpha}\bx_{k,\alpha}^{\top}}{N_v}\right)^{-1}\Bigg\Arrowvert_{op}=\lambda_{min}^{-1}\left(\sum_{k \in U_v}\frac{r_k\pi_k\bx_{k,\alpha}\bx_{k,\alpha}^{\top}}{N_v} \right)=\mathcal{O}_\P(1).
    \end{equation*}

This concludes the proof.
\end{proof}
\begin{lemma}\label{BoundnessA-1}
    Assume \ref{S1}-\ref{S2} and \ref{D1}-\ref{D3}. For $\alpha \in \cA$, we have
    \begin{equation*}
    \Bigg\Arrowvert \left(\frac{\boldsymbol{A}_{r,\alpha}}{N_v} \right)^{-1}\Bigg\Arrowvert_{op}=\mathcal{O}_\P(1).
    \end{equation*}
\end{lemma}
\begin{proof}
   Recall from \eqref{HTMatrix} and \eqref{RMatrix}, we obtain
   \begin{equation*}
       \Bigg\Arrowvert\frac{\boldsymbol{A}_{r,\alpha}}{N_v}-\sum_{k \in U_v}\frac{p(\bx_k)\pi_k\bx_{k,\alpha}\bx_{k,\alpha}^{\top}}{N_v}\Bigg\Arrowvert_{op}=\mathcal{O}_\P\left( \frac{1}{\sqrt{n_v}}\right)
   \end{equation*}
   and
   \begin{equation*}
       \Bigg\Arrowvert\sum_{k \in U_v}\frac{r_k\pi_k\bx_{k,\alpha}\bx_{k,\alpha}^{\top}}{N_v}-\sum_{k \in U_v}\frac{p(\bx_k)\pi_k\bx_{k,\alpha}\bx_{k,\alpha}^{\top}}{N_v}\Bigg\Arrowvert_{op}=\mathcal{O}_\P\left( \frac{1}{\sqrt{n_v}}\right),
   \end{equation*}
   respectively.
   An application of the triangle inequality of operator norm gives
\begin{equation*}
    \Bigg\Arrowvert\frac{\boldsymbol{A}_{r,\alpha}}{N_v}-\sum_{k \in U_v}\frac{r_k\pi_k\bx_{k,\alpha}\bx_{k,\alpha}^{\top}}{N_v}\Bigg\Arrowvert_{op}=\mathcal{O}_\P\left( \frac{1}{\sqrt{n_v}}\right).
\end{equation*}
It follows that
\begin{align*}
    \Bigg|\lambda_{min}\left(\frac{\boldsymbol{A}_{r,\alpha}}{N_v}\right)-\lambda_{min}\left(\sum_{k \in U_v}\frac{r_k\pi_k\bx_{k,\alpha}\bx_{k,\alpha}^{\top}}{N_v} \right) \Bigg|&\leq\Bigg\Arrowvert\frac{\boldsymbol{A}_{r,\alpha}}{N_v}-\sum_{k \in U_v}\frac{r_k\pi_k\bx_{k,\alpha}\bx_{k,\alpha}^{\top}}{N_v}\Bigg\Arrowvert_{op}\\
    &=\mathcal{O}_\P\left( \frac{1}{\sqrt{n_v}}\right).
\end{align*}
As a result, for any $\epsilon>0$
\begin{align*}
    &\lim_{v \rightarrow \infty}\P_{mpq}\left( \Bigg|\lambda_{min}\left(\frac{\boldsymbol{A}_{r,\alpha}}{N_v}\right)-\lambda_{min}\left(\sum_{k \in U_v}\frac{r_k\pi_k\bx_{k,\alpha}\bx_{k,\alpha}^{\top}}{N_v} \right) \Bigg|\leq \epsilon\right)\\&\quad\leq  \lim_{v \rightarrow \infty}\P_{mpq}\left( \lambda_{min}\left(\frac{\boldsymbol{A}_{r,\alpha}}{N_v}\right)\geq\lambda_{min}\left(\sum_{k \in U_v}\frac{r_k\pi_k\bx_{k,\alpha}\bx_{k,\alpha}^{\top}}{N_v} \right)-\epsilon\right)=1.
\end{align*}
Recall from Lemma \ref{BoundnessApi-1}, there exists $\gamma>0$ such that
\begin{equation*}
    \lim_{v \rightarrow \infty}\P_{mpq} \left(\lambda_{min}\left(\sum_{k \in U_v}\frac{r_k\pi_k\bx_{k,\alpha}\bx_{k,\alpha}^{\top}}{N_v} \right)\geq \gamma \right)=1.
\end{equation*}
Note that for any random variables $X, Y$ and any $\gamma,\epsilon>0$, we have
    \begin{equation*}
        \P_{mpq}\left(Y\geq \gamma\right)\leq \P_{mpq}\left(X\geq \gamma-\epsilon\right)+\P(|Y-X|>\epsilon).
    \end{equation*}
Thus
\begin{align*}
    &\P_{mpq} \left(\lambda_{min}\left(\sum_{k \in U_v}\frac{r_k\pi_k\bx_{k,\alpha}\bx_{k,\alpha}^{\top}}{N_v} \right)\geq \gamma \right)\\
    &\quad\leq \P_{mpq}\left( \lambda_{min}\left(\frac{\boldsymbol{A}_{r,\alpha}}{N_v}\right)\geq\gamma-\epsilon\right)\\
    &\quad\quad+\P_{mpq}\left( \Bigg|\lambda_{min}\left(\frac{\boldsymbol{A}_{r,\alpha}}{N_v}\right)-\lambda_{min}\left(\sum_{k \in U_v}\frac{r_k\pi_k\bx_{k,\alpha}\bx_{k,\alpha}^{\top}}{N_v} \right) \Bigg|> \epsilon\right).
\end{align*}
It follows that
\begin{equation*}
    \lim_{v \rightarrow \infty}\P_{mpq}\left( \lambda_{min}\left(\frac{\boldsymbol{A}_{r,\alpha}}{N_v}\right)\geq\gamma-\epsilon\right)=1.
\end{equation*}
This suffices to conclude
    \begin{equation*}
    \Bigg\Arrowvert \left(\frac{\boldsymbol{A}_{r,\alpha}}{N_v} \right)^{-1}\Bigg\Arrowvert_{op}=\mathcal{O}_\P(1).
    \end{equation*}

\end{proof}

\begin{lemma}\label{finitetildec}
    Let $(\widetilde{\textbf{c}}_{\alpha,v})_{v \in \mathbb{N}}$ be a sequence of estimators given by \eqref{tc}. Assume \ref{S1}-\ref{S3} and \ref{D1}-\ref{D3}. We have
    \begin{equation*}
        \limsup_{v \rightarrow \infty}\Arrowvert\widetilde{\textbf{c}}_{\alpha,v}\Arrowvert_2^4<\infty.
    \end{equation*}
\end{lemma}
\begin{proof}
    Recall that
    \begin{align*}
        \Arrowvert\widetilde{\textbf{c}}_{\alpha,v}\Arrowvert_2&=\Bigg\Arrowvert\left(\sum_{k \in U_v}\frac{p(\bx_k)\pi_k\bx_{k,\alpha}\bx_{k,\alpha}^{\top}}{N_v}\right)^{-1}\sum_{k \in U_v}\frac{(1-p(\bx_k))\bx_{k,\alpha}}{N_v}\Bigg\Arrowvert_2^4\\
        &\leq \Bigg\Arrowvert\left(\sum_{k \in U_v}\frac{p(\bx_k)\pi_k\bx_{k,\alpha}\bx_{k,\alpha}^{\top}}{N_v}\right)^{-1}\Bigg\Arrowvert_{op}^4\Bigg\Arrowvert\sum_{k \in U_v}\frac{(1-p(\bx_k))\bx_{k,\alpha}}{N_v}\Bigg\Arrowvert_2^4\\
        &\leq \lambda^4\Bigg\Arrowvert\left(\sum_{k \in U_v}\frac{p(\bx_k)\bx_{k,\alpha}\bx_{k,\alpha}^{\top}}{N_v}\right)^{-1}\Bigg\Arrowvert_{op}^4\Bigg\Arrowvert\sum_{k \in U_v}\frac{(1-p(\bx_k))\bx_{k,\alpha}}{N_v}\Bigg\Arrowvert_2^4.
    \end{align*}
    On the one hand, by the law of large numbers and the continuous mapping theorem, we obtain
    \begin{align*}
        \Bigg\Arrowvert\left(\sum_{k \in U_v}\frac{p(\bx_k)\bx_{k,\alpha}\bx_{k,\alpha}^{\top}}{N_v}\right)^{-1}\Bigg\Arrowvert_{op}^4&=\lambda_{min}^{-4}\left(\sum_{k \in U_v}\frac{p(\bx_k)\bx_{k,\alpha}\bx_{k,\alpha}^{\top}}{N_v} \right)\\
        &\xrightarrow[v \rightarrow \infty]{a.s.}\lambda_{min}^{-4}\left(\E_{\bx}\left[p(\bx_1)\bx_{1,\alpha}\bx_{1,\alpha}^{\top}\right] \right).
    \end{align*}
    This implies
    \begin{equation*}
        \limsup_{v \rightarrow\infty}\Bigg\Arrowvert\left(\sum_{k \in U_v}\frac{p(\bx_k)\bx_{k,\alpha}\bx_{k,\alpha}^{\top}}{N_v}\right)^{-1}\Bigg\Arrowvert_{op}^4<\infty.
    \end{equation*}
    On the other hand, 
    \begin{equation*}
        \Bigg\Arrowvert\sum_{k \in U_v}\frac{(1-p(\bx_k))\bx_{k,\alpha}}{N_v}\Bigg\Arrowvert_2^4\leq \sum_{k \in U_v}\frac{\Arrowvert\bx_{k,\alpha}\Arrowvert_2^4}{N_v}\leq C_0^4
    \end{equation*}
    almost surely.
    As a result, we conclude
    \begin{equation*}
        \limsup_{v \rightarrow \infty}\Arrowvert\widetilde{\textbf{c}}_{\alpha,v}\Arrowvert_2^4<\infty
    \end{equation*}
    almost surely.
\end{proof}
\begin{lemma}\label{FiniteEta}
Let $\alpha \in \cC$.  Assume \ref{S1}-\ref{S3} and \ref{D1}-\ref{D3}. Then,
   \begin{equation*}
       \limsup_{v \rightarrow \infty}\sum_{k \in U_v}\frac{\eta_{k,\alpha}^4}{N_v}<\infty
   \end{equation*}
   almost surely.
\end{lemma}
\begin{proof}
      Using the result of \eqref{proveBoundnessEta}, we obtain
    \begin{align*}
        \sum_{k \in U_v}\frac{\eta_{k,\alpha}^4}{N_v}&\leq \sum_{k \in U_v}\frac{8\Arrowvert \bx_{k,\alpha}\Arrowvert_2^4\Arrowvert\pmb{\beta}_{\alpha}\Arrowvert_2^4+8(8+8\pi_k^2\Arrowvert\textbf{c}_{\alpha,v}\Arrowvert_2^4\Arrowvert\bx_{k,\alpha}\Arrowvert_2^4)\epsilon_k^4}{N_v}\\
        &\leq 8C_0^4\Arrowvert\pmb{\beta}_{\alpha}\Arrowvert_2^4+64(1+C_0^4\Arrowvert\textbf{c}_{\alpha,v}\Arrowvert_2^4)\sum_{k \in U_v}\frac{\epsilon_k^4}{N_v}.
    \end{align*}
   Since
   \begin{equation*}
       \sum_{k\in U_v}\frac{\epsilon_k^4}{N_v}\xrightarrow[v\rightarrow \infty]{a.s.}\E_{\bx m}\left[\epsilon_1^4\right]\leq M_0,
   \end{equation*}
   we have
   \begin{equation*}
       \limsup_{v \rightarrow \infty}\sum_{k \in U_v}\frac{\epsilon_k^4}{N_v}<\infty
   \end{equation*}
   almost surely.
   Combining with Lemma \ref{finitetildec}, we obtain
    \begin{align*}
        \limsup_{v \rightarrow\infty}\sum_{k \in U_v}\frac{\eta_{k,\alpha}^4}{N_v}<\infty
    \end{align*}
    almost surely.
\end{proof}
\subsubsection{Finite moments of various statistics of interest}

\begin{lemma}\label{boundnessCalpha}
    Let $\alpha\in \cA$ and $\{\textbf{c}_{\alpha,v}\}_{v \in \mathbb{N}}$ be the sequence of statistics given by \eqref{cbold}. Assume \ref{S1}-\ref{S3} and \ref{D1}-\ref{D3}. For $\alpha \in \cC$, there exists a constant $K$ such that
    \begin{equation*}
        \E_{q}\left[\Arrowvert\textbf{c}_{\alpha,v}\Arrowvert_2^4 \right]\leq K
    \end{equation*}
    almost surely.
\end{lemma}
\begin{proof}
Write 
$$\textbf{c}_{\alpha,v}=\left(\sum_{k \in U_v}\frac{r_k\pi_k\bx_{k,\alpha}\bx_{k,\alpha}^{\top}}{N_v}\right)^{-1}\sum_{k \in U_v}\frac{(1-r_k)\bx_{k,\alpha}}{N_v}
 := \boldsymbol{A}_{\alpha,v}^{-1}\boldsymbol{b}_{\alpha,v}.$$
We have 
\begin{align} \label{eq:lem22-1}
    \rVert\boldsymbol{b}_{\alpha,v}\rVert_2 =  \bigg\rVert \dfrac{1}{N_v}  \sum_{k \in U_v}(1-r_k)\bx_{k,\alpha}\bigg\rVert_2 \leq \dfrac{1}{N_v}  \sum_{k \in U_v}(1-r_k)\rVert \bx_{k,\alpha}\rVert_2 \leq C_0.
\end{align}
Moreover, recalling that we condition on the high-probability event $\mathcal{E}_N$ defined in Remark \ref{req1}, we have 
\begin{align}\label{eq:lem22-2}
    \rVert\boldsymbol{A}_{\alpha,v}^{-1}\rVert_{op} =  \Bigg\rVert\left(\sum_{k \in U_v}\frac{r_k\pi_k\bx_{k,\alpha}\bx_{k,\alpha}^{\top}}{N_v}\right)^{-1}\Bigg\rVert_{op} = \lambda_{min}^{-1}\left( \sum_{k \in U_v}\frac{r_k\pi_k\bx_{k,\alpha}\bx_{k,\alpha}^{\top}}{N_v}\right) \leq \dfrac{1}{\gamma},
\end{align}
 for some constant $\gamma>0$, independent of $N_v$. 
Finally, combining \eqref{eq:lem22-1} and \eqref{eq:lem22-2}, we have, pointwise $$\rVert\textbf{c}_{\alpha,v} \rVert_2^4 \leq \rVert \boldsymbol{A}_{\alpha, v}^{-1}\rVert_{op}^4\rVert \boldsymbol{b}_{\alpha,v}\rVert_2^4\leq \dfrac{C_0^4}{\gamma^4}.$$  The result follows by integrating on both sides.

\end{proof}

\begin{lemma}\label{boundnessEta}
    Assume \ref{S1}-\ref{S3} and \ref{D1}-\ref{D3}. 
    For $\alpha \in \cC$, there exists a constant $M$ such that
    \begin{equation*}
        \E_{mq}\left[\sum_{k \in U_v}\frac{\eta_{k,\alpha}^4}{N_v} \right]\leq M
    \end{equation*}
    almost surely.
\end{lemma}
\begin{proof}
   Using $(a+b)^4\leq 8a^4+8b^4$,  we obtain
\begin{align}\label{proveBoundnessEta}
\mathbb{E}_{mq}\left[\sum_{k \in U_v}\frac{\eta_{k,\alpha}^4}{N_v}\right]&=\mathbb{E}_{ mq}\left[\sum_{k \in U_v}\frac{(\bx_{k,\alpha}^{\top}\pmb{\beta}_{\alpha}+r_k(1+\pi_k\textbf{c}_{\alpha,v}^{\top}\bx_{k,\alpha})\epsilon_k)^4}{N_v}\right]\nonumber\\
&\leq\mathbb{E}_{ mq}\left[\sum_{k \in U_v}\frac{8(\bx_{k,\alpha}^{\top}\pmb{\beta}_{\alpha})^4+8(r_k(1+\pi_k\textbf{c}_{\alpha,v}^{\top}\bx_{k,\alpha})\epsilon_k)^4}{N_v}\right]\nonumber\\
&\leq\mathbb{E}_{ mq}\left[\sum_{k \in U_v}\frac{8(\bx_{k,\alpha}^{\top}\pmb{\beta}_{\alpha})^4+8(1+\pi_k\textbf{c}_{\alpha,v}^{\top}\bx_{k,\alpha})^4\epsilon_k^4}{N_v}\right]\nonumber\\
&\leq\sum_{k \in U_v}\frac{8\left(\bx_{k,\alpha}^{\top}\pmb{\beta}_{\alpha}\right)^4+8\mathbb{E}_m\left[\epsilon_k^4\right]\mathbb{E}_q\left[\left(1+\pi_k\textbf{c}_{\alpha,v}^{\top}\bx_{k,\alpha}\right)^4\right]}{N_v}\nonumber\\
&\leq\sum_{k \in U_v}\frac{8\Arrowvert\bx_{k,\alpha}\Arrowvert_2^4\Arrowvert\pmb{\beta}_{\alpha}\Arrowvert_2^4+8\mathbb{E}_m\left[\epsilon_k^4\right]\left(8+8\pi_k^4\mathbb{E}_q\left[\Arrowvert\textbf{c}_{\alpha,v}\Arrowvert_2^4\right]\Arrowvert\bx_{k,\alpha}\Arrowvert_2^4\right)}{N_v}\nonumber\\
&\leq8\Arrowvert\pmb{\beta}_{\alpha}\Arrowvert_2^4C_0^4+64M_0+64M_0C_0^4\mathbb{E}_{ q}\left[\Arrowvert\textbf{c}_{\alpha,v}\Arrowvert_2^4\right].
\end{align}
Using Lemma \ref{boundnessCalpha}, we conclude
\begin{equation*}
        \E_{mq}\left[\sum_{k \in U_v}\frac{\eta_{k,\alpha}^4}{N_v} \right]\leq M,
    \end{equation*}
    almost surely.
\end{proof}
\begin{lemma}\label{BoundedV1V2}
    Assume \ref{S1}-\ref{S3} and \ref{D1}-\ref{D3}. For $\alpha \in \cC$, there exist constants $K_1$ and $K_2$ such that
    \begin{equation*}
        \E_{mq}\left[\left(n_v\bar{V}_{1,v}(\alpha)\right)^2\right]\leq K_1,\qquad \E_{q}\left[\left(n_v\bar{V}_{2,v}(\alpha)\right)^2\right]\leq K_2
    \end{equation*}
    almost surely.
\end{lemma}
\begin{proof}
    Write
\begin{align}\label{absV1}
n_v\left|\bar{V}_{1,v}(\alpha)\right|&=\frac{n_v}{N_v^2}\left|\sum_{k \in U_v}\sum_{l \in U_v}\frac{\eta_{k,\alpha}}{\pi_k}\frac{\eta_{l,\alpha}}{\pi_l}\Delta_{kl}\right|\nonumber\\
&\leq \frac{n_v}{N_v^2}\sum_{k \in U_v}\sum_{l \in U_v}\frac{|\eta_{k,\alpha}|}{\pi_k}\frac{|\eta_{l,\alpha}|}{\pi_l}|\Delta_{kl}|\nonumber\\
&\leq \frac{n_v}{N_v}\frac{1-\lambda}{\lambda}\sum_{k \in U_v}\frac{\eta_{k,\alpha}^2}{N_v}+\frac{n_v}{N_v}\sum_{k \in U_v}\sum_{\substack{l\in U_v\\ l\neq k}}\frac{|\eta_{k,\alpha}||\eta_{l,\alpha}|}{N_v\pi_k\pi_l}|\Delta_{kl}|\nonumber\\
&\leq \frac{n_v}{N_v}\frac{1-\lambda}{\lambda}\sum_{k \in U_v}\frac{\eta_{k,\alpha}^2}{N_v}+\frac{n_v}{N_v\lambda^2}\max_{k \neq l \in U_v}|\Delta_{kl}|\sum_{k \in U_v}\sum_{\substack{l\in U_v\\ l\neq k}}\frac{|\eta_{k,\alpha}||\eta_{l,\alpha}|}{N_v}\nonumber\\
&\leq \frac{n_v}{N_v}\frac{1-\lambda}{\lambda}\sum_{k \in U_v}\frac{\eta_{k,\alpha}^2}{N_v}+\frac{n_v}{\lambda^2}\max_{k \neq l \in U_v}|\Delta_{kl}|\sum_{k \in U_v}\frac{\eta_{k,\alpha}^2}{N_v}.
\end{align}
As a result, using Lemma \ref{boundnessEta} and integrating on both sides of \eqref{absV1} leads to 
\begin{align*}
    \E_{mq}\left[\left( n_v\bar{V}_{1,v}(\alpha)\right)^2\right]&\leq \left(\frac{n_v}{N_v}\frac{1-\lambda}{\lambda}+\frac{n_v}{\lambda^2}\max_{k \neq l \in U_v}|\Delta_{kl}| \right)^2\E_{mq}\left[\left(\sum_{k \in U_v}\frac{\eta_{k,\alpha}^2}{N_v}\right)^2 \right]\\
    &\leq \left(\frac{n_v}{N_v}\frac{1-\lambda}{\lambda}+\frac{\bar{\Delta}}{\lambda^2}\right)^2\E_{mq}\left[\sum_{k \in U_v}\frac{\eta_{k,\alpha}^4}{N_v} \right]\leq K_1
\end{align*}
almost surely. Next, write
\begin{align}\label{AbsV2}
n_v\left|\bar{V}_{2,v}(\alpha)\right|&=\frac{n_v\sigma^2}{N_v}\left|\sum_{k \in U_v}\frac{1-r_k+r_k\left(\pi_k\textbf{c}_{\alpha,v}^{\top}\bx_{k,\alpha}\right)^2}{N_v}\right|\nonumber\\&\leq \frac{n_v\sigma^2}{N_v}+\frac{n_v\sigma^2}{N_v}\Arrowvert\textbf{c}_{\alpha,v}\Arrowvert_2^2\sum_{k \in U_v}\frac{\Arrowvert\bx_{k,\alpha}\Arrowvert_2^2}{N_v}\nonumber\\
&=\frac{n_v\sigma^2}{N_v}\left(1+C_0^2\Arrowvert\textbf{c}_{\alpha,v}\Arrowvert_2^2\right)
\end{align}
Integrating on both sides of \eqref{AbsV2} leads to 
\begin{equation*}
    \mathbb{E}_q\left[\left(n_v\bar{V}_{2,v}(\alpha)\right)\right]\leq \frac{n_v^2\sigma^4}{N_v^2}\left(2+2C_0^4\E_q\left[\Arrowvert\textbf{c}_{\alpha,v}\Arrowvert_2^4\right]\right)\leq K_2
\end{equation*}
almost surely using Lemma \ref{boundnessCalpha}.

\end{proof}
\bibliographystyle{imsart-nameyear} 
\bibliography{bibliography}      






\end{document}